\documentclass[final,1p,times]{article}
\usepackage[english]{babel}
\usepackage{amsfonts,amsmath,amsxtra,amsthm,amssymb,latexsym}
\usepackage{verbatim}

\newtheorem{theorem}{Theorem}[section]
\newtheorem{corollary}[theorem]{Corollary}
\newtheorem{lemma}[theorem]{Lemma}

\newtheorem{definition}[theorem]{Definition}
\newtheorem{remark}[theorem]{Remark}
\newtheorem{example}[theorem]{Example}
\newtheorem{proposition}[theorem]{Proposition}

\numberwithin{equation}{section}

\DeclareMathOperator{\supp}{supp} \DeclareMathOperator{\im}{Im}
\DeclareMathOperator{\re}{Re} 
 \DeclareMathOperator{\dom}{dom}
\DeclareMathOperator{\ran}{ran} \DeclareMathOperator{\Ext}{Ext}
\DeclareMathOperator{\Span}{span}\DeclareMathOperator{\ess}{ess}

\DeclareMathOperator{\ac}{ac}
\DeclareMathOperator{\diag}{diag}\DeclareMathOperator{\loc}{loc}
\DeclareMathOperator{\comp}{comp}

\newcommand\R{{\mathbb{R}}}
\newcommand\C{{\mathbb{C}}}

\newcommand\N{{\mathbb{N}}}
\newcommand\Z{{\mathbb{Z}}}

\newcommand\gH{{\mathfrak{H}}}
\newcommand{\gG}{{\Gamma}}
\newcommand{\gT}{{\Theta}}

\newcommand{\gd}{{d}}
\newcommand{\gA}{{\alpha}}
\newcommand{\gB}{{\beta}}

\newcommand\cH{{\mathcal{H}}}
\newcommand\cI{{\mathcal{I}}}

\newcommand\cC{{\mathcal{C}}}
\newcommand\cN{{\mathfrak{N}}}

\newcommand\rH{{{H}}}

\newcommand\rD{{\rm{d}}}
\newcommand\I{{\rm{i}}}
\newcommand\gr{{\rm{gr}}}

\def\Ext{{\rm Ext}}

\def\mul{{\rm mul\,}}
\def\wt#1{{{\widetilde #1} }}

\begin{document}
\author{Raffaele Carlone $^1$, Mark Malamud $^2$, Andrea Posilicano $^3$
\\{\ }\\
$^1$ Dipartimento di Scienze Fisiche, Universit\`a di Napoli Federico II, Napoli, Italy\\
e-mail: {\tt carlone@na.infn.it} \\
$^2$ Institute of Applied Mathematics and Mechanics, NAS,
 Donetsk, Ukraine\\
e-mail: {\tt mmm@telenet.dn.ua}\\
$^3$ DipSAT - Sezione di Matematica, Universit\`a dell'Insubria, Como, Italy\\
 e-mail: {\tt posilicano@uninsubria.it}}
\date{}

\title{On the spectral theory of  Gesztesy--{\v S}eba realizations of\\ 1--D Dirac operators  with point  interactions on a discrete set \\[1mm]
{\em \large Dedicated with deep respect to Fritz Gesztesy  on the occasion of his 60th birthday}
}
\maketitle
\begin{abstract}

\noindent
We investigate spectral properties of  Gesztesy--{\v S}eba realizations $D_{X,\alpha}$ and $D_{X,\beta}$ of the 1-D Dirac differential expression $D$  with point interactions on
a discrete set $X=\{x_n\}_{n=1}^\infty\subset \R.$
Here  $\alpha := \{\alpha_{n}\}_{n=1}^\infty$ and $\beta :=\{\beta_{n}\}_{n=1}^\infty
\subset\R$.
The Gesztesy--{\v S}eba realizations $D_{X,\alpha}$ and $D_{X,\beta}$  
are the relativistic counterparts of the
corresponding Schr\"odinger operators $\rH_{X,\alpha}$ and $\rH_{X,\beta}$ with
$\delta$- and $\delta'$-interactions, respectively.
 We define the  minimal  operator $D_X$ as the direct sum of
the  minimal Dirac operators on the intervals $(x_{n-1}, x_n)$.
Then  using the regularization  procedure for direct sum of boundary triplets we  construct  an appropriate
boundary triplet for the maximal operator $D_X^*$ in the case $d_*(X):=\inf\{|x_i-x_j|
\,, i\not=j\} = 0$. It turns out that the boundary operators $B_{X,\alpha}$ and $B_{X,\beta}$
parameterizing the realizations $D_{X,\alpha}$ and $D_{X,\beta}$
are  Jacobi matrices. 
These matrices  substantially differ
from the ones appearing in spectral theory of Schr\"odinger operators with
point interactions.
We show that certain spectral properties of  the  operators
$D_{X,\alpha}$ and $D_{X,\beta}$ correlate with the corresponding spectral
properties of the Jacobi matrices  $B_{X,\alpha}$ and $B_{X,\beta}$, respectively.
Using this connection  we investigate spectral properties (self-adjointness, discreteness, absolutely continuous and singular spectra) of  Gesztesy--{\v S}eba realizations.
Moreover, we investigate  the non-relativistic limit
as the velocity of light $c\to\infty$. Most  of our results are new even in the case $\gd_*(X)> 0.$
   \end{abstract}



\newpage

\tableofcontents
\section{Introduction}\label{Introductionc}
It is well known that many exactly solvable models describing
complicated physical phenomena are expressed in terms of operators with point interactions
(see \cite{Alb_Ges_88, Alb_Kur_00, Exn_04, KosMal12} and comprehensive  lists of references therein). In the 1-D case, the most known
models are  the Schr\"odinger  operators $\rH_{X,\gA}$ and
$\rH_{X,\gB}$ associated with the formal differential
expressions
   \begin{equation}\label{I_01}
\ell_{X,\gA}:=-\frac{\rD^2}{\rD x^2} + \sum_{x_{n}\in X}\gA_n\delta(x-x_n),\qquad
\ell_{X,\gB}:=-\frac{\rD^2}{\rD x^2} + \sum_{x_{n}\in X}\gB_n\delta'(x-x_n),
  \end{equation}
where $\delta(\cdot)$ is a Dirac delta-function.  These operators describe $\delta$- and $\delta'$-interactions,
respectively, on a discrete set $X=\{x_n\}_{n\in I} \subset \cI=(a,b) \subseteq\R$,
and the coefficients $\gA_n,\ \gB_n\in\R$ are called the strengths of
the interaction at the point $x=x_n$.

Investigation of these
models was originated by the "Kronig--Penney model" \cite{Kro_Pen}, a simple model for a non-relativistic electron moving in a fixed crystal lattice ($X=\Z$, $\cI=\R$, $\gA_n\equiv \gA$).
For a more mathematically rigorous  approach to this model see for instance, \cite{Gro} and \cite{Ges_Hol_87}
and the monograph \cite{Alb_Ges_88}.
\par
Let
$\alpha := \{\alpha_{n}\}_{n=1}^\infty\subset\R\cup
\{+\infty\}$ and  $\beta :=\{\beta_{n}\}_{n=1}^\infty
\subset\R\cup\{+\infty\}$.
There are several ways to associate well-defined linear operators with
$\ell_{X,\gA}$ and $\ell_{X,\gB}$  (see \cite{Alb_Ges_88},  \cite{BSW95},  \cite{Min_86}).
In $L^2(\cI)$, the minimal symmetric  operators $\rH_{X,\gA}$  and
$\rH_{X,\gB}$ are naturally associated  with \eqref{I_01}.
Namely, assuming that $\cI= (a, +\infty)$ and $I=\N$ one defines
the operators $\rH_{X,\gA}^0$ and $\rH_{X,\gB}^0$
 by the differential expression  $-\frac{\rD^2}{\rD x^2}$
on the domains, respectively,
\begin{eqnarray}
\dom(\rH_{X,\gA}^0)=\left\{f\in W^{2,2}_{\comp}(\cI\setminus X): f'(a+)=0,\ \begin{array}{c}  f(x_n+)=f(x_n-)\\
 f'(x_n+)-f'(x_n-)=\gA_n f(x_n)\end{array},\ n\in\N \right\}, \label{I_04}\\
\dom(\rH_{X,\gB}^0)=\left\{f\in W^{2,2}_{\comp}(\cI\setminus X): f'(a+)=0,\ \begin{array}{c}  f'(x_n+)=f'(x_n-)\\
 f(x_n+)-f(x_n-)=\gB_n f'(x_n)\end{array},\ n\in\N \right\}. \label{I_05}
  \end{eqnarray}
Let $\rH_{X,\gA}$ and $\rH_{X,\gB}$ be the closures of $\rH^0_{X,\gA}$ and $\rH^0_{X,\gB}$, respectively.
In general, the operators $\rH_{X,\gA}$ and $\rH_{X,\gB}$ are symmetric but not automatically self-adjoint. Then one is interested in finding self-adjointness criteria and in the spectral analysis of such self-adjoint realizations.
\par

In this paper we investigate  two families of operators with point interactions which are the relativistic counterparts of $\ell_{{X,\alpha}}$ and $\ell_{{X,\beta}}$. Namely, we consider the cases where the differential expression $-\frac{\rD^2}{\rD x^2}$ in \eqref{I_01} is replaced by
the Dirac differential expression
\begin{equation}\label{1.2Intro}
D \equiv D^c := -i\,c\,\frac{d}{dx}\otimes \left(\begin{array}{cc}0 & 1 \\1 &
0\end{array}\right) + \frac{c^{2}}{2}\otimes \left(\begin{array}{cc}1
& 0 \\0 & -1\end{array}\right) \equiv \left(\begin{array}{cc}{c^{2}}/{2} & -i\,c\,\frac{d}{dx} \\-i\,c\,\frac{d}{dx} & -{c^{2}}/{2}\end{array}\right)\,.
\end{equation}
Here  $c>0$ denotes the velocity of light. 
Relativistic operators with point interactions have received a lot of attention  recently  (see e.g. \cite{Alb_Ges_88, Alb_Niz_Tar_04, AvrGro, Ben, BD, Brze_03, DavSte, Dit_Ex_Seb89, Dit_Ex_Seb90, Dit_Ex_Seb, Dit_Ex_Seb8, Dit_Ex_Seb97, GS, Grosc, Gumbs, Hou_01, Hugh, Lapidus, Sha_02, Yoshi} and references therein).
\par
 Assume  that $\cI=(a,b)$ with $-\infty < a < b\le \infty$
and $I=\N$.  Following \cite{GS} (see also \cite[Appendix J]{Alb_Ges_88}
\footnote{There are typos in the definition of $D_{X,\beta}$ given in \cite[Appendix J]{Alb_Ges_88}: in formulae (J.17) and (J.23) there should be a sign $+$ instead of $-$ .}), we define  the operators $D_{X,\alpha}$ and $D_{X,\beta}$ (realizations of $D$)
to be the closures in $L^{2}(\cI)\otimes \C^{2}$ of the operators
   \begin{equation}\label{deltaIntro}
\begin{split}
 D_{X,\alpha}^0 =   D,\quad
\dom(D_{X,\alpha}^0)=&\Big\{f\in W^{1,2}_{\comp}(\cI \backslash X)\otimes
\C^{2}: f_{1}\in AC_{\loc}(\cI),\ f_{2}\in AC_{\loc}(\cI\backslash
X);\\& f_2(a+)=0\,,\quad f_{2}(x_{n}+)-f_{2}(x_{n}-)
=-\frac{i\alpha_{n}}{c}f_{1}(x_{n}),\,\,n\in\N\Big\},
\end{split}
  \end{equation}
and
  \begin{equation}
\begin{split}\label{deltapIntro}
D_{X,\beta}^0 = D, \quad
\dom(D_{X,\beta}^0) = & \Big\{f\in W^{1,2}_{\comp}(\cI \backslash X) \otimes
\C^{2}: f_{1}\in AC_{\loc}(\cI\backslash X),\ f_{2}\in
AC_{\loc}(\cI);\\& f_2(a+)=0\,,\quad
f_{1}(x_{n}+)-f_{1}(x_{n}-) = i\beta_{n}cf_{2}(x_{n}),\,\,n\in\N\Big\},
\end{split}
\end{equation}
respectively, i.e., $D_{X,\alpha} =  \overline{ D_{X,\alpha}^0}$ and  $D_{X,\beta} =  \overline{ D_{X,\beta}^0}$.
It is easily seen that both operators  $D_{X,\alpha}$ and  $D_{X,\beta}$ are  symmetric.
The domains  of the adjoint operators $D_{X,\alpha}^*$ and $D_{X,\beta}^*$ are described explicitly: $\dom(D_{X,\alpha}^*)$ and $\dom(D_{X,\beta}^*)$ are given by formulae \eqref{deltaIntro} and \eqref{deltapIntro}, respectively,  with $ W^{1,2}(\cI \backslash X)$ in place of $ W^{1,2}_{\comp}(\cI \backslash X)$
 (see Theorem \ref{1}(i)).
The important  feature of realizations $D_{X,\alpha}$ and  $D_{X,\beta}$ is that  \emph{they are always self-adjoint, $D_{X,\alpha}= D_{X,\alpha}^*$ and  $D_{X,\beta}=D_{X,\beta}^*$,   provided that the interval $\cI$ is infinite}  (see Proposition \ref{cor_delta_carleman} and Theorem \ref{1}(ii)).

The realizations  $D_{X,\alpha}$ and  $D_{X,\beta}$   have originally been introduced by Gesztesy and {\v S}eba \cite{GS} (see also \cite[Appendix J]{Alb_Ges_88})
in the case of  $\cI = \R$ and $I=\Z$, i.e. $X=\{x_n\}_{n\in \Z}$.
%
In what follows we will call these operators {\it Gesztesy--{\v S}eba realizations}
(in short, GS-realizations). These realizations  turn out to be closely related to their non-relativistic counterparts
$\rH_{X,\gA}$ and $\rH_{X,\gB}$ associated with the
differential expression \eqref{I_01}.


Gesztesy and {\v S}eba \cite{GS}   investigated  the realizations $D_{X,\alpha}$
and $D_{X,\beta}$  in the framework of extension theory of symmetric operators and treating
the operators $D_{X,\alpha}$  and $D_{X,\beta}$  as extensions of the minimal operator
%
%
\begin{equation}\label{I_06Intro}
D_X:=  \bigoplus_{n\in \Z} D_n  \qquad
D_n = D,  \qquad
\dom(D_{n}) = W^{1,2}_0[x_{n-1},x_{n}]\otimes\C^{2}. 
\end{equation}
In fact, they assumed  in addition that $d_*(X)>0$ where
\begin{equation}
 d_{*}(X):=\inf_{n}d_n\,,\qquad
d^{*}(X):=\sup_{n}d_n \qquad\text{and}\qquad d_n:=x_{n}-x_{n-1}\,.
\end{equation}
Clearly,   $D_n$ is a  symmetric operator with deficiency
indices $n_{\pm}(D_n)=2$.
These authors also  computed the resolvent differences  $(D_{X,\alpha}-z)^{-1}-  (D_{free} - z)^{-1}$ and $(D_{X,\beta} - z)^{-1} -  (D_{free} - z)^{-1}$,  where  $D_{free}$ is the free Dirac operator $D$. In  the periodic case $(X=\mathbb Z, \ \alpha_k = \alpha_0,\  \beta_k=\beta_0,$\  $k\in\mathbb Z)$ they proved that the spectra $\sigma(D_{X,\alpha})$ and $\sigma(D_{X,\beta})$ have   a band-zone structure.

Moreover, assuming  $d_*(X)>0$ they proved the following non-relativistic limit
  \begin{equation}\label{1.23GSIntro}
s-\lim_{c\to
  +\infty}\left(D_{X,\alpha}^c - (z + {c^{2}}/{2})\right)^{-1}=
(\rH_{X, \alpha} -z)^{-1}\otimes \left(\begin{array}{cc} 1 & 0 \\0 &
0\end{array}\right)\,.
  \end{equation}
%
%

In the present paper we study  the spectral properties of  the
GS-realizations  $D_{X,\alpha}$  and $D_{X,\beta}$  for arbitrary  $\gd_*(X) \ge 0$.
Moreover, we  investigate the  GS-realizations $D_{X,\gA}(Q):= D_{X,\gA} +Q$  and $D_{X,\gB}(Q):= D_{X,\gB} + Q$  of a general Dirac operator $D + Q$ with $2\times 2$ matrix potential $Q=Q^*.$

Spectral analysis of the GS-operators
$D_{X,\gA}(Q)$ and $D_{X,\gB}(Q)$ consists (at least partially) of the following problems:
\begin{itemize}
\item [(a)]  Finding  self-adjointness criteria for $D_{X,\gA}(Q)$
and $D_{X,\gB}(Q)$.
\item [(b)]   Discreteness of the spectra  of the operators  $D_{X,\gA}(Q)$ and
$D_{X,\gB}(Q)$.
\item [(c)]  Characterization of continuous,
absolutely continuous, and singular parts of the spectra of the
operators $D_{X,\gA}(Q)$ and $D_{X,\gB}(Q)$.
\item [(d)] Resolvent comparability of the operators $D_{X,\gA^{(1)}}(Q)$ and
$D_{X,\gA^{(2)}}(Q)$ with  $\gA^{(1)}\neq \gA^{(2)}$ i.e. finding conditions for the inclusion
$(D_{X, \alpha^{(1)}}(Q) - i)^{-1} - (D_{X, \alpha^{(2)}}(Q) - i)^{-1}
\in\mathfrak S_p(\gH)$  to be valid.  Here  $\mathfrak S_p(\gH)$ denotes the Neumann-Schatten ideal.
%
%
\end{itemize}
%

We investigate spectral properties of these operators  by applying the technique of
boundary triplets and the corresponding Weyl functions (see
Section \ref{Sec_II_Prelim} for the precise definitions).
This new approach to extension theory of symmetric operators has
appeared and was intensively elaborated  during the last three decades (see \cite{Gor84, DM91, DM95, BGP07}, \cite{Sch2012}, \cite[Chapter 9]{DelAnt2011}) and references therein).

The main ingredient of this approach
is the following  abstract version of the
Green formula for the adjoint $A^*$ of a symmetric operator $A$:
  \begin{equation}\label{II.1.2_green_fIntro}
(A^*f,g)_\gH - (f,A^*g)_\gH = (\gG_1f,\gG_0g)_\cH -
(\gG_0f,\gG_1g)_\cH, \qquad f,g\in\dom(A^*).
\end{equation}
Here $\cH$ is an auxiliary Hilbert space and  the mapping $\Gamma :=\binom{\Gamma_0}{\Gamma_1}:  \dom(A^*)
\rightarrow \cH \oplus \cH$ is required to be surjective.
The mapping $\Gamma$ leads to  a natural  parametrization   of  self-adjoint
extensions of $A$ by means of self-adjoint linear relations (multi-valued operators) in $\cH$, see \cite{Gor84, DM91}.
For instance, any    extension  $\widetilde A = \widetilde A^*$  disjoint with $ A_0:=A^*\!\upharpoonright\ker(\Gamma_0)$
admits a representation
   \begin{equation}\label{5.1_Intro}
\widetilde A = A_B:=A^*\!\upharpoonright \ker\bigl(\Gamma_1 - B\Gamma_0\bigr)\qquad
\text{with}\qquad B=B^*  \in \mathcal{C}(\mathcal{H}),
     \end{equation}
where the graph of the "boundary" operator $B$  in $\cH$
is $\Gamma\dom(\widetilde A) := \{\{\Gamma_0f, \Gamma_1f\}: f \in \dom(\widetilde A)\}.$
As distinguished from the von Neumann  approach,  parametrization \eqref{5.1_Intro}  yields
a natural description  of all proper (in particular, self-adjoint) extensions
in terms of (abstract) boundary conditions.

In particular,  this approach was successfully applied
to boundary value problems  for   smooth elliptic operators  on bounded or unbounded domains
with a smooth compact boundary (see \cite{Gru68, BGW08}, \cite{Mal10} and the monograph  \cite{Gru08}),
to the maximal Sturm-Liouville operator $-d^2/dx^2+T$ in
$\mathfrak H = L^2([0,1]; \cH)$ with an unbounded operator potential $T=T^*\ge a
I$, $T\in \mathcal{C}(\mathcal{H})$ (\cite{Gor84}, see also
\cite{DM91} and \cite{MN2012} for the case of $\mathfrak H = L^2(\R_+; \cH)$),
as well as  to  3-D and 2-D Schr\"{o}dinger operators with infinitely many $\delta$-interactions
(see \cite{MalSch12} and references therein).

The most relevant to our paper is the article  \cite{KM}  where this approach was applied  to 1-D Schr\"{o}dinger operators  in  the case  $d_*(X) = 0$
(for the case $d_*(X) > 0$ see works \cite{Koc_89}, \cite{Mih_94a}).
Namely, confining ourselves  to the case of $\cI \subset \R_+$ 
we treat the  GS-operators $D_{X,\alpha}$  and $D_{X,\beta}$
as extensions of the minimal operator $D_X$ given by  \eqref{I_06Intro} with $I=\N$ in place of $I=\Z$.

%
%

A boundary triplet for the operator $A^*$ always exists whenever $n_+(A)= n_-(A)$, though it is not unique. Its
role in extension theory is similar to  that of a coordinate system in analytic
geometry. It enables us to describe all self-adjoint extensions
in terms of (abstract) boundary conditions in place of the second von Neumann formula, although this description is simple and
adequate only under a suitable choice of a boundary triplet. Note that in the case $n_\pm(A)=\infty$ a construction of a suitable
boundary triplet is a rather difficult problem.

For the adjoint operator $D_X^*$ of $D_X$ given by \eqref{I_06Intro} it is natural to search for boundary triplets
constructed as a direct sum of triplets $\Pi_n$  for
operators $D_n^*$, that is,  
$\Pi_D :=\{\cH,\Gamma_0,\Gamma_1\}:=\bigoplus_{n=1}^\infty \Pi_n$, where $\Pi_n$ is a boundary triplet for $D_n^*$, $n\in I$, and
\begin{equation}\label{I_07}
\cH :=\bigoplus_{n\in\N}\cH_n,\qquad \Gamma_0
:=\bigoplus_{n\in\N}\Gamma_0^{(n)},\qquad \Gamma_1
:=\bigoplus_{n\in\N}\Gamma_1^{(n)}.
\end{equation}
If $d_*(X)>0$, then it is easily seen  that the triplet \eqref{I_07}
is a boundary triplet for $D_X^*$ if one chooses $\Pi_n=\{\cH, \Gamma^{(n)}_0,\Gamma^{(n)}_1\}$ in the standard way with  $\Gamma^{(n)}_j,\ j\in \{0,1\},$ given by
\begin{equation}\label{triple1Intro}
\Gamma^{(n)}_0 f = -i\sqrt{\frac{c}{2}}
\binom{f_2(x_{n-1}+)-f_2(x_n -)}{f_1(x_{n-1}+)-f_1(x_n-)},\qquad
\Gamma_1^{(n)}f=
\sqrt{\frac{c}{2}}\binom{f_1(x_{n-1}+) + f_1(x_n-)}{f_2(x_{n-1}+) + f_2(x_n-)}, \qquad n\in \N,
\end{equation}
where $f=\binom{f_1}{f_2}$  (see \cite[formula (66)]{DerMal92}).
However, this direct sum is no longer a  boundary  triplet for $D_X^*$ whenever $d_*(X)=0$  (see Proposition \ref{prop3.5direcsum}). To construct a boundary  triplet we use a  regularization procedure elaborated  in \cite{MN2012} and  \cite{KM}. This procedure was already  applied  in \cite{KM} to 1--D Schr\"odinger operators with point interactions.
However,  in comparison with the Schr\"odinger case, one meets  an additional  difficulty of an algebraic character. Namely, we are searching for a boundary triplet such that the corresponding boundary operator (cf. \eqref{5.1_Intro}) is a Jacobi (tri-diagonal) matrix and for this purpose we need to construct an appropriate boundary  triplet $\Pi_n$ for the Dirac operator $D_n^*$ on the interval $[x_{n-1},x_n]$. Let us emphasize that only a sequence of boundary triplets
$\widetilde{\Pi}^{(n)} = \{\cH, \widetilde{\Gamma}_{0}^{(n)}, \widetilde{\Gamma}_{1}^{(n)}\}$ given by $\cH=\C^2,$
   \begin{equation}\label{triple2Intro}
\widetilde{\Gamma}_{0}^{(n)}f =
\left(\begin{array}{c}
                 f_{1}(x_{n-1}+)\\
                 i\,c\, f_{2}(x_{n}-)
                       \end{array}\right)\,,\qquad
\widetilde{\Gamma}_{1}^{(n)}f
=\left(\begin{array}{c}
                                                                           i\,c\,f_{2}(x_{n-1}+)\\
                                                                            f_{1}(x_{n}-)
                                                                          \end{array}\right)\,, \qquad n\in \N,
\end{equation}
(see \eqref{triple2})  leads after an appropriate regularization  to a new sequence of triplets $\Pi_n$ for $D_n^*$ having desirable properties (see  Theorem \ref{th_bt_2}). Namely, only
\emph{in the  triplet} $\Pi_D = \oplus_1^\infty \Pi_n$ given by \eqref{IV.1.1_12}, \eqref{IV.1.1_12.1},
\emph{the  parametrization of GS-realizations is given by means  of Jacobi matrices}. Let us also mention that a boundary triplet $\widetilde{\Pi}^{(n)}$ for $D_n^*$ of the form  \eqref{triple2Intro} differs from \eqref{triple1Intro} and the other ones  known in the literature (see, e.g., \cite{BraMalNei02, DerMal92}).

Recall that one of the main  results  in  \cite{KM}  states that
certain  spectral properties of $\rH_{X,\gA}$ (self-adjointness,  discreteness, etc.)
correlate with the corresponding spectral properties of the Jacobi matrix
\begin{equation}\label{IV.2.1_01AIntro}
B_{X,\gA}(\rH) := \left(%
\begin{array}{cccccc}
  0& -\gd_1^{-2} & 0 & 0& 0  &  \dots\\
  -\gd_1^{-2} & -\gd_1^{-2}& \gd_1^{-3/2}\gd_2^{-1/2}& 0 & 0&  \dots\\
  0 & \gd_1^{-3/2}\gd_2^{-1/2} & \alpha_1\gd_2^{-1} & -\gd_2^{-2} & 0&   \dots\\
  0 & 0 & -\gd_2^{-2} & -\gd_2^{-2} & \gd_2^{-3/2}\gd_3^{-1/2}&  \dots\\
  0 & 0 & 0 & \gd_2^{-3/2}\gd_3^{-1/2} & \alpha_2\gd_3^{-1}&   \dots\\
\dots& \dots&\dots&\dots&\dots&\dots\\
 \end{array}%
\right).
\end{equation}
As usual we  identify the Jacobi matrix $B_{X,\gA}(\rH)$ with (the closed) minimal symmetric operator
 associated with it and denote it by the same letter.
We emphasize that the Jacobi operator $B_{X,\gA}(\rH)$ is  a boundary operator
for  $\rH_{X,\gA}$ in the triplet $\Pi_H=\{\cH,\Gamma_0,\Gamma_1\}$ in the sense of \eqref{5.1_Intro}, that is
%
%
     \begin{equation}\label{SchredingIntro}
\rH_{X,\gA} = 
\rH_{\min}^* \upharpoonright \dom(\rH_{B_{X,\gA}(\rH)}),\quad \dom(\rH_{B_{X,\gA}}(\rH))=
\{f\in W^{2,2}(\cI\setminus X):\Gamma_1f = B_{X,\gA}(\rH)\Gamma_0f\}\,.
    \end{equation}

In the present paper we establish  similar results for GS-realizations $D_{X,\alpha}$ and $D_{X,\beta}$. For instance,
we show (see Proposition \ref{th_delta_sa} and Theorem \ref{deltadiscr}) that self-adjointness  and discreteness of the spectrum of  $D_{X,\alpha}$
correlate  with the corresponding properties of the following Jacobi matrix
   \begin{equation}\label{IV.2.1_01Intro}
B_{X,\gA}=\left(
\begin{array}{cccccc}
  0& -\frac{\nu(\gd_{1})}{\gd_1^{2}}& 0 & 0& 0  &  \dots\\
   -\frac{\nu(\gd_{1})}{\gd_1^{2}} &  -\frac{\nu(\gd_{1})}{\gd_1^{2}}&\frac{\nu(\gd_{1})}{\gd_1^{3/2}\gd_2^{1/2}}& 0 & 0&  \dots\\
  0 & \frac{\nu(\gd_{1})}{\gd_1^{3/2}\gd_2^{1/2}} & \frac{\alpha_1}{\gd_2} & -\frac{\nu(\gd_{2})}{\gd_2^{2}}& 0&   \dots\\
  0 & 0 & -\frac{\nu(\gd_{2})}{\gd_2^{2}}&  -\frac{\nu(\gd_{2})}{\gd_2^{2}}&\frac{\nu(\gd_{2})}{\gd_2^{3/2}\gd_3^{1/2}}&  \dots\\
  0 & 0 & 0 & \frac{\nu(\gd_{2})}{\gd_2^{3/2}\gd_3^{1/2}}& \frac{\alpha_2}{\gd_3}&   \dots\\
\dots& \dots&\dots&\dots&\dots&\dots\\
 \end{array}%
\right)\,,
   \end{equation}
where  $\nu(x) := \left({1+ (c^2x^2)^{-1}}\right)^{-1/2}\,.$
 Emphasize that similar to  \eqref{SchredingIntro},  the Jacobi operator $B_{X,\gA}$ in \eqref{IV.2.1_01Intro} is just
a boundary operator for the GS realization $D_{X,\gA}$ in the triplet $\Pi_D =\{\cH,\Gamma_0,\Gamma_1\}$, that is
   \begin{equation}\label{DiracDomIntro}
D_{X,\gA}=D_{B_{X,\gA}}=D_{X}^*\upharpoonright\dom(D_{B_{X,\gA}}), \quad
\dom(D_{B_{X,\gA}})=\{f\in W^{1,2}(\cI\setminus X)\otimes \C^2:
\Gamma_1f=B_{X,\gA}\Gamma_0f\}.
     \end{equation}

Representation  \eqref{DiracDomIntro} plays a crucial role in the paper: it allows us to solve the problems
(a)-(d) regarding the operator  $D_{X,\gA}$ by combining  known results on Jacobi matrices with the  technique elaborated in \cite{DM91, DM95}.

For instance,  applying the Carleman test (see e.g.  \cite{Akh}, and \cite[Chapter VII.1.2]{Ber68}) to the Jacobi matrix  $B_{X,\gA}$
we get that  $B_{X,\gA} = B_{X,\gA}^*$, and hence $D_{X,\gA}$ is always self-adjoint
whenever $\cI = \R_+.$
It is not the case for  GS-realizations $D_{X,\gA}$  on a finite interval $\cI$:
under certain conditions on the sequences $\gA$ and $X$ it might happen that either  
$D_{X,\gA}$ has the non-trivial  deficiency indices
$n_{\pm}(D_{X,\gA}) =1$ (see Theorem \ref{not-s.a.}) or $n_{\pm}(D_{X,\gA}) =0$, i.e. it is self-adjoint.
More precisely, applying the Dennis-Wall test (see e.g. \cite[Problem 2, p. 25]{Akh})  to the matrix $B_{X,\gA}$ we show in
Proposition  \ref{interval} that the GS-operator  $D_{X,\gA}$ on a finite interval $\cI =(a,b)$  is self-adjoint provided that
\begin{equation}\label{4.14Intro}
\sum_{n\in \N} \sqrt{d_{n}d_{n+1}}\,|\alpha_n|=+\infty\,.
 \end{equation}
%
%

Next, applying known results on discreteness spectra of Jacobi  matrices \cite{Chi62} to the  matrix $B_{X,\gA}$, we obtain (see Proposition  \ref{prop_IV.2.4_01})
 that the GS-operator $D_{X,\gA}^c$ on the half-line $\R_+$ has discrete
 spectrum provided that  $\lim_{n\to\infty}\gd_n=0$  and
%
%
    \begin{equation}\label{prop_chihara_1Intro}
\lim_{n\to\infty}\frac{|\gA_n|}{\gd_n}=\infty\qquad
\text{and}\qquad \lim_{n\to\infty}\frac{c}{\alpha_n} > -
\frac{1}{4}.
  \end{equation}
Note, that condition  $\lim_{n\to\infty}\gd_n=0$  is necessary for
the  minimal Dirac operator $D_X$ on $\R_+$ to have extensions
(realizations) with  discrete spectrum.  It is worth to mention that conditions  \eqref{prop_chihara_1Intro} provide the discreteness property of the GS-operators
$D_{X,\gA}(Q):= D_{X,\gA} +Q$ with  certain  unbounded potentials $Q$ (see Proposition \ref{prop_IV.2.4_01}  and Example \ref{example5.30}).

Using parametrization   \eqref{DiracDomIntro}, \eqref{IV.2.1_01Intro},   we express the
inclusion   $(D_{X, \alpha} - z)^{-1} - (D_N - z)^{-1} \in\mathfrak S_p(\gH)$
in terms of $\alpha = \{\alpha_n\}_1^\infty$ and $\{d_n\}^{\infty}_1$.
Here  $D_N$ is the Neumann  realization of $D$,
$\dom(D_N) = \{f\in W^{1,2}(\mathbb R_+)\otimes\mathbb C^2:\ f_2(+0)=0\}.$
Based on this result  we prove (see Theorem \ref{th_cont_spectr}) that
  \begin{equation}\label{4.26Intro}
\sigma_{\ess}(D_{X,\alpha}) = \sigma_{\ess}(D_{N}) = {\R}\setminus (-c^2/2, c^2/2), \quad\text{whenever}\quad \lim_{n\to\infty}\alpha_n/d_n = 0,
    \end{equation}
and
     \begin{equation}\label{4.27Intro}
\sigma_{ac}(D_{X,\alpha}) = \sigma_{\ac}(D_{N}) = \mathbb R\setminus(-c^2/2,c^2/2)\qquad  \text{if} \qquad \{\alpha_n/d_n\}^{\infty}_1\in l^1(\mathbb N).
   \end{equation}

We also find conditions guarantying that the spectrum $\sigma(D_{X,\alpha})$ is  purely singular.

Finally,  assuming that the  operators $D^c_{X, \alpha}:= D_{X, \alpha}$ and $\rH_{X, \alpha}$ are self-adjoint  and  using parameterizations \eqref{DiracDomIntro} and  \eqref{SchredingIntro} we prove  (see Theorem \ref{nonrelGS}) the  non-relativistic limit  \eqref{1.23GSIntro} in the case  $d_*(X)\ge 0.$
In particular, \eqref{1.23GSIntro} holds whenever  $\cI =\R_+$ and $\rH_{X, \alpha} = \rH_{X, \alpha}^*$. The latter happens if, for instance, $\rH_{X, \alpha}$ is lower semibounded (see \cite{AKM_10}).\par

Similar results are also valid  for the GS-realizations  $D_{X, \beta}$.
The simplest way to prove that is to extract them from the corresponding properties of
the operators $D_{X, \alpha}$. This can be done  by noticing that
$D_{X, \beta}$ is  unitarily equivalent to $-\widehat D_{X, \alpha}$ where $\alpha = \beta c^{2}$ and that the
resolvent difference $(\widehat D_{X, \alpha}+z)^{-1} - (D_{X, \alpha}+z)^{-1}$ is a rank-one operator (see Proposition \ref{prop6.5}).
\par
The paper is organized as follows.
\par
Section \ref{Sec_II_Prelim} is preparatory. It contains necessary
definitions and statements on the theory of boundary triplets of symmetric operators, Weyl functions,
$\gamma$-fields,  etc.
We also consider a family of symmetric operators $\{S_n\}_{n\in\N}$
and a family of boundary triplets $\Pi_n$ for $S_n^*$, $n\in \N$.
Following \cite{MN2012} and \cite{KM} we discuss conditions  guarantying that the  direct sum $\Pi=\oplus_{n=1}^\infty \Pi_n$ of boundary triplets $\Pi_n$ is either a $B$-generalized or an ordinary boundary triplet.
We also discuss and complete   regularization procedure for
$\Pi_n$ such  that a direct sum of regularized  boundary triplets  forms already a boundary triplet for the operator $A^* = \oplus_{n=1}^\infty S^*_n$  (see Theorem \ref{th_III.2.2_01}).

In section 3 we construct boundary triplets for maximal Dirac operators on finite intervals and half-lines and compute the corresponding Weyl functions. Using the explicit form of the Weyl functions and applying the regularization procedure described in Section 2, we construct a boundary triplet $\Pi_D$ for the maximal operator $D^*_X$.
We also describe trace properties of functions from the space $W^{1,2}(\mathbb R_+\setminus X)$ and show that the direct sum $\oplus_{n=1}^\infty \Pi_n$ is an ordinary boundary triplet if and only if $0<d^*(X)<\infty$ and it is a generalized boundary triplet (in the sense of  \cite{DM95}) whenever $d^*(X)<\infty$.

In Section 4 we apply boundary triplets technique to prove the non-relativistic limit for any $m$-dissipative  ($m$-accumulative) realization of the expression $D_X$. To this end we compute the corresponding limits of the Weyl function and $\gamma$-field.

In Section 5 we investigate spectral properties of GS-realizations $D_{X,\alpha}(Q)$ and $D_{X,\beta}(Q)$ and solve problems $(a)-(d)$. Moreover, we  show  (see Remark 5.10) that the operators $D_{X,\alpha}(Q)$ and $D_{X,\beta}(Q)$ on the line are  selfadjoint for any continuous (not necessarily bounded)  $2\times 2$ potential matrix $Q(\cdot) = Q(\cdot)^*$.
We also find certain sufficient conditions for the operator $D_{X,\alpha}$ on a finite interval either to be self-adjoint or to have deficiency indices $n_{\pm}(D_{X,\alpha}) = 1$ (see Proposition \ref{interval} and Theorem  \ref{not-s.a.}).
Comparison of  these results  shows  that roughly speaking $D_{X,\alpha}$ is  self-adjoint on a finite interval  whenever the sequence $\{\alpha_n\}_1^\infty$ grows  faster than the sequence $\{d_n\}_1^\infty$ decays.

Moreover, using parameterizations \eqref{DiracDomIntro} and  \eqref{SchredingIntro} and the general  result  on non-relativistic limits (see Theorem \ref{2.5}) we prove relation \eqref{1.23GSIntro} as well as similar relation for $D_{X,\beta}^c$.

\vskip15pt\noindent
\textbf{Notations.}
Throughout  the paper  $\mathfrak{H}$, $\cH$ denote  separable
Hilbert spaces. $[\mathfrak{H}, \cH]$ denotes the set of bounded
operators from $\mathfrak{H}$ to $\cH$;
$[\mathfrak{H}]:=[\mathfrak{H},\mathfrak{H}]$.
$\mathcal{C}(\mathfrak{H})$ and
$\widetilde{\mathcal{C}}(\mathfrak{H})$ are the sets of closed
operators and linear relations in $\mathfrak{H}$, respectively. By
${\mathfrak S}_p,\ p\in(0,\infty),$ we denote the
Neumann-Schatten ideals. Let $T$ be a linear operator in a Hilbert
space $\mathfrak{H}$. In what follows, $\dom (T)$, $\ker (T)$,
$\ran (T)$ are the domain, the kernel, the range of $T$,
respectively; $\sigma(T)$, $\sigma_p(T)$, $\sigma_c(T)$, $\sigma_{ac}(T)$, and
$\sigma_s(T)$,  denote the spectrum, point
spectrum, continuous, absolutely continuous and singular spectrum of $T=T^*$, respectively;
$\rho(T)$ and $\widehat{\rho}(T)$ denote  the resolvent set, and the set of
regular type points of $T$, respectively; $R_T \left(\lambda
\right):=\left( T-\lambda I\right)^{-1} $, $\lambda \in \rho(T)$,
is the resolvent of $T$.

\noindent
Let $X$ be a discrete subset of $\cI\subseteq\R$. We define the
following Sobolev spaces
\begin{gather*}
W^{1,2}(\cI\setminus X):=\{f\in L^2(\cI): f\in AC_{loc}(\cI\setminus X), f'\in
  L^2(\cI)\},\\
W^{2,2}(\cI\setminus X):=\{f\in L^2(\cI): f,
f'\in AC_{loc}(\cI\setminus X), f''\in
  L^2(\cI)\},\\
W^{1,2}_0(\cI\setminus X):=\{f\in
W^{1,2}(\cI): f(x_k)=0,\,
\mbox{for all }  x_k\in X\},\\
W^{2,2}_0(\cI\setminus X):=\{f\in W^{2,2}(\cI): f(x_k)=f'(x_k)=0\
\,
\mbox{for all }  x_k\in X\}\,.\\
W^{k,2}_{\comp}(\cI\setminus X):=\{f\in W^{k,2}(\cI\setminus X):
\supp f\ \text{is compact in}\ \cI\} = W^{k,2}(\cI\setminus X)\cap
L^2_{\comp}(\cI).
\end{gather*}
Let $I$ be a subset of $\Z$, $I\subseteq\Z$. For any non-negative sequence $\{c_n\}_{n\in I}$
we denote by $l^2(I;\{c_n\},\cH) := l^2(I;\{c_n\})\otimes \cH$ the weighted Hilbert space of
$\cH$-valued sequences, i.e. $f= \{f_n\}_{n\in I}\in l^2(I;\{c_n\},\cH)$
if $\|f\|^2 = \sum_{n\in I}c_n\|f_n\|_{\cH}^2<\infty$; $l^2_0(I,\cH)$
is a subset of finite sequences in $l^2(I;\{c_n\},\cH),$ i.e. the sequences
with compact supports;  we also abbreviate  $l^2(\N;\{c_n\}):=l^2(\N;\{c_n\},\C)$,
$l^2_0(\N;\{c_n\}) := l^2_0(\N;\{c_n\},\C)$. As usual $l^{p}(\N)$, $p\in [1,\infty),$ denotes the space of $p$-summable complex-valued sequences $f=\{f_{n}\}_{n\in\N}$;
$l^{\infty}(\N)$ denotes  the space of bounded complex-valued sequences and $c_{0}(\N)$ is a subspace of $l^{\infty}(\N)$ consisting of sequences $f=\{f_{n}\}_{n\in\N}$ satisfying $\lim_{n\to\infty}f_{n} = 0$. $\chi_{\pm}(\cdot)$ denotes the indicator
function of $\R_{\pm}.$

\section{Preliminaries}\label{Sec_II_Prelim}

\subsection{Boundary triplets and Weyl functions}\label{Subsec_II.1_btrips}

In this section we briefly recall  the basic facts  of the theory of boundary
triplets and the corresponding  Weyl functions  (we refer to  \cite{DM91, DM95, Gor84} for a detailed exposition of boundary triplets).
Besides, we discuss a regularization procedure for direct sum of boundary triplets following
\cite{KM}  and \cite{MN2012}. Moreover, we slightly complete \cite[Theorem 3.13]{KM}
(see Theorem  \ref{th_III.2.2_01}).

\subsubsection{Linear relations, boundary triplets, and self-adjoint extensions}
\label{sss_II.1.1_lr}

\textbf{1.} \ The set $\widetilde\cC(\cH)$ of closed linear relations in $\cH$ is the set of closed linear subspaces of $\cH\oplus\cH$.
Recall that $\dom(\Theta) =\bigl\{
f:\{f,f'\}\in\Theta\bigr\} $, $\ran(\Theta) =\bigl\{
f^\prime:\{f,f'\}\in\Theta\bigr\} $, and $\mul(\Theta) =\bigl\{
f^\prime:\{0,f'\}\in\Theta\bigr\} $ are the domain, the range, and the multivalued part of $\Theta$. A closed linear operator $A$ in $\cH$ is identified
with its graph $\gr(A)$, so that the set  $\cC(\cH)$  of closed linear
operators in $\cH$ is viewed as a subset of $\widetilde\cC(\cH)$.
In particular, a linear relation $\Theta$ is an operator if and
only if 
$\mul(\Theta)$ is trivial.
We recall that the adjoint relation
$\Theta^*\in\widetilde\cC(\cH)$ of $\Theta\in \widetilde\cC(\cH)$ is defined by
\begin{equation*}
\Theta^*= \left\{
\begin{pmatrix} h\\h^\prime
\end{pmatrix}: (f^\prime,h)_{\cH}=(f,h^\prime)_{\cH}\,\,\text{for all}\,
\begin{pmatrix} f\\f^\prime\end{pmatrix}
\in\Theta\right\}.
\end{equation*}
A linear relation $\Theta$ is said to be {\it symmetric} if
$\Theta\subset\Theta^*$ and self-adjoint if $\Theta=\Theta^*$.

For a symmetric linear relation $\Theta\subseteq\Theta^*$ in $\cH$
the multivalued part $\mul(\Theta)$ is the orthogonal complement
of $\dom(\Theta)$ in $\cH$. Therefore  setting $\cH_{\rm
op}:=\overline{\dom(\Theta)}$ and $\cH_\infty=\mul(\Theta)$, one
arrives at the orthogonal decomposition  $\Theta= \Theta_{\rm
op}\oplus \Theta_\infty$ where  $\Theta_{\rm op}$ is a  symmetric
operator in $\cH_{\rm op}$, the operator part of $\Theta,$ and
$\Theta_\infty=\bigl\{\bigl(\begin{smallmatrix} 0 \\ f'
\end{smallmatrix}\bigr):f'\in\mul(\Theta)\bigr\}$,  a ``pure'' linear relation
 in $\cH_\infty$. 

\textbf{2.} \ Let $A$ be a densely defined closed symmetric
operator in a separable Hilbert space $\gH$ with equal deficiency
indices $\mathrm{n}_\pm(A)=\dim \cN_{\pm \I} \leq \infty,$ where
$\cN_z:=\ker(A^*-z)$ is the defect subspace.

\begin{definition}[\cite{Gor84}]\label{def_ordinary_bt}
A triplet $\Pi=\{\cH,\gG_0,\gG_1\}$ is called an {\rm (ordinary)
boundary triplet} for the adjoint operator $A^*$ if $\cH$ is an auxiliary
Hilbert space and $\Gamma_0,\Gamma_1:\  \dom(A^*)\rightarrow \cH$
are linear mappings such that the second abstract Green identity
\begin{equation}\label{II.1.2_green_f}
(A^*f,g)_\gH - (f,A^*g)_\gH = (\gG_1f,\gG_0g)_\cH -
(\gG_0f,\gG_1g)_\cH, \qquad f,g\in\dom(A^*),
\end{equation}
holds
and the mapping $\gG:=\begin{pmatrix}\Gamma_0\\\Gamma_1\end{pmatrix}:  \dom(A^*)
\rightarrow \cH \oplus \cH$ is surjective.
\end{definition}
First,  note that a boundary triplet
for $A^*$ exists whenever  the deficiency indices of $A$ are
equal, $\mathrm{n}_+(A)= \mathrm{n}_-(A)$.  Moreover,
$\mathrm{n}_\pm(A) = \dim \cH$ and
$\ker(\Gamma) = \ker(\Gamma_0) \cap \ker(\Gamma_1)= \dom(A)$. Note
also that $\Gamma$ is a bounded mapping from $\mathfrak H_+ = \dom(A^*)$ equipped
with the graph norm to $\cH\oplus\cH.$

A boundary triplet for $A^*$ is not unique. Moreover, for any
self-adjoint extension $\wt A := \wt A^*$ of $A$ there exists a
boundary triplet $\Pi=\{\cH,\gG_0,\gG_1\}$ such that
$\ker(\Gamma_0) = \dom(\wt A)$.
  \begin{definition}
\item $(i)$  A closed extension $A'$ of $A$ is called a \emph{proper
extension}, if $A\subset A' \subset A^*$.  The set of all proper
extensions of  $A$ completed by the (non-proper) extensions $A$ and $A^*$
is  denoted  by $\Ext_A$.

\item $(ii)$  Two proper extensions $A', A'',$ of $A$ are called disjoint
if $\dom( A')\cap \dom( A'') = \dom( A)$ and  transversal if in
addition $\dom( A') + \dom( A'') = \dom( A^*).$
   \end{definition}

Recall that an operator $T\in \cC(\cH)$ is called dissipative if $\im(Tf,f) \ge 0$ for $f\in \dom(T)$.  It is called $m$-dissipative if it has no proper dissipative extensions.
It is known (and easily seen) that dissipative $T$ is  $m$-dissipative if and only if
$\C_-\subset \rho(T)$.

The operator $T$ is called accumulative ($m$-accumulative)  if $-T$ is
dissipative  ($m$-dissipative).

Any dissipative (accumulative) extension $\wt A$ of $A$ is necessarily a proper extension,
$\wt A\in \Ext_A$. Moreover, if $A'$ and $A''$ are disjoint and
selfadjoint, then  $\dom( A') + \dom( A'')$ is dense in $\dom(
A^*).$

Fixing a boundary triplet $\Pi$ one can parameterize the set $\Ext_A$ in the following way.
\begin{proposition}[\cite{DM95}]\label{prop_II.1.2_01}
Let $A$ be as above and let $\Pi=\{\cH,\gG_0,\gG_1\}$ be a boundary
triplet for $A^*$. Then the mapping
     \begin{equation}\label{II.1.2_01A}
\Ext_A\ni \widetilde A \to  \Gamma \dom(\widetilde A)
=\{\{\Gamma_0 f,\Gamma_1f \} : \  f\in \dom(\widetilde A) \} =:
\Theta \in \widetilde\cC(\cH)
     \end{equation}
establishes  a bijective correspondence between the sets $\Ext_A$
and  $\widetilde\cC(\cH)$. We put $A_\Theta :=\widetilde A$ where
$\Theta$ is defined by \eqref{II.1.2_01A}, i.e.  $A_\Theta:=
A^*\upharpoonright \Gamma^{-1}\Theta=A^*\upharpoonright
\bigl\{f\in\dom(A^*): \ \{\Gamma_0f,\Gamma_1f\} \in\Theta\bigr\}.$
Then:

\item $(i)$ $A_\Theta$ is $m$-dissipative ($m$-accumulative) if
and only if so is $\Theta$.

\item $(ii)$  $A_\Theta$ is symmetric
(self--adjoint) if and only if so is $\Theta$.
Moreover,  $\mathrm{n}_\pm(A_\Theta)=\mathrm{n}_\pm(\Theta)$.

\item $(iii)$  The extensions $A_\Theta$ and $A_0$ are disjoint
(transversal) if and only if $\Theta$  is an operator. 
In this case $A_\Theta$ admits a representation
  \begin{equation}\label{II.1.2_01AB}
A_\Theta = A^*\!\upharpoonright\ker(\gG_1 - \Theta\gG_0).
  \end{equation}
Moreover, the extensions $A_\Theta$ and $A_0$ are transversal if and only if $\Theta \in [\cH]$.
      \end{proposition}
The linear relation $\gT$ (the operator $B$) in the correspondence
\eqref{II.1.2_01A} (resp.  \eqref{II.1.2_01AB}) is  called
\emph{the boundary relation (the boundary operator)}.

We emphasize that \emph{in the case of differential operators
opposed to the von Neumann  parametrization  the parametrization
\eqref{II.1.2_01A}--\eqref{II.1.2_01AB} describes the set of
proper extensions directly in terms of boundary conditions}.

It follows immediately from Proposition \ref{prop_II.1.2_01} that
the extensions
$$
A_0:=A^*\!\upharpoonright\ker(\gG_0)\quad \text{and}\quad
A_1:=A^*\!\upharpoonright\ker(\gG_1)
$$
are self-adjoint. Clearly, $A_j=A_{\Theta_j}, \ j\in \{0,1\},$
where the subspaces $\gT_0:= \{0\} \times \cH$ and $\gT_1 := \cH
\times \{0\}$ are self-adjoint relations in $\cH$. Note that
$\gT_0$ is a "pure" linear relation.

\subsubsection{Weyl functions, $\gamma$-fields, and Krein type formula for resolvents}\label{sss_II.1.3_weylsec}

\textbf{1.} \
In \cite{DM91, DM95}  the concept of the classical Weyl--Titchmarsh $m$-function
from the theory of Sturm-Liouville operators  was generalized to
the case of symmetric operators  with equal deficiency indices.
The role of abstract Weyl functions in the extension theory  is
similar to that of the classical Weyl--Titchmarsh $m$-function in the spectral
theory of singular Sturm-Liouville operators.

\begin{definition}[{\cite{DM91}}]\label{def_Weylfunc}

Let $A$ be a densely defined closed symmetric operator in $\gH$
with equal deficiency indices and let $\Pi=\{\cH,\gG_0,\gG_1\}$ be
a boundary triplet for $A^*$.
The operator valued functions $\gamma(\cdot) :\rho(A_0)\rightarrow  [\cH,\gH]$ and
$M(\cdot):\rho(A_0)\rightarrow  [\cH]$ defined by
  \begin{equation}\label{II.1.3_01}
\gamma(z):=\bigl(\Gamma_0\!\upharpoonright\cN_z\bigr)^{-1}
\qquad\text{and}\qquad M(z):=\Gamma_1\gamma(z), \qquad
z\in\rho(A_0),
  \end{equation}
are called the {\em $\gamma$-field} and the {\em Weyl function},
respectively, corresponding to the boundary triplet $\Pi.$
\end{definition}

The $\gamma$-field $\gamma(\cdot)$ and the Weyl function
$M(\cdot)$ in \eqref{II.1.3_01}
are well defined.
Moreover, both $\gamma(\cdot)$ and $M(\cdot)$ are holomorphic on
$\rho(A_0)$ and the following relations hold (see \cite{DM91})
    \begin{gather}\label{II.1.3_02'}
\gamma(z)=\bigl(I+(z-\zeta)(A_0-z)^{-1}\bigr)\gamma(\zeta)\qquad z,\ \zeta\in\rho(A_0),  \\
\label{II.1.3_02}
M(z)-M(\zeta)^*=(z-\overline\zeta)\gamma(\zeta)^*\gamma(z),  \qquad z,\ \zeta\in\rho(A_0).
      \end{gather}
Identities  \eqref{II.1.3_02'} and  \eqref{II.1.3_02}  mean that
$\gamma(\cdot)$ and $M(\cdot)$  are  the $\gamma$-field and the
$Q$-function of the operator $A_0,$ respectively, in the sense of
M. Krein (see \cite{KL71}). It follows from  \eqref{II.1.3_02}
that $M(\cdot)$ is an $R[\cH]$-function (or {\it Nevanlinna
function}), i.e., $M(\cdot)$ is an ($[\cH]$-valued) holomorphic
function on $\C\setminus \R$ satisfying
     \begin{equation}\label{II.1.3_03}
 \im z\cdot\im M(z)\geq 0,\qquad  M(z^*)=M(\overline
z),\quad  z\in \C\setminus \R.
\end{equation}
Moreover,   due to \eqref{II.1.3_02}  $M(\cdot)\in R^u[\cH],$ i.e.
it satisfies $0\in \rho(\im M(i)).$

It is well known  that  $M(\cdot)$ admits an integral
representation (see, for instance, \cite{Akh}, \cite{Akh_Glz})
   \begin{equation}\label{WF_intrepr}
M(z)=C_0+\int_{\R}\left(\frac{1}{t-z}-\frac{t}{1+t^2}\right)d\Sigma_M(t),\qquad
z\in\rho(A_0),
  \end{equation}
where $\Sigma_M(\cdot)$ is an operator-valued Borel measure on
$\R$ satisfying $\int_\R \frac{1}{1 + t^2}d\Sigma_M(t) \in [\cH]$
and $C_0 = C_0^*\in [\cH]$. The integral in (\ref{WF_intrepr}) is
understood in the strong sense. Note that the spectral measure $E_{A_0}(\cdot)$ of the extension $A_0 = A_0^*$ and the measure  $\Sigma_M(\cdot)$ from the integral representation \eqref{WF_intrepr} are equivalent (see \cite{BraMalNei02}). Moreover, these operator measures are spectrally equivalent in the sense of \cite{MM03}.
Note also that a  linear term $C_1z$ is missing in \eqref{WF_intrepr}
since $A$ is densely defined (see \cite{DM91}).

\textbf{2.} \  Recall that a symmetric operator $A$ in $\mathfrak H$ is said to be {\it simple} if
there is no non-trivial subspace which reduces it to a
self-adjoint operator. In other words,  $A$ is simple if it does not admit an (orthogonal) decomposition $A=A'\oplus S$ where $A'$ is a symmetric operator and $S$ is a selfadjoint operator acting on a nontrivial Hilbert space.

It is easily seen (and well-known) that $A$
is simple if and only if  $\Span\{\mathfrak N_z(A): z\in\C\setminus\mathbb R \}=\mathfrak H$.

If $A$ is  simple, then  the Weyl function $M(\cdot)$ determines
the boundary triplet  $\Pi$ uniquely up to the unitary equivalence  (see \cite{DM91}).
In particular, $M(\cdot)$ contains the full information about the
spectral properties of $A_0$.
Moreover, the spectrum  of a  proper  (not
necessarily self-adjoint) extension $A_\Theta\in \Ext_A$ can be
described  by means of $M(\cdot)$ and  the   boundary relation $\Theta$.
   \begin{proposition}[\cite{DM91}]\label{prop_II.1.4_spectrum}
Let $\Pi=\{\cH,\gG_0,\gG_1\}$ be
a boundary triplet for $A^*$ and let  $M(\cdot)$  and $\gamma(\cdot)$ be the corresponding
Weyl function and the  $\gamma$-field. Then for any $\widetilde A = A_\Theta \in \Ext_A$
with $\rho(A_\Theta)\not = \emptyset$   the following Krein type formula holds
   \begin{equation}\label{II.1.4_01}
(A_\Theta - z)^{-1} - (A_0 - z)^{-1} = \gamma(z) (\Theta -
M(z))^{-1} \gamma^*({\overline z}), \quad z\in \rho(A_0)\cap
\rho(A_\Theta).
  \end{equation}
Moreover, if $\widetilde A$  is simple, then for any  $z\in \rho(A_0)$  the following equivalence
holds
\[
z\in\sigma_i(A_\Theta) \quad \Longleftrightarrow\quad 0\in
\sigma_i(\Theta-M(z)),\qquad i\in\{\rm p,\ c,\ r\}.
\]
\end{proposition}

Formula \eqref{II.1.4_01} is a generalization of the  classical  Krein
formula for canonical resolvents (cf. \cite{Akh_Glz}, \cite{KL71}). It
establishes  a one-to-one correspondence between the set of proper
extensions $\wt A = A_\gT$ with non-empty resolvent set and the set
of the corresponding linear relations $\Theta$ in $\cH$.
Note also that all objects in \eqref{II.1.4_01} are expressed
in terms of the boundary triplet $\Pi$ (see formulae
\eqref{II.1.2_01AB} and \eqref{II.1.3_01}) (cf. \cite{DM91,DM95}).

We emphasize that \emph{precisely two parameterizations
\eqref{II.1.2_01A}--\eqref{II.1.2_01AB} and \eqref{II.1.4_01} of the set
$\Ext_A$ make it possible application of Krein's type formula for resolvents to boundary
value problems}.

The following result is deduced from \eqref{II.1.4_01}
    \begin{proposition}[\cite{DM91}]\label{prop_II.1.4_02}
Let $\Pi=\{\cH,\gG_0,\gG_1\}$  be a boundary triplet for $A^*,$
$\Theta_1,\Theta_2\in \wt\cC(\cH)$ and  let  ${\mathfrak S}_p(\gH),\
p\in(0,\infty],$ be the Neumann-Schatten ideal in $[\gH]$. Then
\item  $(i)$ For any $z
\in\rho(A_{\Theta_1})\cap\rho(A_{\Theta_2})$ and
$\zeta\in\rho(\Theta_1)\cap\rho(\Theta_2)$ the following
equivalence holds
\begin{equation}\label{II.1.4_02}
(A_{\Theta_1}-z)^{-1} - (A_{\Theta_2}-z)^{-1}\in{\mathfrak
S}_p(\gH)
\quad \Longleftrightarrow \quad (\Theta_1 - \zeta)^{-1}-
(\Theta_2 - \zeta )^{-1}\in{\mathfrak S}_p(\cH).
\end{equation}
\noindent In particular, $(A_{{\Theta}_1} - z)^{-1} - (A_0 -
z)^{-1} \in {\mathfrak S}_p(\gH) \Longleftrightarrow
\bigl(\Theta_1 - \zeta\bigr)^{-1} \in {\mathfrak S}_p(\cH).$
%
%
\item $(ii)$\  If, in addition, $\Theta_1, \Theta_2\in\cC(\cH)$
and $\dom(\Theta_1) = \dom(\Theta_2)$, then the  following implication
holds
      \begin{equation}\label{II.1.4_03}
\overline{\Theta_1 - \Theta_2} \in{\mathfrak S}_p(\cH)
\Longrightarrow (A_{\Theta_1}-z)^{-1} -
(A_{\Theta_2}-z)^{-1}\in{\mathfrak S}_p(\gH).
      \end{equation}
\item $(iii)$\  Moreover, if  $\Theta_1, \Theta_2\in[\cH]$, then
implication \eqref{II.1.4_03} becomes the equivalence.
\end{proposition}
%
%

\subsubsection{Generalized boundary triplets of bounded type}\label{sss_II.1.4 GBT}

In many applications the notion of a boundary triplet is too
restrictive because of the assumption $\dom(\Gamma_j) = \mathfrak H_+$, $j\in \{0,1\}$.
Inspiring by possible applications as well as certain  theoretical reasons
this concept was   relaxed  in \cite[Section 6]{DM95}.
   \begin{definition}[{\cite{DM95}}]\label{def_II.2.1_generalized_bt}
Let $A$ be a closed densely defined symmetric operator in $\gH$
with equal deficiency indices. Let $A_* \supseteq A$ be a not
necessarily closed extension of $A$ such that $(A_*)^* = A$. A
triplet $\Pi=\{\cH,\gG_0,\gG_1\}$  is called a
\emph{generalized boundary triplet of bounded type} (in short, \emph{$B$-generalized boundary triplet}) for $A^*$   if $\cH$ is a Hilbert space and
$\Gamma_j: \dom(\Gamma):= \dom(\Gamma_0)\cap
\dom(\Gamma_1)=\dom(A_*)\to\cH$, $j\in\{0,1\},$ are linear
mappings such that

\item $(B1)$ $\gG_0$ is surjective,

 \item $(B2)$
$A_{*0}:=A_*\upharpoonright\ker(\gG_0)$ is a self-adjoint
operator,

\item $(B3)$ the Green's identity  holds
   \begin{equation}\label{II.2.1_03}
(A_*f, g)_{\gH} - (f, A_*g)_{\gH} = (\gG_1f, \gG_0 g)_{\cH} -
(\gG_0 f, \gG_1 g)_{\cH},\qquad f,g\in\dom(A_*)=\dom(\Gamma).
       \end{equation}
       \end{definition}
Note that one always has $A \subseteq A_* \subseteq A^* =
\overline{A_*}$.
%

For any $B$-generalized boundary triplet $\Pi=\{\cH,\gG_0,\gG_1\}$
we set $A_{*j}:=A^*\lceil\ker(\gG_j)$, $j \in\{ 0,1\}$. Note  that
the extensions $A_{*0}$ and $A_{*1}$ are always  disjoint but not
necessarily transversal.

Starting with Definition \ref{def_II.2.1_generalized_bt} of a
$B$-generalized boundary triplet $\Pi$, one can  introduce
concepts of the (generalized) $\gamma$-field $\gamma(\cdot)$ and
the Weyl function $M(\cdot)$ corresponding to $\Pi$ in much the
same way as  in Definition  \ref{def_Weylfunc} for
(ordinary) boundary triplet (for detail see \cite{DM95}). Let us
mention only the following result  (cf. \cite[Proposition
6.2]{DM95} and \cite[Proposition 5.9]{DHMS06}).
     \begin{proposition}\label{prop2.10}
Let $\Pi=\{\cH,\gG_0,\gG_1\}$ be a $B$-generalized boundary
triplet for   $A^*$, $A_{*}=A^*\lceil \dom(\Gamma)$, and let
$M(\cdot)$ be the corresponding Weyl function. Then:

\item $(i)$\  $M(\cdot)$ is an $[\cH]$-valued Nevanlinna function
satisfying $\ker(\im M(z))= \{0\}, \ z\in \C_+$.

\item  $(ii)$\  $\Pi$ is an ordinary boundary triplet  if and only
if $0\in \rho(\im M(i))$.

\item  $(iii)$\  Moreover, if $\Pi=\{\cH,\gG_0,\gG_1\}$ is a
generalized boundary triplet for $A^*$ (a boundary relation in the
sense of \cite{DHMS06}) and $M(\cdot)$ is an $R[\cH]$-function
satisfying $\ker(\im M(i))= \{0\}$, then $\Pi=\{\cH,\gG_0,\gG_1\}$
is a $B$-generalized boundary triplet for $A^*$, i.e. $(B1)$ and
$(B2)$ are satisfied.
    \end{proposition}

\subsection{Direct sums of boundary triplets}\label{subsec_III.1_dirsum}

Let $S_n$ be a densely defined symmetric operator in a Hilbert
space $\mathfrak{H}_n$ with  
$\mathrm{n}_+(S_n) = \mathrm{n}_-(S_n) \le \infty,\ n\in \N.$
Consider the operator $A:= \bigoplus^{\infty}_{n=1}S_n$ acting in
$\mathfrak{H} := \bigoplus_{n=1}^{\infty}\mathfrak{H}_n$, the Hilbert
direct sum of Hilbert  spaces $\mathfrak{H}_n.$ By definition, $
\mathfrak{H} = \{f= \oplus^{\infty}_{n=1} f_n: \ f_n\in
\mathfrak{H}_n, \ \sum^{\infty}_{n=1}\|
 f_n\|^2<\infty\}.$ Clearly,
   \begin{equation}\label{III.1_01}
A^* = \bigoplus^{\infty}_{n=1}S^*_n,\qquad \dom(A^*) =  \{f =
\oplus^{\infty}_{n=1} f_n\in \mathfrak{H}:\
 f_n\in\dom(S^*_n),\ \  \sum_{n\in \N}\|S^*_n
 f_n\|^2<\infty\}.
   \end{equation}
We equip   the domains $\dom(S^*_n)=: \gH_{n+}$ and
$\dom(A^*)=: \gH_+$ with the graph norms  $\|f_n\|^2_{\gH_{n+}} :=
\|f_n\|^2 + \|S^*_nf_n\|^2$ and  $\|f\|^2_{\gH_+}  := \|f\|^2 + \|A^*f\|^2= \sum_n
\|f_n\|^2_{\gH_{n+}}$, respectively.

Further, let   $\Pi_n=\{\cH_n, \Gamma^{(n)}_0, \Gamma^{(n)}_1\}$
be a boundary triplet  for $S^*_n$, $n\in \N$.
By $\|\Gamma_j^{(n)}\|$ we denote  the norm of the linear mapping $\Gamma^{(n)}_j\in[\gH_{n+},\cH_n]$, $
j\in\{0,1\}$,  $n\in \N$.
Let also $\cH :=\bigoplus_{n=1}^{\infty}\cH_n$ be a Hilbert direct sum of
$\cH_n$. Define  mappings $\Gamma_0$ and $\Gamma_1$
by setting
   \begin{equation}\label{III.1_02}
\Gamma_j  := \bigoplus_{n=1}^{\infty} \Gamma^{(n)}_j,\qquad
\dom(\Gamma_j) = \bigl\{f = \oplus^{\infty}_{n=1} f_n
\in\dom(A^*): \ \sum_{n\in \N}\|\Gamma^{(n)}_j f_n\|^2_{\cH_n}
<\infty\bigr\}.
   \end{equation}
\ \ \  Clearly,  $\dom(\Gamma) :=\dom(\Gamma_1)\cap\dom(\Gamma_0)$
is dense in $\gH_+$.
Define the operators $S_{nj} := S_{n}^*\upharpoonright \ker
\Gamma^{(n)}_j$ and ${A}_j := \bigoplus^{\infty}_{n=1}S_{nj}$, $j\in
\{0,1\}$. Then ${A}_0$ and ${ A}_1$ are self-adjoint extensions of
$A$. Note that ${A}_0$ and ${A}_1$ are disjoint but not
necessarily transversal.
Finally, we  set
\begin{equation}\label{III.1_03}
A_* = A^*\upharpoonright\dom(\Gamma)\quad \text{and}\quad A_{*j} :=
A_*\upharpoonright\ker(\Gamma_j),\quad  j\in\{0,1\}.
   \end{equation}
Clearly, $A_{*j}$ is symmetric (not necessarily self-adjoint or
even closed!) extension of $A$, $A_{*j} \subset {A}_j$, $j\in
\{0,1\}$,  and
\[
\dom (A_{*j})= \{f =
\oplus^{\infty}_{n=1} f_n\in \mathfrak{H}:\
 f_n\in \ker \Gamma^{(n)}_j,\ \  \sum_{n\in \N}\bigl(\|S^*_n
 f_n\|^2 + \|\Gamma_{j'}^{(n)}f_n\|^2\bigr) <\infty\},\quad (0':=1,\ 1':=0).
\]

       \begin{definition}\label{def_III.1_01}
Let  $\Gamma_j$ be defined by \eqref{III.1_02} and $\cH
=\bigoplus_{n=1}^{\infty}\cH_n$. A collection $\Pi=\{\cH, \Gamma_0,
\Gamma_1\}$ will be called  a direct sum of boundary triplets and
will be assigned as $\Pi:= \bigoplus^{\infty}_{n=1}\Pi_n$.
    \end{definition}

It easily follows from  \eqref{III.1_01}--\eqref{III.1_03} and
Definition \ref{def_III.1_01},  that for
 $f=\oplus^{\infty}_{n=1} f_n,$\ $g =\oplus^{\infty}_{n=1}
g_n\in\dom(A_*) = \dom (\Gamma)$ Green's identity
\eqref{II.2.1_03} holds
   \begin{gather}\label{III.1_06}
(A_*f, g)_\gH- (f, A_* g)_\gH = \sum_{n\in \N}[(S^*_n f_n, g_n)_{\gH_n}- (f_n,S^*_n g_n)_{\gH_n}] \notag \\
= \sum_{n\in \N}\left[(\Gamma^{(n)}_1
f_n,\Gamma^{(n)}_0g_n)_{\cH_n}
 - (\Gamma^{(n)}_0 f_n,\Gamma^{(n)}_1 g_n)_{\cH_n}\right] = (\Gamma_1
f,\Gamma_0 g)_{\cH}-(\Gamma_0 f,\Gamma_1 g)_{\cH}.
   \end{gather}

The series in the above equality converge due to \eqref{III.1_01}
and \eqref{III.1_02}. However the direct sum
$\Pi=\bigoplus_{n=1}^\infty\Pi_n$ \emph{is not a boundary triplet and
even a $B$-generalized  boundary triplet for $A^*$ without
additional restrictions}. This fact was discovered  in \cite{Koc_79} (in this connection see also simple examples in  \cite{KM, MN2012} and Proposition \ref{prop3.5direcsum} below). At the same
time, according to \cite[Theorem 3.2]{KM}
$\Pi=\bigoplus_{n=1}^\infty\Pi_n$  is always a boundary relation in
the sense of \cite{DHMS06}.

The following criterion has been obtained in \cite{MN2012}, \cite{KM}.
    \begin{theorem}\label{th_criterion(bt)}
Let  $\Pi_n=\{\cH_n, \Gamma_0^{(n)}, \Gamma_1^{(n)} \}$ be a
boundary triplet for $S_{n}^*$ and  $M_n(\cdot)$  the
corresponding Weyl function, $n\in \N $. A direct sum
$\Pi=\bigoplus_{n=1}^{\infty}\Pi_n$ forms  an ordinary boundary
triplet for the operator $A^* =\bigoplus_{n=1}^{\infty}S_n^*$ if and only if
        \begin{equation}\label{WF_criterion}
C_1= \sup_n\|M_n(i)\|_{\mathcal{H}_n}  < \infty
\quad\text{and}\quad   C_2 = \sup_n\|(\im
M_n(i))^{-1}\|_{{\cH}_n} < \infty.
   \end{equation}
         \end{theorem}
 Theorem \ref{th_criterion(bt)} makes it possible to
construct a boundary triplet by regularizing each summand  in a direct sum
$\Pi=\bigoplus_{n=1}^\infty \Pi_n$ of arbitrary boundary triplets. The
corresponding result was obtained in  \cite[Theorem
5.3]{MN2012} (see also \cite{MalNei11} and \cite[Theorems 3.10,
3.11]{KM}).

    \begin{theorem}[\cite{MN2012, MalNei11}]\label{th_III.2.1_02}
Let $S_n$ be a symmetric operator in $\mathfrak{H}_n$ with
deficiency indices $\mathrm{n}_{\pm}(S_k) = \mathrm{n}_n\le
\infty$ and $S_{n0}= S_{n0}^*\in \Ext S_n, \  n\in \N.$ Then for
any $n\in \N$ there exists a boundary triplet
$\Pi_n=\{\mathcal{H}_n, \Gamma_0^{(n)}, \Gamma_1^{(n)} \}$  for
$S_{n}^*$ such that $\ker\Gamma_0^{(n)}= \dom(S_{n0})$ and
$\Pi=\bigoplus_{n=1}^{\infty}\Pi_n$ forms an ordinary  boundary
triplet for $A^* =\bigoplus_{n=1}^{\infty}S_n^*$ satisfying
$\ker\Gamma_0 = \dom(\widetilde{A}_{0}):=
\bigoplus_{n=1}^{\infty}\dom(S_{n0}).$
     \end{theorem}

Next we assume that the operator $A = \bigoplus^{\infty}_{n=1}S_n$
has a regular real point, i.e., there exists $a={\overline
a}\in{\hat\rho}(A)$. The latter is equivalent to the existence of
$\varepsilon> 0$ such that
   \begin{equation}\label{III.2.2_01}
(a-\varepsilon, a + \varepsilon) \subset \cap^{\infty}_{n=1}
{\widehat\rho}(S_n).
       \end{equation}

Emphasize that  condition
$a\in\cap^{\infty}_{n=1}{\widehat\rho}(S_n)$ is not enough for the
inclusion ${a\in\widehat\rho}(A).$

It is known \cite{Kre47} (see also \cite{DM91}) that under condition \eqref{III.2.2_01}  for every
$k\in{\N}$ there exists a selfadjoint extension ${\widetilde
S}_{k} = {\widetilde S}^*_{k}$ of $S_k$ preserving the gap
$(a-\varepsilon, a+\varepsilon)$. Moreover, the Weyl function of
the pair $\{S_k, {\widetilde S}_{k}\}$ is regular within the gap
$(a-\varepsilon, a+\varepsilon)$. Assuming condition
\eqref{III.2.2_01} to be satisfied, one can simplify conditions
\eqref{WF_criterion} of Theorem \ref{th_criterion(bt)} (cf.
\cite[Theorem 3.13]{KM}). In the following theorem  we slightly complete \cite[Theorem
3.13]{KM}.   
       \begin{theorem}\label{th_III.2.2_01}
Let $\{S_n\}_{n= 1}^\infty$ be a sequence of symmetric operators
satisfying \eqref{III.2.2_01}. Let also  $\Pi_n=\{\cH_n,
\Gamma^{(n)}_0, \Gamma^{(n)}_1\}$ be a boundary triplet for
$S^*_n$ such that
$(a-\varepsilon,a+\varepsilon)\subset\rho(S_{n0})$ and let
$M_n(\cdot)$ be  the corresponding Weyl function. Then:

\item $(i)$ $\Pi= \bigoplus^{\infty}_{n=1} \Pi_n$ forms a $B$-generalized
boundary triplet for $A^*= \bigoplus^{\infty}_{n=1}S_n^*$ if and only
if
    \begin{equation}\label{III.2.2_02}
C_3 := \sup_{n\in\N}\|M_n(a)\|_{\cH_n} < \infty \qquad \text{and}
\qquad
C_4:= \sup_{n\in\N}\|M'_n(a)\|_{\cH_n} < \infty,   
     \end{equation}
where  $M'_n(a):=({dM_n}(z)/{dz})|_{z=a}$.

\item $(ii)$  $\Pi= \bigoplus^{\infty}_{n=1} \Pi_n$ is an ordinary  boundary
triplet for $A^*= \bigoplus^{\infty}_{n=1}S_n^*$ if and only if in
addition to \eqref{III.2.2_02}  the following condition is
satisfied
    \begin{equation}\label{III.2.2_02NEW}
 C_5 := \sup_{n\in\N}\|\bigl(M'_n(a)\bigr)^{-1}\|_{\cH_n} < \infty.
     \end{equation}
        \end{theorem}
 \begin{proof} (i) According to \eqref{WF_intrepr} each $M_n(\cdot),\  n\in\mathbb N$, admits a representation
     \begin{equation}\label{2.21}
M_n(z) = C_{0,n} + \int_{\R\setminus
G_{\varepsilon}}\left(\frac{1}{t-z}-\frac{t}{1+t^2}\right)d\Sigma_n(t),\qquad
\int_\R \frac{1}{1 + t^2}d\Sigma_n(t) \in [\cH_n],
   \end{equation}
where $C_{0,n} = C_{0,n}^*\in [\cH_n]$ and
$G_{\varepsilon}:=(a-\varepsilon, a+\varepsilon)$. Hence
    \begin{equation}\label{2.22}
M_n(a) = C_{0,n} + \int_{\R\setminus
G_{\varepsilon}}\frac{1+at}{(t-a)(1+t^2)}d\Sigma_n(t) \quad
\text{and}\quad M_n'(a) =  \int_{\R\setminus
G_{\varepsilon}}\frac{1}{(t-a)^2}d\Sigma_n(t).
    \end{equation}
Noting that with some $k>0$
   \begin{equation}
|{(1+at)(t-a)}{(1+t^2)^{-1}}|\le k, \qquad  t\in{\mathbb R},
  \end{equation}
we get from \eqref{2.22} that the second condition in
\eqref{III.2.2_02} implies
  $$
\sup_n\|M_n(a) - C_{0,n}\|_{\cH_n} \le
k\sup_n\|M_n'(a)\|_{\cH_n} < \infty.
$$
Combining this estimate with the first condition in
\eqref{III.2.2_02}  yields $\sup_n\|C_{0,n}\|_{\cH_n} < \infty.$

Further, it follows from  \eqref{2.21} that
     \begin{equation}\label{2.24}
M_n(i) = C_{0,n} + i \int_{\R\setminus
G_{\varepsilon}}\frac{1}{1+t^2} d\Sigma_n(t) \in [\cH_n].
   \end{equation}
It is easily seen that there exist constants $k_1, \  k_2>0$ such
that
    \begin{equation}\label{2.25}
0<k_1< (1 + t^2)(t-a)^{-2} < k_2, \qquad    t\in\mathbb
R\setminus(a-\varepsilon, a+\varepsilon).
    \end{equation}
Taking this inequality into account and combining \eqref{2.24}
with \eqref{2.22}   we get that  the second condition  in
\eqref{III.2.2_02}  is equivalent to
$\sup_n\|M_n(i)-C_{0,n}\|_{\cH_n} < \infty.$ Combining this
estimate with $\sup_n\|C_{0,n}\|_{\cH_n} < \infty$ yields
$\sup_n\|M_n(i)\|_{\cH_n} < \infty$,  i.e. $M(i)\in[\cH]$. The
latter  means that $M(\cdot)\in R[\cH]$. Since $\ker \im M(i)=0$,
it remains to  apply Proposition \ref{prop2.10}(iii).  

(ii)  Using representation \eqref{2.22}  for $M'_n(a)$ we rewrite
condition \eqref{III.2.2_02NEW}  as
     \begin{equation}
\int_{\R\setminus G_{\varepsilon}}\frac{1}{(t-a)^2}d\Sigma_n(t) =
M_n'(a)   \ge C_5^{-1}, \qquad n\in\mathbb N.
     \end{equation}
Combining these inequalities with representation \eqref{2.24} and
taking into  account  inequality \eqref{2.25}  we obtain
    \begin{eqnarray}
\im M_n(i)= \int_{\R\setminus G_{\varepsilon}}\frac{1}{1+t^2}
d\Sigma_n(t) = \int_{\R\setminus
G_{\varepsilon}}\frac{1}{(t-a)^2}\cdot\frac{(t-a)^2}{1+t^2}
d\Sigma_n(t)  \nonumber   \\
\ge k^{-1}_2 \int_{\R\setminus
G_{\varepsilon}}\frac{1}{(t-a)^2}d\Sigma_n(t)  \ge k^{-1}_2
C^{-1}_5, \qquad n\in \N.
    \end{eqnarray}
This is amount  to saying that  $\sup_n\|(\im
M_n(\I))^{-1}\|_{{\cH}_n} \le k_2 C_5$. To complete the proof it remains to apply
Theorem \ref{th_criterion(bt)}.
 \end{proof}

For operators $A = \bigoplus_{n=1}^\infty S_n$ satisfying
\eqref{III.2.2_01} we complete Theorem \ref{th_III.2.2_01}  by
presenting a regularization procedure for $\Pi =
\bigoplus_{n=1}^\infty\Pi_n$ leading to a boundary triplet.
In applications to symmetric  operators with a gap this regularization  is
substantially simpler  than the one  described in Theorem
\ref{th_III.2.1_02}.
   \begin{corollary}\label{cor_III.2.2_02}
Let $\{S_n\}^{\infty}_{n=1}$ be a sequence of symmetric operators
satisfying \eqref{III.2.2_01}. Let also ${\widetilde \Pi}_n =
\{\cH_n, {\widetilde\Gamma}^{(n)}_0, {\widetilde\Gamma}^{(n)}_1\}$
be a boundary triplet for $S^*_n$ such that $(a-\varepsilon,
a+\varepsilon)\subset\rho(S_{n0})$, $S_{n0}=
S_n^*\upharpoonright\ker({\widetilde\Gamma}^{(n)}_0)$, and
${\widetilde M}_n(\cdot)$ the corresponding Weyl function. Assume
also that for some  operators $R_n$  such that $R_n,  R^{-1}_n
\in[\cH_n]$, the following conditions are satisfied
%
       \begin{equation}\label{III.2.2_11}
\sup_n \|R_n^{-1}({\widetilde M}'_n(a))(R_n^{-1})^*\|_{\cH_n} <
\infty \quad \text{and}\quad \sup_n \|R_n^*({\widetilde
M}'_n(a))^{-1}R_n\|_{\cH_n} < \infty, \quad n\in\N.
       \end{equation}
Then the direct sum $\Pi=\bigoplus_{n=1}^\infty \Pi_n$ of boundary
triplets
    \begin{equation}\label{III.2.2_08}
\Pi_n=\{\cH_n,\Gamma^{(n)}_0,\Gamma^{(n)}_1\}\quad
\text{with}\quad \Gamma^{(n)}_0 := R_n{\widetilde\Gamma}^{(n)}_0,\qquad
 \Gamma^{(n)}_1:= (R^{-1}_n)^*\bigl({\widetilde\Gamma}^{(n)}_1 -
 \wt M_n(a){\widetilde\Gamma}^{(n)}_0\bigr),
     \end{equation}
forms a boundary triplet for $A^* = \bigoplus_{n=1}^\infty S_n^*$.
  \end{corollary}
   \begin{example}\label{examp2.14}
\item $(i)$  Let $F(z)= Bz$ where $B\in[\cH],\  B=B^*>0$ and
$0\in\sigma(B)\setminus\sigma_p(B)$. Then $0(\in \sigma_c(B))$ is
an accumulation point for $\sigma(B)$ and the operator  $B$ admits a
decomposition $B = \oplus^{\infty}_{n=1} B_n$ with $B_n=B^*_n\in[\cH_n]$
and  $0\in\rho(B_n), n \in \mathbb N.$  Clearly,
$F(\cdot) = \bigoplus^{\infty}_{n=1} F_n(\cdot)\in R[\cH],$ where $F_n(z) =
B_n z$ and $F_n(\cdot)\in R^u[\cH_n]$, i.e. $0\in\rho\bigl(\im
F_n(i)\bigr)$, $n\in \N$.  However, $-F^{-1}(\cdot)\in
R(\cH)\setminus R[\cH]$, i.e. $-F^{-1}(\cdot)$ is a Nevanlinna
function with (unbounded) values in $\mathcal C(\cH)$. Clearly,
    \begin{equation}
F_n(0)=0,\qquad  F'_n(0)=B_n\in[\cH_n] \quad \text{and}\quad
\sup_{n\in \N}\|F'_n(0)\| = \|B\| <\infty,
    \end{equation}
and conditions  \eqref{III.2.2_02}  are satisfied. At the same
time, $\sup_{n\in \N}\|\bigl(F'_n(0)\bigr)^{-1}\| = \sup_{n\in
\N}\|B_n^{-1}\| = \infty$ and condition \eqref{III.2.2_02NEW} is
violated. Thus, condition \eqref{III.2.2_02NEW} is not implied  by  conditions \eqref{III.2.2_02}.

\item $(ii)$  Let $F(z) = Bz$ where $B=B^*\in \mathcal C(\cH)\setminus [\cH]$ is unbounded positively definite operator, $0\in \rho(B)$.
Clearly, $B = \bigoplus^{\infty}_{n=1}B_n$ where  $B_n\in[\cH],  n\in \N$ and $F(\cdot) = \bigoplus^{\infty}_{n=1} F_n(\cdot)$ with $F_n = B_nz$.
It is easily seen that
       \begin{equation}
F_n(0)=0, \quad    \bigl(F'_n(0)\bigr)^{-1} = B^{-1}_n \quad \text{and}\quad  \sup_{n\in \N}\|\bigl(F'_n(0)\bigr)^{-1}\| = \sup_{n\in\mathbb N}\|B_n^{-1}\| = \|B^{-1}\| < \infty.
       \end{equation}
On the other hand, $\sup_{n\in \N}\|F'_n(0)\| = \sup_{n\in
\N}\|B_n\| = \infty$ and the second condition in \eqref{III.2.2_02}  is
violated. This example shows that  the second condition in  \eqref{III.2.2_02}
does not follow from  the first one  and condition \eqref{III.2.2_02NEW}.
      \end{example}
      \begin{remark}
\item $(i)$ In  \cite[Theorem 3.13]{KM} it is incorrectly stated that the
direct sum $\Pi= \bigoplus^{\infty}_{n=1} \Pi_n$  forms an ordinary
boundary triplet for $A^*$  whenever both  the
first condition in  \eqref{III.2.2_02} and condition \eqref{III.2.2_02NEW} are satisfied. However the  proof in \cite[Theorem 3.13]{KM}  can easily  be fixed by posing
the second condition in \eqref{III.2.2_02} and using formula
\eqref{II.1.3_02} connecting $M(i)$ and $M(a)$. In Theorem
\ref{th_III.2.2_01}(ii) we presented another proof of this fact
that seems to be simpler.

\item $(ii)$ Note also that the first  inequality in \eqref{III.2.2_11}
was occasionally missed in  \cite[Corollary 3.15]{KM}. We mention
also a misprint in formula (59) of \cite{KM}: there should be
$R_n$ in place of $R_n^{-1}$.
   \end{remark}

\section{Dirac operators with point interactions on the line}
Let $D$ be the differential expression
\begin{equation}\label{1.2}
D=-i\,c\,\frac{d}{dx}\otimes\sigma_{1}+\frac{c^{2}}{2}\otimes\sigma_{3}=\left(\begin{array}{cc}{c^{2}}/{2} & -i\,c\,\frac{d}{dx} \\-i\,c\,\frac{d}{dx} & -{c^{2}}/{2}\end{array}\right)
\end{equation}
acting on $\C^2$-valued functions of a real variable. Here
\begin{equation}\label{pauli}
\sigma_{1}=\left(\begin{array}{cc}0 & 1 \\1 &
0\end{array}\right),\quad\sigma_{2}=\left(\begin{array}{cc}0 & -i
\\i & 0\end{array}\right), \quad\sigma_{3}=\left(\begin{array}{cc}1
& 0 \\0 & -1\end{array}\right),
\end{equation}
are the Pauli matrices in $\C^{2}$ and $c>0$ denotes the velocity
of light.\par We set
\begin{equation}\label{1.5}
k(z):=c^{-1}\sqrt{z^{2}-\left({c^{2}}/{2}\right)^{2}},\qquad z\in\C,
\end{equation}
where the branch of the multifunction $\sqrt{\cdot}$ is selected
such that $k(x)>0$ for $x>{c^{2}}/{2}$. It is easily seen that
$k(\cdot)$ is holomorphic in $\C$ with two cuts along the half-lines
$\left(-\infty,-{c^{2}}/{2}\right]$ and
$\left[{c^{2}}/{2},\infty\right)$ and
$k(\overline{z})=-\overline{k(z)}$.

\noindent
Thus, $k(\cdot)$ itself is not R-function (Nevanlinna function), although the extension
   \begin{equation}
\widetilde{k}(z)=\left\{\begin{array}{c}k(z),\quad z\in\C_{+}, \\
\\ -k(\overline{z}),\quad z\in\C_{-}\end{array}\right.
  \end{equation}
is already a R-function, i.e. a holomorphic function in
$\C\backslash\R$, that maps $\C_{+}$ into $\C_{+}$ and satisfies
$f(\overline{z})=\overline{f(z)}.$
Next we put
\begin{equation}\label{1.6}
k_{1}(z):=\frac{c\,k(z)}{z+{c^{2}}/{2}}
=\sqrt{\frac{z-{c^{2}}/{2}}{z+{c^{2}}/{2}}},\qquad z\in\C.
\end{equation}

We can independently define the right-hand side in
$\C\backslash\big\{\left(-\infty,-{c^{2}}/{2}\right]\bigcup\left[{c^{2}}/{2},\infty\right)\big\}$
by selecting the branch of the corresponding multifunction in such
a way that $\sqrt{\frac{x -{c^{2}}/{2}}{x +{c^{2}}/{2}}}>0$ for
$x>{c^{2}}/{2}$,
\par

Next we construct boundary triplets for $D_X^*$ using the technique elaborated in
\cite{KM} and \cite{MN2012}. 

\subsection{Boundary triplets for  Dirac building blocks}
We begin with a construction of a  boundary triplet for the
maximal Dirac operator on an interval.

\subsubsection{The case of a finite interval} Let $D_{n}$ be the minimal  operator generated in
$L^{2}[x_{n-1},x_{n}]\otimes\C^{2}$  by the differential
expression (\ref{1.2}) 
  \begin{equation}\label{3.6A}
D_{n}=D\upharpoonright\dom(D_n),
\qquad\dom(D_{n})=W^{1,2}_{0}[x_{n-1},x_{n}] \otimes\C^{2}.
\end{equation}
We also  put  $d_n:=x_{n}-x_{n-1}>0$.
      \begin{lemma}\label{2.3} $D_{n}$ is a symmetric operator
with deficiency indices $n_{\pm}(D_{n})=2$.  Its adjoint $D_{n}^*$ is given by

%
%
%
%

%
%
$$
D_{n}^*=D\upharpoonright\dom(D_{n}^*), \qquad
\dom(D_{n}^*)=W^{1,2}[x_{n-1},x_{n}]\otimes\C^{2}.
$$
The defect subspace $\mathfrak N_z :=\text{\rm ker}(D_{n}^*-z)$ is
spanned by the vector functions $f_n^\pm(\cdot, z)$,
  \begin{equation}\label{3.7}
f_n^{\pm}(x, z) := \begin{pmatrix}e^{\pm i\, k(z)\,x}\\\pm
k_{1}(z)e^{\pm i\,k(z)\,x}\end{pmatrix}.
   \end{equation}
Moreover, the following  is  true:
%
\item $(i)$ The triplet $\widetilde{\Pi}^{(n)}=\big\{\C^{2},
\widetilde{\Gamma}_{0}^{(n)},\widetilde{\Gamma}_{1}^{(n)}\big\}$,
where
   \begin{equation}\label{triple2}
\widetilde{\Gamma}_{0}^{(n)}f:=\widetilde{\Gamma}_{0}^{(n)}
\begin{pmatrix}f_{1}\\f_{2}\end{pmatrix}=
\left(\begin{array}{c}
                 f_{1}(x_{n-1}+)\\
                 i\,c\, f_{2}(x_{n}-)
                       \end{array}\right) \qquad\text{and}\qquad
\widetilde{\Gamma}_{1}^{(n)}f:=
\widetilde{\Gamma}_{1}^{(n)}\begin{pmatrix}f_{1}\\f_{2}\end{pmatrix}
=\left(\begin{array}{c}
                                                                           i\,c\,f_{2}(x_{n-1}+)\\
                                                                            f_{1}(x_{n}-)
                                                                          \end{array}\right)\,,
\end{equation}
forms a boundary triplet for $D^{*}_n$.

\item $(ii)$  The spectrum of the operator
${D}_{n,0}:=D_n^{*}\upharpoonright\ker\widetilde{\Gamma}^{(n)}_{0}$,  where
    \begin{equation}\label{3.8}
\dom(D_{n,0}) = 
\{\{f_1, f_2\}^{\tau}\in
W^{1,2}[x_{n-1},x_{n}]\otimes\C^{2}:  f_{1}(x_{n-1}+)=
f_{2}(x_{n}-)=0\},
     \end{equation}
is discrete,
\begin{equation}\label{spetcom}
\sigma({D}_{n,0})= \sigma_d({D}_{n,0})=
\left\{\pm\sqrt{\frac{
c^2\pi^2}{d^2_n}\,\left(j+\frac12\right)^{2}+\left(\frac{c^{2}}{2}\right)^{2}}\,,\quad j=0,1,\dots
\right\}.
\end{equation}
\item $(iii)$  The $\gamma$-field $\wt\gamma_n(\cdot):\C^2\to
  L^2[x_{n-1},x_{n}]\otimes\C^2$, corresponding to the triplet
  $\wt\Pi^{(n)}$ is given in the standard basis in $\C^2$ by
\begin{equation}\label{1.21+}
\wt\gamma_n(z)\binom{v_1}{v_2} = \frac{1}{\cos(d_nk(z))}
\left(\begin{array}{cc}\cos(k(z)(x_{n}-x)) & -(c\,k_1(z))^{-1}
\sin(k(z)(x_{n-1}-x)) \\
-i\,k_1(z)\sin(k(z)(x_{n}-x))&
-i\,c^{-1}\cos(k(z)(x_{n-1}-x))\end{array}\right)\binom{v_1}{v_2}, \qquad z\in\rho(D_{n,0}).
   \end{equation}
\item $(iv)$ The  Weyl function $ \widetilde{M}_{n}(\cdot)$
corresponding to the triplet $\wt\Pi^{(n)}$ is
           \begin{equation}\label{IV.1.1_09.2}
 \widetilde{M}_{n}(z)=\frac{1}{\cos(d_n\,k(z))}\left(\begin{array}{cc}c\,k_1(z)\,\sin(d_n\,k(z))
   & 1 \\1 & (c\,k_1(z))^{-1}\sin(d_n\,k(z))
\end{array}\right),\qquad z\in\rho(D_{n,0}).
         \end{equation}
%
%

\end{lemma}

\begin{proof}

\noindent
(i) and (ii)
are straightforward.

(iii)\  Since $f_n^-$ and $f_n^+$ form a basis in the defect
subspace  $\mathfrak N_z,$ we get from  the definition of the
$\gamma$-field,
$$
\wt\gamma_n(z)\begin{pmatrix}v_{1}\\v_{2}\end{pmatrix}=
w_1(z)f_n^-(x,z) + w_2(z)f_n^+(x,z). \quad
$$
Applying to this identity the mapping $\widetilde{\Gamma}_0^{(n)}$
and using \eqref{3.7}, \eqref{triple2} and definition
\eqref{II.1.3_01} we get
    \begin{equation}
\binom{v_1}{v_2}  =
\begin{pmatrix}
e^{-i k(z)x_{n-1}} & e^{ik(z)x_{n-1}}\\
-i\,c\,k_{1}(z)\,e^{-i k(z)x_n} & i\,c\,k_{1}(z)e^{ik(z)x_n}.
   \end{pmatrix}
\binom{w_1(z)}{w_2(z)} =: \Lambda(z)\binom{w_1(z)}{w_2(z)}
    \end{equation}
Hence $\binom{w_1(z)}{w_2(z)} = \Lambda^{-1}(z)\binom{v_1}{v_2}.$
Setting
$\Delta(z):=\det\Lambda(z)=2\,i\,c\,k_{1}(z)\,\cos\bigl(d_n
k(z)\bigr)$ we find
   \begin{eqnarray*}
\widetilde{\gamma}(z)\binom{v_1}{v_2}=\frac{1}{\Delta(z)} \left(i\,c\,k_{1}(z)\,e^{ik(z)x_{n}}v_1 - e^{ik(z)x_{n-1}}v_2\right)\begin{pmatrix}e^{-i\,
k(z)\,x}\\ -
k_{1}(z)e^{-i\,k(z)\,x}
 \end{pmatrix} \nonumber \\
+\  \frac{1}{\Delta(z)}\left(i\,c\,k_{1}(z)\,e^{-ik(z)x_{n}}v_1 +
e^{-ik(z)x_{n-1}}v_2\right)
\begin{pmatrix}e^{i\, k(z)\,x}\\
k_{1}(z)e^{i\,k(z)\,x}
\end{pmatrix} \nonumber \\
= \frac{1}{\cos(k(z)d_{n})
} \left(\begin{array}{cc}\cos(k(z)(x_{n}-x))
&
-(c\,k_{1}(z))^{-1}\sin(k(z)(x_{n-1}-x)) \\
-i\,k_1(z)\sin(k(z)(x_{n}-x))&
(i\,c)^{-1}\cos(k(z)(x_{n-1}-x))\end{array}\right) \binom{v_1}{v_2}.
   \end{eqnarray*}
This proves  \eqref{1.21+}.

(iv) This statement is immediate  from  \eqref{1.21+},
\eqref{triple2}   and the identity $\wt M_n(z)=
\wt\Gamma_1^{(n)}\wt\gamma_n(z)$.
\end{proof}

\subsubsection{The case of a half-line}
   In this section we construct  boundary triplets for the Dirac operator
$D$ on  half-lines $\R_a^{-}:=(-\infty,a)$ and  $\R_b^{+}:=(b,+\infty)$.

\noindent
Denote by $D_{a-}$ the minimal Dirac operator
generated  by differential expression (\ref{1.2})
in $L^{2}(\R_a^{-})\otimes\C^{2}$,  i.e.
  \begin{equation}
D_{a-}=D\upharpoonright\dom(D_{a-}),
\quad\dom(D_{a-})=W^{1,2}_{0}(\R_a^{-})
\otimes\C^{2}\,.
\end{equation}
    \begin{lemma}\label{2.3-}

$D_{a-}$ is a closed symmetric operator  with deficiency indices
$n_{\pm}(D_{a-})=1$. Its adjoint $D_{a-}^{*}$ is given by
   $$D_{a-}^{*}=D\upharpoonright\dom(D_{a-}^{*})\,,\qquad
\dom(D_{a-}^{*})=W^{1,2}(\R_a^{-})\otimes\C^{2}\,.$$
The defect subspace $\text{\rm ker}(D^*_{a-}-z)$ is spanned by the
vector function  
   \begin{equation}\label{3.17}
f_{a}^{-}(x, z) := \begin{pmatrix}e^{- i\, k(z)\,x}\\ -k_{1}(z)
e^{- i\,k(z)\,x}\end{pmatrix} .
   \end{equation}
Moreover, the following hold
%
\item $(i)$ The triplet ${\Pi}^{(a-)}=\big\{\C,
{\Gamma}_{0}^{(a-)},{\Gamma}_{1}^{(a-)}\big\}$ where
  \begin{equation}\label{triple.a}
{\Gamma}_{0}^{(a-)}f := {\Gamma}_{0}^{(a-)}
\begin{pmatrix}f_{1}\\f_{2}\end{pmatrix}=i\,c\,f_{2}(a-) \qquad\text{and}\qquad
{\Gamma}_{1}^{(a-)}f:={\Gamma}_{1}^{(a-)}
\begin{pmatrix}f_{1}\\f_{2}\end{pmatrix}
=f_{1}(a-)\,,
  \end{equation}
forms a boundary triplet for $D^{*}_{a-}$\,.
\item $(ii)$  The spectrum of the operator ${D}_{a-,0}:=D_{a-}^{*}
\upharpoonright\ker\widetilde{\Gamma}^{(a-)}_{0}$ is absolutely
continuous,  of the multiplicity one,
  \begin{equation}
\sigma({D}_{a-,0})=\sigma_{ac}({D}_{a-,0})
=(-\infty,c^2/2]\cup[c^2/2,+\infty)\,.
   \end{equation}
\item $(iii)$ The corresponding $\gamma$-field
$\gamma_{a-}(\cdot):\C\to L^2(\R_a^{-})\otimes\C^2$, is given by
%
%
  \begin{equation}\label{gamma.a}
\gamma_{a-}(z)w =w\,\frac{i\,e^{i\,
k(z)\,a}}{c\,k_1(z)}\,f_{a}^-(z),  \qquad z\in\rho(D_{a_-,0}).
  \end{equation}
\item $(iv)$
 The  Weyl function $ {M}_{a-}(\cdot)$, corresponding to the triplet
 $\Pi^{(a-)}$  is   
           \begin{equation}\label{IV.1.1_09.a}
{M}_{a-}(z)=\frac{i}{c\,k_1(z)}, \qquad z\in\rho(D_{a_-,0}).
         \end{equation}
%
%
   \end{lemma}
Next we denote by $D_{b+}$ the minimal Dirac  operator
generated by  differential expression (\ref{1.2})
in $L^{2}(\R_b^{+})\otimes\C^{2}$, i.e.
  \begin{equation}\label{3.20}
D_{b+}=D\upharpoonright\dom(D_{b+}),
\quad\dom(D_{b+})=W^{1,2}_{0}(\R_b^+)
\otimes\C^{2}\,.
\end{equation}
  \begin{lemma}\label{2.3+}
$D_{b+}$ is a symmetric operator  with the  deficiency indices
$n_{\pm}(D_{b+})=1$. The adjoint operator $D_{b+}^*$ is given by
   \begin{equation}\label{3.20A}
D_{b+}^{*}=D\upharpoonright\dom(D_{b+}^{*})\,,\qquad
\dom(D_{b+}^{*})=W^{1,2}(\R_b^{+})\otimes\C^{2}\,,
  \end{equation}
and the defect subspace $\mathfrak N_z = \text{\rm ker}(D^*_{b+}-z)$ is spanned by the vector
function $f_{b}^+(\cdot, z)$,
  \begin{equation}\label{3.23}
f_{b}^{+}(x, z):=\begin{pmatrix}e^{i\, k(z)\,x}\\ k_{1}(z) e^{
  i\,k(z)\,x}\end{pmatrix}.
  \end{equation}
Moreover, the following is true:
%
\item $(i)$ The triplet  ${\Pi}^{(b+)}=\big\{\C,
{\Gamma}_{0}^{(b+)},{\Gamma}_{1}^{(b+)}\big\}$ where
\begin{equation}\label{triple.b}
{\Gamma}_{0}^{(b+)}f := {\Gamma}_{0}^{(b+)}
\begin{pmatrix}f_{1}\\f_{2}\end{pmatrix}=f_{1}(b+) \qquad\text{and}\qquad
{\Gamma}_{1}^{(b+)}f:={\Gamma}_{1}^{(b+)}
\begin{pmatrix}f_{1}\\f_{2}\end{pmatrix}
=i\,c\,f_{2}(b+)\,,
  \end{equation}
forms a boundary triplet for $D^{*}_{b+}$.
\par\noindent

\item $(ii)$ The spectrum of the operator
${D}_{b+,0}:=D^*_{b+}\upharpoonright\ker{\Gamma}^{(b+)}_0=D_{b+,0}^*$
is absolutely continuous of the multiplicity  one,
\begin{equation}
      \sigma(D_{b+,0})=\sigma_{ac}(D_{b+,0}) = (-\infty,c^2/2]\cup[c^2/2,+\infty).
\end{equation}
\item $(iii)$  The corresponding $\gamma$-field
$\gamma_{b+}(\cdot):\C\to L^2(\R_b^{+})\otimes\C^2$, is
  \begin{equation}\label{gamma.b}
\gamma_{b+}(z)w =w\,e^{-i\, k(z)\,b}f_{b}^+(z), \qquad
z\in\rho(D_{b+,0}).
   \end{equation}
\item $(iv)$
 The Weyl function $ {M}_{b+}(\cdot)$,
corresponding to the triplet $\Pi^{(b+)}$ is 
           \begin{equation}\label{IV.1.1_09.b}
{M}_{b+}(z)=i\,c\,k_{1}(z), \qquad z\in\rho(D_{b+,0}).
         \end{equation}
%
\end{lemma}

The proofs of  Lemmas \ref{2.3-} and \ref{2.3+} are straightforward.


\subsection{Trace properties of functions from the Sobolev space $ W^{1,2}(\R_+\setminus X)$}
Let $X=\{x_{n}\}_{n=1}^\infty$ be a discrete subset of
the interval $\cI= (a, b)$,  $x_{n-1}<x_{n},\ n\in \N,$ with the accumulation point $b$, i.e. such that  $b:=\sup X\equiv\lim_{n\to\infty}x_n$.  We set  $x_{0}:=a$
and
%
%
\begin{equation}\label{3.26AAA}
d_n:=x_{n}-x_{n-1}\,,\qquad d_{*}(X):=\inf_{n}d_n\,,\qquad
d^{*}(X):=\sup_{n}d_n\,.
\end{equation}

In what follows we always  assume for convenience that $a:=x_{0}= 0.$
Then we define the minimal operator $D_X$ on $\mathfrak H =  L^{2}({\cI})\otimes \mathbb C^2$ 
by setting
  \begin{equation}\label{3.26}
D_X:=\bigoplus_{n\in \N}D_n\,. \qquad
   \end{equation}
where $D_n$, $n\in \N,$ is given by  \eqref{3.6A}.  Clearly,
   \begin{equation}\label{3.27}
\dom(D_X) = W^{1,2}_0({\R_+}\setminus X)\otimes \mathbb C^2 =
\bigoplus_{n\in \N} W^{1,2}_0[x_{n-1},x_{n}]\otimes\C^{2}.
   \end{equation}
Investigating non-relativistic limit in the case $b < \infty$ we will also consider operators
$D_X \bigoplus D_{b+}.$  

Here we construct (ordinary) boundary triplets for Dirac operators with point
interactions on the halfline as well as on the line. It is natural
to define a boundary triplet for $D^*_X= \bigoplus_{n=1}^\infty
D_n^*$ as the direct sum $\widetilde{\Pi} =
\bigoplus_{n=1}^\infty\widetilde{\Pi}_n$ of boundary triplets
$\widetilde{\Pi}^{(n)}=\big\{\C^{2},
\widetilde{\Gamma}_{0}^{(n)},\widetilde{\Gamma}_{1}^{(n)}\big\}$
defined in Lemma \ref{2.3}. However, $\widetilde{\Pi}$ is not an
ordinary boundary triplet,  in general. First we find necessary and
sufficient conditions for a discrete set
$X=\{x_n\}^{\infty}_{n=0}$ which guarantee this property for the
direct
sum $\bigoplus_{n=1}^\infty\widetilde{\Pi}_n.$  
This problem is closely related  to the property of trace
mapping defined on the Sobolev space $ W^{1,2}({\mathbb R_+})$ by
   \begin{equation}\label{3.29}
\pi:\   W^{1,2}({\mathbb R_+})\to l_2({\mathbb N}), \qquad \pi(f) =
\{f(x_n)\}^{\infty}_{n=1}.
   \end{equation}
%
%
   \begin{proposition}\label{prop3.5trace}
Let $X=\{x_n\}_{n=1}^\infty$ be as above. Then the mapping $\pi$
is surjective if and only if $d_*(X)>0$.
   \end{proposition}
   \begin{proof}
\emph{Sufficiency } Let $d_*(X)>0$. Denote by $u_0(\cdot)\in
C^{\infty}_0(\mathbb R)$ a function with compact support $\supp
u_0\subset \bigl(-d_*(X)/2,  d_*(X)/2\bigr)$ and satisfying
$u_0(0)=1$. Next we put $u_n(x) := u_0(x -x_{n-1} + d_n/2),\  n\in{\mathbb N},$ and
note that $\supp u_n\subset [x_{n-1}, x_n]$ and  $\|u_n\|_{W^{1,2}} = \|u_0\|_{W^{1,2}},\  n\in\mathbb N$.
Since $\supp u_k\cap\supp u_j=\emptyset$ for $j\not = k$,  for any sequence $\{a_k\}^{\infty}_1\in l^2(\mathbb N)$ the following
series converges  in $W^{1,2}(\R_+)$,
    \begin{equation*}
f := \sum^{\infty}_{k=1}a_k u_k\in W^{1,2}(\mathbb R_+)
\qquad\text{and}\qquad
\|f\|^2_{W^{1,2}(\R_+)} = \|u_0\|^2_{W^{1,2}(\R_+)} \cdot\sum^{\infty}_{k=1}a^2_k.
    \end{equation*}
Clearly, $f(x_k) = a_k, \ k\in \N$, i.e.
 $\pi(f)=\{a_k\}^{\infty}_1$ and the mapping $\pi$ is surjective.

\emph{Necessity.}  Assume that $\pi$ is surjective. Choose any sequence $\{u_n\}^{\infty}_1\in
W^{1,2}({\mathbb R_+})$ satisfying
    \begin{equation}
u_n(x_{n})=1,\quad  u_n(x_{n-1})=0, \qquad  n\in\N.
    \end{equation}
Then, with the above  notation $d_n = x_{n} - x_{n-1}$ we get
 \begin{equation}\label{3.32}
d^{-1}_n = d^{-1}_n\left(\int_{x_{n-1}}^{x_n}u'_n(t)dt\right)^2
\le \int_{x_{n-1}}^{x_n}|u'_n(t)|^2 dt \le
\|u_n\|^2_{W^{1,2}({\R_+})}, \qquad  n\in \N.
 \end{equation}
If $\pi$ is surjective, then, by closed graph theorem,
there  exists a bounded "inverse", i.e. a  surjective mapping $\pi^{(-1)}$ such that
    \begin{equation}
\pi^{(-1)}:l^2(\mathbb N)\to W_1\subset W^{1,2}(\R_+), \qquad
\pi\pi^{(-1)}=I_{l^2},
    \end{equation}
where $W_1$ is a (closed) subspace of $W^{1,2}({\R_+})$. Hence there
exists a bounded in $W^{1,2}({\R_+})$ sequence
$\{v_n\}^{\infty}_1\subset W^{1,2}({\R_+})$ covering the coordinate
basis  $e_n:=\{\delta_{mn}\}^{\infty}_{m=1},\ n\in{\N},$ in
$l^2(\N)$, i.e. satisfying $\pi(v_n) = e_n,\  n\in{\N}$.
Substituting the sequence  $\{v_n\}^{\infty}_1$ in  \eqref{3.32}
in place of $\{u_n\}^{\infty}_1$,  we conclude that the sequence
$\{d^{-1}_n\}_1^\infty$ is bounded, i.e. $d_*(X)>0$.
  \end{proof}

Next we give a complete trace characteristic  of the space
$W^{1,2}(\R_a^+\setminus X)$ assuming for convenience  that $a=0$.
Due to the embedding theorem,  the trace mappings 
     \begin{equation}\label{3.31}
\pi_{\pm}:  W^{1,2}(\R_+\setminus X) \to l^2(\N), \qquad
\pi_+(f)=\{f(x_{n-1}+)\}^{\infty}_1, \qquad
\pi_-(f)=\{f(x_n-)\}^{\infty}_1,
     \end{equation}
are well defined for functions with compact supports, i.e. for
$f\in \bigoplus_1^N W^{1,2}[x_{n-1},x_n],$  $N\in \mathbb N.$
We  assume  $\pi_{\pm}$ to be defined on its maximal domain $\dom(\pi_{\pm}):= \{f\in W^{1,2}(\R_+\setminus X):\ \pi_{\pm}f\in l^2(\N)\}$. Clearly,  $\dom(\pi_{\pm})$ is dense in  $W^{1,2}(\R_+\setminus X)$ although, in
general, $\dom(\pi_{\pm}) \not =  W^{1,2}(\R_+\setminus X).$
   \begin{proposition}\label{prop3.6}
Let $X =\{x_{n}\}_{n=1}^\infty$ be as above with $x_{0}=0$ and
$X\subset\overline\R_+$.  Then:

\item $(i)$  For any pair of sequences $a^{\pm} = \{a^{\pm}_n\}_1^{\infty}$ satisfying
    \begin{equation}\label{3.36}
a^{\pm} = \{a^{\pm}_n\}^{\infty}_1  \in l^2({\mathbb
N};\{d_n\}) \qquad \text{and} \qquad \{a^+_n-a^-_n\}^{\infty}_1\in l^2({\mathbb
N};\{d^{-1}_n\}),
   \end{equation}
there exists a (non-unique) function
$f\in W^{1,2}(\R_+\setminus X)$ such that $\pi_{\pm}(f) =
a^{\pm}$.  Moreover,  the mapping $\pi_+-\pi_-:
W^{1,2}(\R_+\setminus X)  \to l^2({\mathbb N};\{d^{-1}_n\})$ is
surjective and contractive, i.e.
    \begin{equation}\label{3.30AA}
\sum_{n\in\N} d_n^{-1}{|f(x_n-) - f(x_{n-1}+)|^2} \le
\|f\|^2_{W^{1,2}({\R}_+\setminus X)},\qquad f\in
W^{1,2}(\R_+\setminus X).
    \end{equation}

\item $(ii)$  Assume in addition, that $d^{*}(X)<\infty.$ Then the mapping $\pi_{\pm}$
can be extended  to a bounded surjective  mapping from $W^{1,2}(\R_+\setminus X)$
onto $l^2({\mathbb N};\{d_n\})$. More precisely,  the following estimate holds
   \begin{equation}\label{3.30AB}
\sum_{n\in\N}d_n \left( |f(x_{n-1}+)|^2 + |f(x_n-)|^2\right)\le 4\left(d^{*}(X)^2 \|f'\|^2_{L^{2}({\R}_+)} + \|f\|^2_{L^{2}({\R}_+)}\right)
\le C_1\|f\|^2_{W^{1,2}({\R}_+\setminus X)}, \qquad f\in
W^{1,2}(\R_+\setminus X).
    \end{equation}
where  $C_1 := 4\max\{d^{*}(X)^2, 1\}$. Besides, the traces $a^{\pm} := \pi_{\pm}(f)$ of each $f\in W^{1,2}({\R}_+\setminus X)$  satisfy conditions \eqref{3.36}.  Moreover,  the assumption $d^{*}(X)<\infty$ is necessary  for the inequality   \eqref{3.30AB} to hold with some $C_1>0$.

\item $(iii)$ The trace mapping $\pi_{\pm}: \dom(\pi_{\pm}) \to
l^2(\N)$ is  closed. Moreover, it is  surjective, $\ran(\pi_{\pm})\supset l^2({\mathbb N}\}$, if and only if  $d^{*}(X)<\infty.$

\item $(iv)$ The  mapping $\pi_{\pm}$ is  bounded, i.e. $\dom(\pi_{\pm})
= W^{1,2}({\mathbb R}_+\setminus X)$, if and only if
    \begin{equation}\label{3.34}
0<d_{*}(X)<d^{*}(X)<\infty.
    \end{equation}
   \end{proposition}
  \begin{proof} (i)
Let  conditions \eqref{3.36} be satisfied.
Define a function $g_n$ by setting
    \begin{equation}\label{3.37}
g_n(x) = a^+_n+d^{-1}_n(x-x_{n-1})(a^-_n-a^+_n),\qquad
x\in[x_{n-1},x_n], \quad n\in{\mathbb N}.
  \end{equation}
Let us check that the piecewise linear function $g = \oplus_{n\in
\N}g_n$ has the required properties.
Clearly, $g_n(x_{n-1}+)= a^+_n$ and  $g_n(x_{n}-) = a^-_n$. Moreover,
$g'_n(x) = d^{-1}_n(a^-_n-a^+_n)$ and
     \begin{equation} \label{3.39A}
\|g'\|^2_{L^2({\mathbb R}_+)} = \sum_n \|g'_n\|^2_{L^2[x_{n-1},x_n]}
= \sum_n d^{-1}_n|a^-_n - a^+_n|^2 < \infty.
     \end{equation}
In other words, $g'\in L^2(\mathbb R_+)$ if and only if $\{a^+_n-a^-_n\}^{\infty}_1\in
l^2({\mathbb N};\{d^{-1}_n\})$.

Next,  it is easily seen that
     \begin{eqnarray}\label{3.38A}
\|g_n\|^2_{L^2[x_{n-1},x_n]} = d_n\left( \re(a^+_n\overline{a^-_n}) + 3^{-1}|a^-_n - a^+_n|^2\right) \nonumber   \\
= 3^{-1}d_n \left(|a^+_n|^2 + |a^-_n|^2 + \re(a^+_n\overline{a^-_n})\right),
\qquad  n\in\mathbb N.
     \end{eqnarray}
On the other hand, by the Cauchy-Schwartz inequality
     \begin{equation}\label{3.39B}
6^{-1}(|z_1|^2+|z_2|^2)\le 3^{-1}\bigl( |z_1|^2+|z_2|^2 +
\re(z_1\overline{z}_2)\bigr) \le 2^{-1}(|z_1|^2+|z_2|^2).
  \end{equation}
Combining  \eqref{3.38A} with \eqref{3.39B}  we arrive at  the
following two-sided estimate for $g = \oplus_{n\in \N} g_n$
    \begin{eqnarray}\label{3.40}
6^{-1}\sum^{\infty}_{n=1}d_n \bigl(|a^+_n|^2 + |a^-_n|^2 \bigr)\le
\|g\|^2_{L^2(\mathbb R_+)} =  \sum^{\infty}_{n=1}
\|g_n\|^2_{L^2[x_{n-1},x_n]}   
\le 2^{-1} \sum^{\infty}_{n=1}d_n
\bigl(|a^+_n|^2 + |a^-_n|^2 \bigr).
    \end{eqnarray}
In other words,   $g\in L^2(\mathbb R_+)$ if and
only if $a^{\pm}  \in l^2(\mathbb N;\{d_n\})$.
Thus, it follows from \eqref{3.40} and \eqref{3.39A}  that $g=\oplus_{n\in \N} g_n\in W^{1,2}(\mathbb R_+\setminus X)$
if and only if both assumptions in \eqref{3.36} are satisfied.

To prove   surjectivity of the mapping $\pi_+:W^{1,2}(\mathbb
R_+\setminus X)\to l^2(\mathbb N;\{d_n\})$ we choose any
$a^+=\{a^+_n\}\in l^2(\mathbb N;\{d_n\})$ and put
$a^- := \{a^-_n\}=a^+$. Clearly,
$\{a^+_n-a^-_n\}^{\infty}_1=\{0\}^{\infty}_1\in l^2(\mathbb
N;\{d^{-1}_n\})$ and both conditions \eqref{3.36} are satisfied.
Thus, the step function  $g=\oplus_{n\in \N} g_n$ with $g_{n}:=
a_n \in W^{1,2}[x_{n-1},x_n], \  n\in \mathbb N,$ belongs to $W^{1,2}(\mathbb R_+\setminus X)$ and  satisfies  $\pi_{\pm}(g)=a^{\pm}$.

Further,  for any $f_n\in W^{1,2}[x_{n-1}, x_n],$ $n\in \mathbb N,$
one easily gets
    \begin{equation}\label{3.35}
 d_n^{-1}{|f_n(x_n-) - f_n(x_{n-1}+)|^2} = d^{-1}_n\left(\int_{x_{n-1}}^{x_n}f'_n(t)dt\right)^2 \le
\int_{x_{n-1}}^{x_n}|f'_n(t)|^2 dt \le \|f_n\|^2_{W^{1,2}[x_{n-1},
x_n]}.
   \end{equation}
Taking a sum one arrives at the  inequality \eqref{3.30AA}.

(ii)  Next, let  $d^*(X)<\infty$.  By the  Sobolev embedding
theorem, for any $f\in W^{1,2}(\mathbb R_+\setminus X)$
    \begin{equation}\label{3.45}
\max\left\{d_n |f(x_{n-1}+)|^2,\  d_n|f(x_n-)|^2\right\}
\le 2 \left(d^2_n\|f'\|^2_{L^2[x_{n-1},x_n]} +
\|f\|^2_{L^2[x_{n-1},x_n]}\right), \qquad n\in \mathbb N.
   \end{equation}
Taking a sum   and noting that  $d^*(X)<\infty$ we arrive at   \eqref{3.30AB}. It follows that the mapping $\pi_{\pm}$ (see \eqref{3.31}) originally defined on functions with compact supports can be extended
to  bounded surjective  mappings from $W^{1,2}(\R_+\setminus X)$
onto $l^2({\mathbb N};\{d_n\})$.

Further, let $f\in W^{1,2}(\mathbb R_+\setminus X)$ and let  $a^+_n :=f(x_{n-1}+),\  a^-_n := f(x_{n}-)$.
Since $d^*(X)<\infty$, it follows from  \eqref{3.30AB} that
the sequences $a^{\pm}=\{a^{\pm}_n\}_1^\infty$ satisfy the first condition in  \eqref{3.36}.
The second condition in  \eqref{3.36} is implied by \eqref{3.30AA}.

It remains to prove  the necessity of the assumption $d^{*}(X)<\infty$  for  the validity of inequality \eqref{3.30AB}.
Choose  $f_0 \in W^{1,2}[0,1]$ such that
$\|f_0\|^2_{L^2[0,1]} = \frac{1}{C_1}$ and $f_0(0)=f_0(1) = 1$  and  put
     \begin{equation}
f_n(x) := f_0\bigl((x-x_{n-1})d^{-1}_n\bigr), \qquad    n\in\mathbb N.
     \end{equation}
Clearly,  $f_n(x_{n-1}+) = f_n(x_{n}-) =1$  and
    \begin{equation}\label{3.48}
\|f_n\|^2_{W^{1,2}[x_{n-1}, x_n]} = \frac{1}{d_n}\int^1_0|f'_0(t)|^2 dt +  d_n \int^1_0|f_0(t)|^2 dt, \qquad n\in \N.
   \end{equation}
Substituting $f_n$ in \eqref{3.30AB} with account of \eqref{3.48} we arrive at the estimate
    \begin{eqnarray*}
2d_{n} \le C_1 \left(d_{n}^{-1}\|f_0'\|^2_{L^2[0,1]} + d_{n}\|f_0\|^2_{L^2[0,1]}\right)
= C_1 \left(d_n^{-1} \|f'_0\|^2_{L^{2}[0,1]}  + C_1^{-1}d_{n}\right), \quad n\in \N.
    \end{eqnarray*}
This estimate  is equivalent  to
$$
d_{n} \le C_1 d_{n}^{-1}\|f_0'\|^2_{L^2[0,1]}, \qquad n\in \N.
$$
In turn, the latter is equivalent to  $d_{n}^2 \le C_1 \|f_0'\|^2_{L^2[0,1]}$  for $n\in \N$, which implies  $d^*(X)<\infty$.

(iii) Let $\lim_{n\to\infty}f_n = f$ in $W^{1,2}(\mathbb R_+\setminus X)$ and   $\lim_{n\to\infty}\pi_{\pm}(f_n) = h_{\pm}$ in  $l^2({\mathbb N})$.  By the embedding theorem,
$\pi_{\pm}(f_n)$  weakly converges to $\pi_{\pm}(f)$ as $n\to\infty.$ Thus, $f\in \dom(\pi_{\pm})$, $\pi_{\pm}(f) = h_{\pm}$  and the mapping  $\pi_{\pm}$ is closed.

Further, let   $d^*(X)<\infty.$  Then the space $l^2({\mathbb N})$  is
continuously  embedded in $l^2({\mathbb N};\{d_n\})$. Therefore the
surjectivity  of the mapping  $\pi_{\pm}$  is implied by the
statement (i).

Conversely, let the trace mapping $\pi_+: \dom(\pi_+) \to l^2({\mathbb N})$ be  surjective, i.e.
$\ran(\pi_+) = l^2({\mathbb N}).$ Then, by (i), $l^2({\mathbb N})$  is a
subset of $l^2({\mathbb N};\{d_n\}).$  It is easily seen  that the
identical  embedding $i_+: l^2({\mathbb N}) \hookrightarrow
l^2({\mathbb N};\{d_n\})$ is closed.  By the closed graph
theorem, $i_+$ is continuous, i.e. there exists a constant $C>0$
such that $\sum_n d_n|a_{n}|^2 \le C\sum_n|a_{n}|^2.$  It follows
that $d^*(X) = \sup_n d_n \le C.$

(iv) Let conditions \eqref{3.34} be satisfied and
$f=\bigoplus^{\infty}_1 f_n\in W^{1,2}({\mathbb R}_+\setminus X)$. Since
$d_*(X)>0$, we get from \eqref{3.30AB}
    \begin{equation}
\|\pi_+(f)\|_{l^2(\N)}= \sum_{n\in\N} |f_n(x_{n-1}+)|^2  \le
d_*(X)^{-1} \sum_{n\in\N}d_n |f_n(x_{n-1}+)|^2  \le C_3\|f\|^2_{
W^{1,2}({\mathbb R}_+\setminus X)},
    \end{equation}
where $C_3 = C_1d_*(X)^{-1}.$ Similarly we get
$\|\pi_-(f)\|_{l^2(\N)}= \sum_{n\in\N} |f_n(x_{n}-)|^2 \le
C_3\|f\|^2_{ W^{1,2}({\mathbb R}_+\setminus X)}.$

Conversely, let $\dom(\pi_+)= W^{1,2}({\mathbb R}_+\setminus X)$.
Then $l^2(\mathbb N)$ is isomorphic algebraically and topologically
to the quotient  space $W^{1,2}({\mathbb R}_+\setminus X)/\ker\pi_+.$
Combining this fact with the statement (i), we get that  $l^2(\mathbb
N)$ is a subset of the weighted $l^2$-space $l^2({\mathbb N};\{d_n\})$.
Hence, as it is proved in the step (iii), $d^*(X)<\infty.$  In turn, by the statement (ii),  we get that $l^2(\N)$ coincides algebraically and topologically with
$l^2({\mathbb N};\{d_n\})$. This immediately yields
conditions \eqref{3.34}.
       \end{proof}
   \begin{remark}
Let $d^*(X)<\infty$.  Starting with $f  \in W^{1,2}(\mathbb R_+\setminus X)$ we  set $a^+_n :=f(x_{n-1}+),\  a^-_n := f(x_{n}-),\ n\in \mathbb N,$ and  define the  function  $g = \oplus_1^{\infty}g_n$   where  $g_n, \ n\in \N$, is  given by \eqref{3.37}.
It is proved in statement (ii) that
the sequences $a^{\pm}=\{a^{\pm}_n\}_1^\infty$ satisfy conditions \eqref{3.30AB}
and  $g\in W^{1,2}(\mathbb R_+\setminus X)$.
Therefore  $f$ admits the unique decomposition
 $$
f=g + f_0, \qquad \text{where}\qquad  f_0:= f-g \in W^{1,2}_0(\mathbb R_+\setminus X).
$$
In the case $d^*(X) = \infty$ this decomposition fails since $g\not \in W^{1,2}_0(\mathbb R_+\setminus X)$, in general.
   \end{remark}
\begin{remark} Assume that $d^*(X)<\infty$. Then using the continuity  and surjectivity  of the trace  mapping  $\tau:=(\pi_-,\pi_+)$ furnished in Proposition \ref{prop3.6},  and  following  the approach from \cite{Pos01} one can  obtain a description of the set of self-adjoint extensions of the operator $D_{X}$  by means  of the  Krein type formula for resolvents.
It is a way alternative  to that discussed  in the next  section.
\end{remark}

\subsection{Boundary triplets for Dirac operators with point interactions}
Here we  construct a boundary triplet  for the operator $D_X^* :=\bigoplus_{n=1}^\infty D_n^*.$
First we show that without additional restriction on $X$ the direct sum $\bigoplus_{n=1}^\infty
\widetilde{\Pi}^{(n)}$  of boundary triplets $\widetilde{\Pi}^{(n)}$  given by
\eqref{triple2}  forms only \emph{a $B$-generalized boundary triplet}  for $D_X^*.$
%
  \begin{proposition}\label{prop3.5direcsum}
Let $X$ be as above, $d^*(X)< \infty$,  and let
$\widetilde{\Pi}^{(n)}=\big\{\C^{2},
\widetilde{\Gamma}_{0}^{(n)},\widetilde{\Gamma}_{1}^{(n)}\big\}$
be the boundary triplet for the operator $D_n^*$ defined in Lemma
\ref{2.3}. Let also $A := D_X:=\bigoplus_{1}^\infty D_n,$   $\cH =
l^{2}(\N)\otimes\C^{2},$ and $\widetilde{\Gamma}_{j}
=\bigoplus_{n=1}^\infty \widetilde{\Gamma}^{(n)}_{j},$
$j\in\{0,1\},$ i.e.
   \begin{equation}\label{triple2NEW}
\widetilde{\Gamma}_{0} \begin{pmatrix}f_{1}\\f_{2}\end{pmatrix}=
\left\{\left(\begin{array}{c}
                 f_{1}(x_{n-1}+)\\
                 i\,c\, f_{2}(x_{n}-)
                       \end{array}\right)\right\}_{n\in \N}\,\qquad\text{and}\qquad
\widetilde{\Gamma}_{1}\begin{pmatrix}f_{1}\\f_{2}\end{pmatrix}
=\left\{\left(\begin{array}{c}
                                      i\,c\,f_{2}(x_{n-1}+)\\
                                       f_{1}(x_{n}-)
                                       \end{array}\right)\right\}_{n\in \N}\,,\qquad f= \binom{f_1}{f_2}\in \dom(D_X^*).
\end{equation}
%
%
Then:  

\item $(i)$ The mappings $\widetilde{\Gamma}_0$ and $\widetilde{\Gamma}_1$
are densely defined and closed. Moreover, the operator  $A_* := A^*\upharpoonright\dom(A_*)$ satisfies
      \begin{eqnarray}\label{3.42A}
\dom(A_*) := \dom(\widetilde{\Gamma}_0)\cap\dom(\wt{\Gamma}_1)=\dom(\wt{\Gamma}_0)=\dom(\wt{\Gamma}_1)\qquad \qquad \qquad \nonumber  \\
=\left\{f = \binom{f_1}{f_2}\in W^{1,2}({\R_+}\setminus X)\otimes
\mathbb
C^2:\{f_{j}(x_{n-1}+)\}^{\infty}_1,\{f_{j}(x_n-)\}^{\infty}_1\in
l^2(\mathbb N),\  j\in \{1,2\}\right\}.
  \end{eqnarray}

\item $(ii)$ The direct sum  $\widetilde{\Pi} := \bigoplus_{n=1}^\infty
\widetilde{\Pi}^{(n)} =\big\{\cH,\widetilde{\Gamma}_{0},
\widetilde{\Gamma}_{1}\big\}$
forms a $B$-generalized boundary triplet  for $D^*_X$ in the sense
of Definition \ref{def_II.2.1_generalized_bt}. In particular,
$\ran(\widetilde{\Gamma}_{0}) = \cH.$

\item $(iii)$ The transposed triplet $\wt{\Pi}^\top =
\{\cH,\wt{\Gamma}^\top_0,\wt{\Gamma}^\top_1\}:=\{\cH,\wt{\Gamma}_1,-\wt{\Gamma}_0\}$
also forms a $B$-generalized boundary triplet for $D_X^*$. In
particular, $\ran(\widetilde{\Gamma}_{1}) = \cH.$

\item $(iv)$ The triplet  $\widetilde{\Pi}$ is  an (ordinary) boundary
triplet for the operator $D_X^* = \bigoplus_{n=1}^\infty D_{n}^{*}$
if and only if   
$
0 < d_{*}(X)<d^{*}(X)<\infty. \quad 
$
%
%
\end{proposition}

\begin{proof}
(i)  By Definition \ref{def_III.1_01} and formula  \eqref{triple2}, the domain  of $\wt{\Gamma}_0$
is given by
    \begin{equation}\label{3.38}
\dom(\wt{\Gamma}_0)  = \left\{f = \binom{f_1}{f_2}\in
W^{1,2}({\R_+}\setminus X)\otimes \mathbb
C^2:\{f_{1}(x_{n-1}+)\}^{\infty}_1,\
 \{f_{2}(x_n-)\}^{\infty}_1\in l^2(\mathbb N)\right\}.
    \end{equation}
Since $d^*(X)<\infty$, it follows from \eqref{3.30AA} that
     \begin{equation}
\sum_{n\in\mathbb N}|f_{1}(x_n-)-f_{1}(x_{n-1}+)|^2\le d^*(X)
\sum_{n\in\mathbb N}d_n^{-1}|f_{1}(x_n-)-f_{1}(x_{n-1}+)|^2  \le d^*(X) \|f\|^2_{W^{1,2}({\R}_+\setminus X)} < \infty.
     \end{equation}
Combining this inequality with \eqref{3.38}, yields
$\{f_{1}(x_n-)\}^{\infty}_1\in l^2({\mathbb N})$. The inclusion
$\{f_{2}(x_{n-1}+)\}^{\infty}_1\in l^2({\mathbb N})$ is proved similarly.
Hence   $\dom(\wt{\Gamma}_0) = \dom(\widetilde{\Gamma}_0)\cap\dom(\wt{\Gamma}_1).$
The equality  $\dom(\wt{\Gamma}_1) = \dom(\widetilde{\Gamma}_0)\cap\dom(\wt{\Gamma}_1)$
is proved in much the same way.

(ii)  Due to (i)
$\ker\widetilde\Gamma_0\subset\dom(\widetilde\Gamma_1)=\dom(A_*)$.
Hence
\begin{equation}
A_{*0}:= A_*\lceil\ker\widetilde\Gamma_0=
A^*\lceil\ker\widetilde\Gamma_0 = \bigoplus^{\infty}_{n=1}D_{n,0}=A_0
=A_0^*,
\end{equation}
i.e. $A_{*0}=A_0$ is selfadjoint. The Green's identity
\eqref{III.1_06} is obviously satisfied for $f,g\in\dom(A_*)$ (see \eqref{III.1_06}). It
remains to show that $\ran(\Gamma_0)=\cH = l^2({\mathbb
N})\otimes\C^{2}$. This fact is immediate from \eqref{3.38} and
Proposition \ref{prop3.6}(iii).


(iii) The proof is similar to the proof of (ii).

(iv)  Let conditions \eqref{3.34} be satisfied. Then, by
Proposition \ref{prop3.6}(iii), $\dom(\pi_{\pm}) =
W^{1,2}({\R_+}\setminus X).$  Combining this fact with
\eqref{3.42A} we get $\dom(\wt{\Gamma}_j) =
W^{1,2}({\R_+}\setminus X)\otimes \mathbb C^2,\  j\in\{0,1\}$. Hence
Green's formula \eqref{II.2.1_03} holds for all $f,g\in
\dom(A^*).$

Next let us prove the surjectivity of the mapping
$\wt{\Gamma}=\{\wt{\Gamma}_0, \wt{\Gamma}_1\}.$ Let
$a_k=\{a_{kn}\}^{\infty}_{n=1}\in l^2(\mathbb N)$, $k\in
\{1,...,4\}.$   By Proposition \ref{prop3.6}(iii), there exist
$f_1, f_2\in W^{1,2}({\R_+}\setminus X)$ such that
    \begin{eqnarray*}
\pi_+(f_1) = \{f_{1}(x_{n-1}+)\}^{\infty}_{n=1} =
\{a_{1n}\}^{\infty}_{n=1}, \qquad \pi_-(f_1) =
\{f_{1}(x_{n}-)\}^{\infty}_{n=1} = \{a_{4n}\}^{\infty}_{n=1}, \\
\pi_+(f_2) = \{icf_{2}(x_{n-1}+)\}^{\infty}_{n=1} =
\{a_{3n}\}^{\infty}_{n=1}, \qquad \pi_-(f_2) =
\{icf_{2}(x_{n}-)\}^{\infty}_{n=1} = \{a_{2n}\}^{\infty}_{n=1}.
  \end{eqnarray*}
Combining these relations with \eqref{triple2NEW}, yields the
surjectivity of the  mapping $\wt{\Gamma}$.

Conversely let  $\dom(\wt{\Gamma}_j) = W^{1,2}({\mathbb R}_+\setminus
X)\otimes \C^2,\ j\in \{0,1\}.$
Then, by \eqref{triple2NEW}, $\dom(\pi_\pm) = W^{1,2}({\mathbb
R}_+\setminus X).$  Now Proposition \ref{prop3.6}(iii) yields the
condition $d_*(X)>0$.
    \end{proof}
   \begin{remark}
\item $(i)$  We emphasize the  difference between the trace mappings $\pi_{\pm}: W^{1,2}(\R_+\setminus X) \to l^2(\N)$ and $\pi: W^{1,2}({\mathbb R}_+)\to l^2(\N)$ (see \eqref{3.29}).
According to Proposition \ref{prop3.6}(iii) the
first one 
is surjective if and only if $d^{*}(X) < \infty.$   At the same
time, by Proposition \ref{prop3.5trace}, the second one is surjective
if and only if $0< d_*(X) < d^{*}(X) < \infty.$

\item $(ii)$  We emphasize that the  mapping $\wt{\Gamma}_j,\ j\in\{0,1\},$
is bounded, $\wt{\Gamma}_j \in[\mathfrak H_+,\cH]$, if and only if
$d_*(X)>0$. Indeed, it follows from \eqref{1.2} and  \eqref{triple2} that the estimate
    \begin{equation*}
\|\wt{\Gamma}_{0}^{(n)}f_n\|_{\mathbb C^2}\le\wt{C}_{0n}\left(\|D_n^*
f_n\|^2_{L^2[x_{n-1}, x_n]\otimes\mathbb C^2} + \|f_n\|^2_{L^2[x_{n-1},
x_n]\otimes\mathbb C^2}\right), \quad f_n\in \dom(D_n^*), \quad n\in\mathbb N,
     \end{equation*}
yields (in fact, is equivalent to) the estimate
    \begin{equation*}
|f_{1n}(x_{n-1}+)|^2 + |f_{2n}(x_n-)|^2 \le k_n
\left(\|f_{1n}\|^2_{W^{1,2}[x_{n-1}, x_n]} +
\|f_{2n}\|^2_{W^{1,2}[x_{n-1}, x_n]}\right), \quad f_{1n},
f_{2n}\in W^{1,2}[x_{n-1}, x_n].
\end{equation*}
Thus,  the mapping $\wt{\Gamma}_0 = \bigoplus_{n\in \N} \wt{\Gamma}_{0}^{(n)}$ is bounded, $\wt{\Gamma}_0 \in[\mathfrak H_+,\cH],$ if and only if $\sup_nk_n<\infty$.  In turn, according to the Sobolev embedding theorem, the latter is amount to saying that  $d_*(X)>0$.

This fact is similar to that for Schr\"odinger operator (cf.
\cite[Corollary 4.9]{KM}). It also shows that the condition
$\sup_n \|\wt{\Gamma}_{0}^{(n)}\| < \infty$
 is  only sufficient for a triplet $\wt \Pi= \bigoplus_{n\in
\N}\wt\Pi_n$ to form a $B$-generalized boundary triplet (cf. \cite[Proposition 3.6]{KM}).
         \end{remark}

To obtain an appropriate  boundary triplet for the operator $D_X^* =
\bigoplus_{n=1}^\infty D_{n}^{*}$ in the case $d_{*}(X)=0$ we
regularize the boundary triplets $\widetilde \Pi_n$ for $D_n^*,\
n\in \N,$ given by  \eqref{triple2}. To this end we apply
the regularization procedure proposed in Corollary
\ref{cor_III.2.2_02} (cf. formula \eqref{III.2.2_08}).
%
\begin{theorem}\label{th_bt_2} Let $X=\{x_{n}\}_{n=1}^\infty$ be
as above and  $d^*(X)<+\infty$. Define the mappings
$$\Gamma_j^{(n)}: W^{1,2}[x_{n-1},x_n]\otimes\C^2\to\C^2\,, \quad n\in \N\,,\quad j\in \{0,1\}\,,
$$
by setting
\begin{equation}\label{IV.1.1_12}
\Gamma_0^{(n)}f:=\left(\begin{array}{c}
                 \gd_n^{1/2}  f_{1}(x_{n-1}+)\\
                 i\,c\,\gd_n^{3/2}\sqrt{1+\frac{1}{c^{2}\gd_{n}^{2}}}\, f_{2}(x_{n}-)
                       \end{array}\right),
\end{equation}
 \begin{equation}\label{IV.1.1_12.1}
                       \Gamma_1^{(n)}f:=\left(\begin{array}{c}
                                                                           i\,c\,\gd_n^{-1/2}\,(f_{2}(x_{n-1}+)-f_{2}(x_{n}-))\\
                                                                           \gd_n^{-3/2}\left(1+\frac{1}{c^{2}\,\gd_{n}^{2}}\right)^{-1/2}
                                                                           (f_{1}(x_{n}-)-f_{1}(x_{n-1}+)-i\,c\,\gd_n\,f_2(x_{n}-))
                                                                 \end{array}\right).
\end{equation}
Then:
\item $(i)$ For any $n\in \N$,  $\Pi^{(n)}=\{\C^2,
\Gamma_0^{(n)},\Gamma_1^{(n)}\}$ is a boundary triplet for
$D_n^*$.

\item $(ii)$ The direct sum $\Pi:=\bigoplus_{n=1}^\infty\Pi^{(n)} =
\{\cH, \Gamma_0,\Gamma_1\}$ with $\cH = l^{2}(\N,\C^{2})$ and
${\Gamma}_{j} =\bigoplus_{n=1}^\infty {\Gamma}^{(n)}_{j},$
$j\in\{0,1\},$  is a boundary triplet for the operator $D_X^* =
\bigoplus_{n=1}^\infty D_{n}^{*}$.
       \end{theorem}
      \begin{proof}
According to Lemma \ref{2.3}(ii) 
 each operator $D_{n,0} = D_{n,0}^*$ has a gap
$(-\alpha_{n},\alpha_{n})\supset(-c^2/2,c^2/2)$. Hence the
symmetric operator  $D_X = \bigoplus_{n=1}^\infty D_{n}$ has a gap
$(-\alpha,\alpha) :=  \cap_n^\infty (-\alpha_{n},\alpha_{n})$.
Since $d^*(X)<\infty,$  it follows from \eqref{spetcom} that
$\alpha >c^2/2$, i.e. $(-\alpha,\alpha)\supset(-c^2/2, c^2/2)$.
Moreover, by \eqref{IV.1.1_09.2},
 \begin{equation}\label{3.33B}
 \wt M_{n}\left(\frac{c^2}{2}\right)=\left(\begin{array}{cc}
                           0 & 1\\
                           1 & d_{n}
                           \end{array}\right)\qquad \text{and} \qquad \wt M'_{n}\left(\frac{c^2}{2}\right)=\left(\begin{array}{cc}
                           d_n & d^2_n/2\\
                           d^2_n/2 & d_n/c^2 + d_n^3/3
                           \end{array}\right)\,.
 \end{equation}

Since $c^2/2 \in (0,\alpha),$  we can apply Corollary \ref{cor_III.2.2_02}
 to regularize the sequence of boundary triplets
$\wt{\Pi}^{(n)}=\big\{\C^{2},
\wt{\Gamma}_{0}^{(n)},\wt{\Gamma}_{1}^{(n)}\big\}$  for $D_n^*, \ n\in \N,$ defined by
\eqref{triple2}. Starting with  \eqref{3.33B} we define the
matrices $R_n =R_n^*$ and $Q_n = Q_n^*$ by setting
 \begin{equation}\label{IV.1.1_19}
 R_n :=\left(\begin{array}{cc}
                           \gd_n^{1/2} & 0\\
                   0 & \gd_n^{3/2}\sqrt{1+\frac{1}{c^{2}\,\gd_{n}^{2}}}
                           \end{array}\right)\,,
                           \qquad
Q_n :=\wt M_n\left(\frac{c^{2}}{2}\right)=\left(\begin{array}{cc}
                           0 & 1\\
                           1 & \gd_n
                           \end{array}\right), \qquad n \in \N.
\end{equation}
Next we define  a new sequence of boundary triplets $\Pi_n=\{\mathbb
C^{2},\Gamma^{(n)}_0,\Gamma^{(n)}_1\}$ by formulas
\eqref{III.2.2_08},
%
%
   \begin{equation}\label{IV.1.1_09AAA}
\Gamma^{(n)}_0 := R_n{\widetilde\Gamma}_0^{(n)},\quad
 \Gamma^{(n)}_1 :=
R^{-1}_n({\widetilde\Gamma}_1^{(n)}-Q_n{\widetilde\Gamma}_0^{(n)}),
\qquad n \in \N.
       \end{equation}
%
Clearly, the corresponding Weyl function is
      \begin{equation}\label{IV.1.1_09A}
M_n(z) = R^{-1}_n(\widetilde{M}_n(z)-Q_n)R^{-1}_n,\qquad n \in \N.
       \end{equation}
Let us check that the family $\{M_n(\cdot)\}_{n=1}^\infty$ of the
Weyl functions satisfies conditions \eqref{III.2.2_11} of Corollary
\ref{cor_III.2.2_02}.  Indeed, by \eqref{IV.1.1_19},
$M_n({c^{2}}/{2}) = 0.$ Moreover, combining \eqref{IV.1.1_09A}
with  \eqref{IV.1.1_19} and \eqref{3.33B} we get
   \begin{equation}\label{IV.1.1_08B}
\begin{split}
M_{n}'\left(\frac{c^{2}}{2}\right)&=R^{-1}_n{\widetilde M}'_n\left(\frac{c^{2}}{2}\right)R^{-1}_n =R^{-1}_n\begin{pmatrix}
\gd_n & \gd_n^2/2\\
\gd_n^2/2 & \gd_n/c^{2}+\gd_n^3/3
\end{pmatrix}R^{-1}_n =\\
&=\begin{pmatrix}
1 & \frac{1}{2}\left(1+\frac{1}{c^{2}\,\gd_n^2}\right)^{-1/2}\\
\frac{1}{2}\left(1+\frac{1}{c^{2}\,\gd_n^2}\right)^{-1/2} &
\frac{1}{3}\frac{3 + c^2\gd_n^2}{1 + c^{2}\,\gd_n^2}
\end{pmatrix}, \qquad n\in \N.
\end{split}
  \end{equation}
Hence $\sup_{n\in \N}\|M_{n}'(c^{2}/{2})\| <\infty$ and the first
condition in  \eqref{III.2.2_11} is satisfied. Further,
   \begin{equation}\label{IV.1.1_08BB}
\begin{split}
(M_{n}'(c^{2}/{2}))^{-1} = R_n{\widetilde M}'_n(c^{2}/{2})^{-1}R_n
= \frac{1}{\Delta(c^2/2)}
    \begin{pmatrix}
\frac{1}{3}\frac{3 + c^2\gd_n^2}{1 + c^{2}\,\gd_n^2} & -\frac{1}{2}\left(1+\frac{1}{c^{2}\,\gd_n^2}\right)^{-1/2}\\
-\frac{1}{2}\left(1+\frac{1}{c^{2}\,\gd_n^2}\right)^{-1/2} & 1
\end{pmatrix},
\end{split}
  \end{equation}
where ${\Delta(c^2/2)}=12(1+c^2 d^2_n)(12+c^2 d^2_n)^{-1}$. Hence
$\sup_{n\in{\N}}\|\bigl(M'_n(c^2/2)\bigr)^{-1}\|<\infty$ and the
second condition in  \eqref{III.2.2_11} is satisfied too. Thus, by
Corollary \ref{cor_III.2.2_02},  the direct sum
$\bigoplus_{n=1}^\infty\Pi^{(n)}$ is an ordinary  boundary triplet for the
operator $D_X^*$.

To complete the proof it remains to note that   the mappings $\Gamma^{(n)}_0$
and $\Gamma^{(n)}_1$, $n\in \N,$
defined by \eqref{IV.1.1_09AAA}  coincide with the  mappings given by  \eqref{IV.1.1_12}, \eqref{IV.1.1_12.1}.
\end{proof}
   \begin{remark}
It follows from \eqref{3.33B} that both  conditions
  $$
\sup_{n\in{\mathbb
N}}\|\wt{M}_n(c^2/2)\|<\infty \qquad \text{and} \qquad \sup_{n\in{\mathbb
N}}\|\wt{M}'_n(c^2/2)\|<\infty
$$
are  satisfied for any discrete
sequence $X=\{x_n\}_{n=1}^\infty$ whenever  $d^*(X)<\infty$.
Applying Theorem \ref{th_III.2.2_01}(i)  we obtain an alternative proof of Proposition
\ref{prop3.5direcsum}(ii).

Further, it is easily seen that $\wt M_n^{-1}\left(\frac{c^{2}}{2}\right)=\left(\begin{array}{cc}
                           -\gd_n & 1\\
                           1 & 0
                           \end{array}\right)$
and
    \begin{equation*}
\bigl(\wt {M}^{-1}_n\bigr)'\left(\frac{c^2}{2}\right)= \wt{M}^{-1}_n\left(\frac{c^2}{2}\right) \wt{M}'_n\left(\frac{c^2}{2}\right) \wt{M}^{-1}_n\left(\frac{c^2}{2}\right) =
\begin{pmatrix}
d^3_n/3 + d_n/c^2 & -d^2_n/2\\
-d^2_n/2 & d_n
\end{pmatrix}, \quad n\in \N.
    \end{equation*}
Thus, the sequence $\{-\wt {M}^{-1}_n(\cdot)\}_{n\in \N}$ satisfies  conditions \eqref{III.2.2_02} at the point $a=c^2/2$ provided that   $d^*(X)<\infty$.  Again, by  Theorem \ref{th_III.2.2_01}(i), the direct sum  $\widetilde{\Pi} := \bigoplus_{n=1}^\infty
\widetilde{\Pi}^{(n)}$ forms a $B$-generalized boundary triplet  for $D^*_X$ in the sense
of Definition \ref{def_II.2.1_generalized_bt}. These reasonings
give an alternative proof of  Proposition \ref{prop3.5direcsum}(iii).

At the same time, condition
$\sup_{n\in{\mathbb N}}\|\bigl(\wt{M}'_n(c^2/2)\bigr)^{-1}\|<\infty$
is satisfied if and only if $d_*(X)>0$. Hence Theorem
\ref{th_III.2.2_01}(ii)  gives an alternative proof of Proposition
\ref{prop3.5direcsum}(iv).
On the other hand, this example
shows that condition \eqref{III.2.2_02NEW}  of Theorem
\ref{th_III.2.2_01}  is not implied by conditions
\eqref{III.2.2_02} (cf.  Example \ref{examp2.14}(i)).
   \end{remark}
\begin{corollary}\label{s.a.Dirac}
Let $\Pi=\{\cH,\Gamma_0,\Gamma_1\}$ be a boundary triplet for
the operator $D_X^{*}$  defined in Theorem \ref{th_bt_2}, i.e.
 $\Pi=\bigoplus_{n=1}^{\infty} {\Pi}^{(n)}$.   
Then:
%
%
\item $(i)$  The set of closed proper extensions of $D_X$ is
parameterized as  follows
  \begin{equation}\label{31d}
\widetilde{D_X} =
D_{X,\Theta}:=D_X^{*}\upharpoonright\dom(D_{X,\Theta}),
\quad\dom(D_{X,\Theta})= \left\{f\in W^{1,2}(\cI \setminus X)\otimes \C^2:
\{\Gamma_{0}f,\Gamma_{1}f\}\in\Theta\right\},
  \end{equation}
\noindent where $\Theta\in \tilde \cC(\cH)\setminus \{\{0\}\cup
\cH\oplus\cH$\}.
\item $(ii)$   $D_{X,\Theta}$ is symmetric (self-adjoint) if and only if
so is ${\Theta}$. Moreover,  $n_{\pm}(D_{X,\Theta}) =
n_{\pm}({\Theta}).$
\item $(iii)$   $D_{X,\Theta}$  is $m$-dissipative
($m$-accumulative) if and only if so is $\Theta$.
\item $(iv)$   $\widetilde{D_X}=D_{X,\Theta}$ is disjoint with
$D_{X,0}:=D_X^{*}\upharpoonright\ker\Gamma_{0}$ if and only if
$\Theta$ is a closed operator. In this case relation (\ref{31d})
takes the form
\begin{equation}\label{32d}
\widetilde{D_X}=D_{X,\Theta}:=D^{*}\upharpoonright\ker
\left(\Gamma_{1}-\Theta\Gamma_{0}\right).
\end{equation}
Moreover, $\widetilde{D_X}=D_{X,\Theta}$ and $D_{X,0}$ are
transversal if and only if (\ref{32d}) holds with $\Theta\in
\left[\cH\right]$.
\end{corollary}
  \begin{proof}
According to \eqref{3.26} $D_X^{*} =  \bigoplus^{\infty}_{n=1}D_n^{*}$, hence
\[\dom(D_X^{*}) =  \bigoplus^{\infty}_{n=1}\dom(D_n^{*})  = \bigoplus^{\infty}_{n=1} W^{1,2}[x_{n-1}, x_n]\otimes \C^2
= W^{1,2}(\cI \setminus X)\otimes \C^2.
\]
One completes the  proof by applying  Proposition \ref{prop_II.1.2_01}.
   \end{proof}
   \begin{remark}\label{RMD} 
Consider  a Dirac operator with point  interactions supported on the set $X=\{x_k\}_{k\in I}\subset \cI=(a,b)$, $x_{k-1}<x_k,\ k\in I,$ where either  $I =\N$ or $I =\Z$. Moreover,  we  assume as  usual that   $\lim_{k\to\infty}x_{k} = b\le \infty$  and  $\lim_{k\to-\infty}x_{k} = a \ge -\infty$ in the second case.
Now in place of \eqref{3.26} the minimal operator is  $D_X:=\bigoplus_{n\in I} D_n$.\par
To investigate the non-relativistic limit on the line we also will consider the operators
  \begin{equation}\label{3.26B}
D_{a-}\bigoplus D_X \bigoplus D_{b+}\,,
   \end{equation}
where the first (resp. the third)  summand is missing whenever  $a =-\infty$ (resp. $b = \infty$).
The corresponding maximal operators are given  by
$D_X^*= \bigoplus_{n\in I}D^{*}_{n}$ and $D^{*}_{a-}\bigoplus D_X^* \bigoplus D^{*}_{b+}$, respectively.\par
The appropriate  boundary triplets for the maximal operators are of the form
$\bigoplus_{n\in I}{\Pi}^{(n)}$ and $\Pi^{(a-)}\bigoplus\left(\bigoplus_{n\in I}^{\infty}{\Pi}^{(n)}\right)\bigoplus\Pi^{(b+)}$, respectively.
Here  $\Pi^{(a-)}$ and $\Pi^{(b+)}$ are the boundary triplets defined in
Lemmas  \ref{2.3-} and \ref{2.3+}, respectively.
Using these boundary triplets one parameterizes the set of proper extensions of the operators  $D_X$
and  \eqref{3.26B}  in just the same way as in Corollary \ref{s.a.Dirac}.
  \end{remark}

\section{Non-relativistic limit}
Here we investigate the non-relativistic resolvent limit
of maximal dissipative (accumulative) extension $D^c_{X,\Theta}$ defined
by (\ref{31d}).
\noindent To this end  we consider the
operator $-{d^2}/{dx^2}$ with point interactions
on a discrete set and following \cite{KM} describe the
corresponding boundary triplets, Weyl functions, etc.

\subsection{Boundary triplets for  Schr\"{o}dinger building blocks}
First we present a boundary triplet for the maximal operator  $-\frac{d^2}{dx^2}$ on a finite  interval.  Let $H_{n}$ denote
the minimal  operator associated with the differential expression
$-\frac{d^2}{dx^2}$ in $L^{2}[x_{n-1},x_{n}]$ by
%
%
  \begin{equation}
H_{n} := -\frac{d^2}{dx^2}\upharpoonright\dom(H_n),
\qquad\dom(H_{n})=W^{2,2}_{0}[x_{n-1},x_{n}],\qquad n\in \N.
  \end{equation}
%
    \begin{lemma}\label{2.3.2.H}
The operator $H_{n}$ is  symmetric in $L^{2}[x_{n-1},x_{n}]$
 with the deficiency indices
$n_{\pm}(H_{n})= 2$. Its adjoint $H_n^*$ is given by
$$
H_{n}^{*}=-\frac{d^2}{dx^2}\upharpoonright\dom(H_{n}^{*}),
\qquad\dom(H_{n}^{*}) = W^{2,2}[x_{n-1},x_{n}],
$$
and the defect subspace $\mathfrak N_z = \text{\rm ker}(H^*_{n}-z)$ is spanned by
the functions $f_{n,H}^\pm(\cdot),$
  \begin{equation}\label{3.43}
f_{n,H}^{\pm}(z)(x):=e^{\pm i\, \sqrt z\,x}\,,\qquad\text{\rm Im}(\sqrt z)\ge
0\,.
  \end{equation}
Moreover, the following holds.
%
\item $(i)$  The triplet $\widetilde\Pi_H^{(n)}=\big\{\C^{2},
\widetilde{\Gamma}_{0,H}^{(n)},\widetilde{\Gamma}_{1,H}^{(n)}\big\}$
where
\begin{equation}\label{triple2.H}
\widetilde{\Gamma}_{0,H}^{(n)}f:=
\left(\begin{array}{c}
                 f(x_{n-1}+)\\
                 f'(x_{n}-)
                       \end{array}\right)\,,\qquad
\widetilde{\Gamma}_{1,H}^{(n)}f:=\left(\begin{array}{c}
                                                                           f'(x_{n-1}+)\\
                                                                            f(x_{n}-)
                                                                          \end{array}\right)\,,
\end{equation}
forms a boundary triplet for $H^{*}_{n}$.
\item $(ii)$  The spectrum of the operator
${H}_{n,0}:=H_{n}^{*}\upharpoonright\ker\widetilde{\Gamma}^{(n)}_{0,H}$ is discrete,
  \begin{equation}
\sigma({H}_{n,0})=\sigma_d({H}_{n,0})=
\left\{\frac{\pi^2}{d^2_n}\,\left(j+\frac12\right)^2\,,\quad j\in
\{0\}\cup \N  \right\}.
  \end{equation}
\item $(iii)$  The $\gamma$-field $\wt\gamma_n(\cdot):\C^2\to
  L^2[x_{n-1},x_{n}]$,   corresponding to the triplet
  $\wt\Pi^{(n)}$ is given  by 
   \begin{align}\label{1.21+.H2}
\wt\gamma_n(z) \begin{pmatrix} w_1\\w_2
\end{pmatrix} = &\frac{1}{\cos(d_n\sqrt z)}\left(
\cos(\sqrt z\,(x_{n}-x))\,,\  -\frac{\sin(\sqrt z
(x_{n-1}-x))}{\sqrt z}\right)\begin{pmatrix} w_1\\w_2
\end{pmatrix} \nonumber  \\
=&\frac{1}{\cos(d_n\sqrt z)}\left(w_1\cos(\sqrt z\,(x_{n}-x))
-w_2\,\frac{\sin(\sqrt z (x_{n-1}-x))}{\sqrt z}\right) \,,\quad
z\in\rho(H_{n,0}).
\end{align}
\item $(iv)$   The  Weyl function $\widetilde{M}_{n,H}(\cdot)$,  
corresponding to the triplet $\wt\Pi_H^{(n)}$ is
           \begin{equation}\label{IV.1.1_09.2.H}
 \widetilde{M}_{n,H}(z)=\frac{1}{\cos(d_n\,\sqrt z)}\left(\begin{array}{cc}\sqrt z\,\sin(d_n\,\sqrt z)
   & 1 \\1 & z^{-1/2}\sin(d_n\,\sqrt z)
\end{array}\right),\qquad z\in\rho(H_{n,0}).
         \end{equation}
%
\end{lemma}

Next we present  a boundary triplet for the operator   
$-\frac{d^2}{dx^2}$ on the half-line. Denote by  $H_{a-}$   the minimal
operator associated with  the differential  expression $-\frac{d^2}{dx^2}$
on $L^{2}(\R_a^-)$ by
\begin{equation}
H_{a-}=-\frac{d^2}{dx^2}\upharpoonright\dom(H_{a-}),
\qquad\dom(H_{a-}) = W^{2,2}_{0}(\R_a^-).
\end{equation}
Similarly,  $H_{b+}$ denotes  the minimal operator generated by
the expression $-d^2/dx^2$  on $L^{2}(\R_b^+)$,
  \begin{equation}
H_{b+}=-\frac{d^2}{dx^2}\upharpoonright\dom(H_{b+}),
\qquad\dom(H_{b+})=W^{2,2}_{0}(\R_b^+)\,.
\end{equation}
   \begin{lemma}\label{2.3-.H}
The operator $H_{a-}$ is  symmetric  with the deficiency indices
$n_{\pm}(H_{a-})= 1$. Its adjoint $H^*_{a-}$ is given by
$$
H_{a-}^{*} := (H_{a-})^{*} =-\frac{d^2}{dx^2}
\upharpoonright\dom(H_{a-}^{*})\,,
\qquad\dom(H_{a-}^{*})=W^{2,2}(\R_a^{-})\,,
$$
and the defect subspace $\mathfrak N_z = \text{\rm ker}(H^*_{a-}-z)$ is spanned by
the vector function
   \begin{equation}\label{3.51}
f_{a}^{-}(x, z):=e^{- i\, \sqrt z\,x}\,, \qquad\text{\rm Im}(\sqrt
z)> 0.
\end{equation}
Moreover, the following holds:
%
%
 \item $(i)$  The triplet $\Pi_H^{(a-)}=\big\{\C^{2},
{\Gamma}_{0,H}^{(a-)},{\Gamma}_{1,H}^{(a-)}\big\}$, where
  \begin{equation}\label{triple.a.H}
{\Gamma}_{0,H}^{(a-)}f:=f'(a-)\,,\qquad
{\Gamma}_{1,H}^{(a-)}f:=f(a-)\,,
  \end{equation}
forms a boundary triplet for the operator $H^{*}_{a-}$.

\item $(ii)$   The spectrum of the operator
${H}_{a-,0}:=H_{a-}^{*}\upharpoonright\ker{\Gamma}^{(a-)}_{0,H}=({H}_{a-,0})^*$
is absolutely continuous,
\begin{equation}
\sigma({H}_{a-,0})=\sigma_{ac}({H}_{a-,0})=[0,+\infty)\,.
\end{equation}
\item $(iii)$  The corresponding $\gamma$-field
$\gamma_{a-,H}(\cdot):\C\to L^2(\R_a^{-})$, is
  \begin{equation}\label{gamma.a.H}
\gamma_{a-,H}(z)w =w\,\frac{i}{\sqrt z}\,e^{i\, \sqrt
z\,a}f_{a,H}^-(\cdot, z)\,, \quad w\in \C_+,  \quad  z\in\rho(H_{a-,0}).
  \end{equation}
\item $(iv)$  The  Weyl function ${M}_{a-,H}(\cdot)$ corresponding
to the triplet $\Pi_H^{(a-)}$ is
           \begin{equation}\label{IV.1.1_09.a.H}
{M}_{a-,H}(z)=\frac{i}{\sqrt z},\qquad  z\in\rho(H_{a-,0}).
         \end{equation}
%
%
\end{lemma}
     \begin{lemma}\label{2.3+.H}
The operator $H_{b+}$ is  symmetric  with the  deficiency indices
$n_{\pm}(H_{b+})= 1$. Its adjoint $H^*_{b+}$ is given by
$$
H_{b+}^{*} := (H_{b+})^{*} =-\frac{d^2}{dx^2}
\upharpoonright\dom(H_{b+}^{*})\,,
\qquad\dom(H_{b+}^{*})=W^{2,2}(\R_b^{+})\,,
$$ and the defect subspace
$\mathfrak N_z = \text{\rm ker}(H^*_{b+}-z)$ is spanned by the vector function
  \begin{equation}\label{3.57}
f_{b, H}^{+}(x, z) := e^{i\, \sqrt z\,x}\,,\qquad\text{\rm Im}(\sqrt z)>0\,.
  \end{equation}
Moreover, the following holds:
%
\item $(i)$ The triplet   $\Pi_H^{(b+)}=\big\{\C^{2},
{\Gamma}_{0,H}^{(b+)},{\Gamma}_{1,H}^{(b+)}\big\}$, where
  \begin{equation}\label{triple.b.Hb}
{\Gamma}_{0,H}^{(b+)}f:=f(b+)\,,\qquad
{\Gamma}_{1,H}^{(b+)}f:=f'(b+),
   \end{equation}
forms a boundary triplet for the operator $H^{*}_{b+}$.
\item $(ii)$  The spectrum of the operator ${H}_{b+,0} :=
H_{b+,H}^{*}\upharpoonright\ker{\Gamma}^{(b+)}_{0,H} =
{H}_{b+,0}^*$ is absolutely continuous,
  \begin{equation}
\sigma({H}_{b+,0})=\sigma_{ac}({H}_{b+,0})=[0,+\infty)\,.
  \end{equation}
\item $(iii)$   The  $\gamma$-field
$\gamma_{b+,H}(\cdot):\C\to L^2(\R_b^{+})$ corresponding to the
triplet $\Pi_H^{(b+)}$ is given by
  \begin{equation}\label{gamma.b.H}
\gamma_{b+}(z)w = w\,e^{-i\, \sqrt z\,b}f_{b,z, H}^{+}, \quad w\in \C, \quad
z\in\rho(H_{b+,0}).
  \end{equation}
\item $(iv)$ The  Weyl function $ {M}_{b+,H}(\cdot)$
corresponding to the triplet $\Pi_H^{(b+)}$ is  
           \begin{equation}\label{IV.1.1_09.2.Hb}
{M}_{b+,H}(z)=i\,\sqrt z, \qquad z\in\rho(H_{b+,0}).
         \end{equation}
%
%
\end{lemma}
\subsection{Boundary triplet for Schr\"{o}dinger operators with point interactions}

Let $X=\{x_n\}_{n=1}^\infty$, $a=x_0,\  x_{n-1}<x_n,$ be a
discrete set as in Section 3.2. We define $H_X$ by
\begin{equation}\label{3.28.H}
H_X:=\bigoplus_{n\in\N} H_n
\,.
\end{equation}
Now we are ready to construct a boundary triplet for the operator
$H_X^{*}$.
    \begin{proposition}(\cite{KM})
Let  $\widetilde\Pi_H^{(n)}=\big\{\C^{2},
\widetilde{\Gamma}_{0,H}^{(n)},\widetilde{\Gamma}_{1,H}^{(n)}\big\}$
be the boundary triplet for the operator $H_n^*,\ n\in \N,$
defined in \eqref{triple2.H}. Then the  direct sum
    \begin{equation}\label{3.30.H}
\wt\Pi= \bigoplus_{n=1}^\infty\widetilde{\Pi}_H^{(n)}
=\big\{\cH,\widetilde{\Gamma}_{0,H},
\widetilde{\Gamma}_{1,H}\big\}\,,\quad\cH =l_{2}(\N)\otimes
\C^2\,,\quad \widetilde{\Gamma}_{j,H}=\bigoplus_{n=1}^\infty
\widetilde{\Gamma}^{(n)}_{j,H}\,,\quad j\in\{0,1\}\,,
   \end{equation}
forms a boundary triplet for the operator $H_X^*=
\bigoplus_{n=1}^\infty H_{n}^{*}$ if and only if $0<d_{*}(X)<
d^{*}(X)<\infty$.
 \end{proposition}

In the case $d_{*}(X)=0$ a boundary triplet for the operator
$H_X^{*}$  was constructed in \cite[Theorems 4.1, 4.7]{KM} by
applying to triplets
 $\widetilde{\Pi}_H^{(n)}$  the regularization procedure
described in Corollary \ref{cor_III.2.2_02}.
    \begin{theorem}(\cite[Theorem 4.7]{KM})\label{th_bt_2.H}
Assume that $d^*(X)<+\infty$ and define the mappings $
\Gamma_{j,H}^{(n)}: W^{2,2}[x_{n-1},x_n]\to\C^2\,,\  n\in \N\,,
\  j\in \{0,1\},$ by setting  
%
%
%
%
%
\begin{equation}\label{IV.1.1_06.H.2}
\Gamma_{0,H}^{(n)}f
:=\left(\begin{array}{c}
                  d_n^{1/2}f(x_{n-1}+)\\
                 d_n^{3/2}  f'(x_{n}-)
                       \end{array}\right),
\end{equation}
\begin{equation}
\label{IV.1.1_06_1.H.2}
\Gamma_{1,H}^{(n)}f:=
\left(\begin{array}{c}
                                    d_n^{-1/2} (f'(x_{n-1}+)-f'(x_{n}-))\\
                                   d_n^{-3/2} (f(x_{n}-)-f(x_{n-1}+))-d_n^{-1/2}f'(x_{n}-)
                                  \end{array}\right).
\end{equation}
Then:
%
%
\item $(i)$  For any $n\in \N$,  $\Pi_H^{(n)}=\{\C^2,
\Gamma_{0,H}^{(n)},\Gamma_{1,H}^{(n)}\}$ is a boundary triplet for
$H_{n}^*$.

\item $(ii)$ The direct sum $\Pi = \bigoplus_{n=1}^\infty\Pi_H^{(n)} =
\{\cH,\Gamma_{0,H},\Gamma_{1,H}\}$
with $\cH = l^{2}(\N, \C^{2})$ and ${\Gamma}_{j,H}
=\bigoplus_{n=1}^\infty {\Gamma}^{(n)}_{j,H},$ $j\in\{0,1\},$  is
a boundary triplet for the operator $H_X^* = \bigoplus_{n=1}^\infty
H_{n}^*$.
      \end{theorem}
%
%
 \begin{corollary}\label{s.a.H}
Let $\Pi_H = \{\cH,\Gamma_{0,H},\Gamma_{1,H}\}
:=\bigoplus_{n=1}^\infty \Pi_H^{(n)}$ be the  boundary triplet
for $H^{*}_X$ defined in Theorem \ref{th_bt_2.H}.
Then:
%
\item $(i)$  The set of closed proper extensions of $H_X$ is
parameterized as follows:
   \begin{equation}\label{31d.H}
\widetilde{H_X} = H_{X,\Theta} := H_X^{*}
\upharpoonright\dom(H_{X,\Theta}), \quad
\dom(H_{X,\Theta})=\left\{f\in W^{2,2}(\cI\setminus X):\{\Gamma_{0,H}f,
\Gamma_{1,H}f\}\in\Theta\right\},
  \end{equation}
\noindent where $\Theta\in \tilde \cC(\cH)\setminus \{\{0\}\cup
\cH\oplus\cH\}$.
\item $(ii)$ $H_{X,\Theta}$ is symmetric (self-adjoint)
if and only if
so is ${\Theta}$. Moreover,  $n_{\pm}(H_{X,\Theta}) =
n_{\pm}({\Theta}).$
\item $(iii)$   $H_{X,\Theta}$  is $m$-dissipative
($m$-accumulative) if and only if so is $\Theta$.
\item $(iv)$   $\widetilde{H_X}=H_{X,\Theta}$ is disjoint with
$H_{X,0}:=H_X^{*} \upharpoonright\ker\Gamma_{0,H} (= H_{X,0}^*)$
if and only if $\Theta$ is a closed operator. In this case
relation (\ref{31d.H}) takes the form
  \begin{equation}\label{32d.H}
\widetilde{H_X}=H_{X,\Theta}:=H_X^{*}\upharpoonright
\ker\left(\Gamma_{1,H}-\Theta\Gamma_{0,H}\right).
  \end{equation}
Moreover,  $\widetilde{H_X}=H_{X,\Theta}$ and $H_{X,0}$ 
are transversal if and only if (\ref{32d.H}) holds with
$\Theta\in \left[\cH\right]$.
%
\end{corollary}
\begin{remark}\label{RMH}
Consider  Schr\"odinger operator with point  interactions supported on the set $X=\{x_k\}_{k\in I}\subset \cI=(a,b)$, $x_{k-1}<x_k,\ k\in I,$ where either  $I =\N$ or $I =\Z$. Moreover,  we  assume as  usual that   $\lim_{k\to\infty}x_{k} = b\le \infty$  and  $\lim_{k\to-\infty}x_{k} = a \ge -\infty$ if $I =\Z$.
Now in place of \eqref{3.28.H} the minimal operator is  $H_X:=\bigoplus_{n\in I} H_n$.
In the next section  we  will also use  the operators
  \begin{equation}\label{3.26H}
H_{a-}\bigoplus H_X \bigoplus H_{b+}\,,
   \end{equation}
where the first (resp. the third)  summand is missing whenever  $a =-\infty$ (resp. $b = \infty$).
The corresponding maximal operators are given  by
$H_X^*= \bigoplus_{n\in I}H^{*}_{n}$ and $H^{*}_{a-}\bigoplus H_X^* \bigoplus H^{*}_{b+}$, respectively.\par
The appropriate  boundary triplets for the maximal operators are of the form
$\bigoplus_{n\in I}{\Pi}_{H}^{(n)}$ and $\Pi_{H}^{(a-)}\bigoplus\left(\bigoplus_{n\in I}^{\infty}{\Pi}_{H}^{(n)}\right)\bigoplus\Pi_{H}^{(b+)}$, respectively.
Here  $\Pi_{H}^{(a-)}$ and $\Pi_{H}^{(b+)}$ are the boundary triplets defined in
Lemmata \ref{2.3-.H} and \ref{2.3+.H}.
Using these boundary triplets one parameterizes the set of proper extensions of the operators  $H_X$
and  \eqref{3.26H}  in just the same way as in Corollary \ref{s.a.H}.
\end{remark}

\subsection{The non-relativistic  limit for general realizations}

 In this section we
investigate  the non-relativistic limit of any
$m$-accumulative,  $m$-dissipative (in particular self-adjoint)
extension $\widetilde D_X$ of $D_X$. We confine ourself to the case
of the  half-line $\R_a = (a, +\infty)$ only, i.e. assume that $b=+\infty$.  The cases of Dirac operators either in $L^{2}(-\infty,b)\otimes\C^{2}$ or in $L^{2}(\R)\otimes\C^{2}$, with $b<+\infty$, can be treated similarly by using
the boundary triplets discussed  in Remarks \ref{RMD} and \ref{RMH}.\par

Denote by   $X=\{x_{n}\}_{n=1}^\infty$ the discrete subset of  $\R_a$,
$x_0=a$ (cf. \eqref{3.26AAA}).  We will use  the boundary triplet
$\Pi=\{\cH,\Gamma_0,\Gamma_1\}=
 \bigoplus_{n=1}^{\infty}{\Pi}^{(n)}
 $
defined in  Theorem \ref{th_bt_2}  
 where ${\Pi}^{(n)}$ is given  by (\ref{IV.1.1_12}) and
(\ref{IV.1.1_12.1}).
In what follows we equip all objects
related to the Dirac operators by index $c$ to exhibit
dependence on the velocity of light. For instance,
 we write $D_X^{c}, D_n^{c}, M_{n,c}(\cdot)$ in place of $D_X, D_n, M_{n}(\cdot)$.
   \begin{theorem}\label{2.5}
Assume that $d^*(X)<+\infty$,  $\Pi_D=\{\cH,\Gamma_0,\Gamma_1\}$
and $\Pi_H=\{\cH,\Gamma_{0,H},\Gamma_{1,H}\}$ are the  boundary
triplets for $(D_X^c)^*$ and $H_{X}^*$  defined in Corollaries
\ref{s.a.Dirac} and \ref{s.a.H}, respectively.  Let also
$\widetilde D_X^c$ and $\widetilde H_{X}$ be
$m$-accumulative ($m$-dissipative), in particular, selfadjoint,
extensions of  $D_X^c$ and  $H_X$, respectively, and let
$\Theta_c$ and $\Theta$ be the corresponding boundary relations in
the boundary triplets  $\Pi_D$ and $\Pi_H$, i.e. $\widetilde D_X^c
= {D_{X,\Theta_c}^c}$ and $\widetilde H_{X} =  H_{X,\Theta}$
according to formulae  \eqref{31d} and \eqref{31d.H},
respectively. If  $l^2_0(\N, \C^2)$ is a core for $\Theta$,
 $l^2_0(\N, \C^2)\subset \cap_{c>1} \dom(\Theta_c)\cap
\dom(\Theta)$ and
%
%
\begin{equation}\label{1.23.0}
\lim_{c\to +\infty}\Theta_c h =\Theta h, \qquad h\in l^2_0(\N,
\C^2),
  \end{equation}
%
%
then
  \begin{equation}\label{1.23}
s-\lim_{c\to
  +\infty}\left(D_{X,\Theta_c}^c-(z+{c^{2}}/{2})\right)^{-1}=
(H_{X,\Theta}-z)^{-1}\bigotimes \left(\begin{array}{cc} 1 & 0 \\0 &
0\end{array}\right), \qquad z\in \C_+ \, (z\in \C_-).
  \end{equation}
   \end{theorem}
\begin{proof} 
(i) First we investigate the limit of the Weyl functions
$M_{n,c}(\cdot)$ given by \eqref{IV.1.1_09A}.  It follows from
\eqref{1.5} and \eqref{1.6} that
    \begin{equation}\label{limk}
\lim_{c\to+\infty}\, k(z+c^2/2)=\lim_{c\to+\infty}\,
c\,k_1(z+c^2/2)=\sqrt z\,,\quad z\in \C_{\pm}.
   \end{equation}
Next we find  the  limits as $c\to \infty$  of the basis defect
vectors defined by \eqref{3.7}, \eqref{3.17} and \eqref{3.23},
respectively. Taking into account \eqref{limk} and relation
\eqref{3.43}
we obtain
$$
\lim_{c\to+\infty}\, f^\pm_n(\cdot, z+c^2/2)
=\begin{pmatrix}f_{n,H}^\pm(\cdot, z)\\0\end{pmatrix}\,,
$$
where the convergence is in the
Hilbert spaces $L^2(x_{n-1},x_{n})\otimes \C^2$. 
%
\par
According to  \eqref{IV.1.1_09A} and \cite[equation (83)]{KM}
  \begin{equation}\label{3.81}
M_{n,c}(z) = R_n(c)^{-1}(\wt M_{n,c}(z)-Q_n)R_n(c)^{-1}\,,\qquad
M_{n,H}(z)= R_{n,H}^{-1}(\wt M_{n,H}(z)-Q_n)R_{n,H}^{-1}\,,
    \end{equation}
where $\wt M_{n,c}(\cdot)$ and $\wt M_{n,H}(\cdot)$ are defined
by \eqref{IV.1.1_09.2} and \eqref{IV.1.1_09.2.H},  respectively.
Besides, $R_n(c):= R_n$ and $R_{n,H}$ are defined by
\eqref{IV.1.1_19} and \cite[formula (94)]{KM}, respectively, i.e.
      \begin{equation}\label{IV.1.1_19H}
 R_n(c) := R_n= \left(\begin{array}{cc}
                           \gd_n^{1/2} & 0\\
                   0 & \gd_n^{3/2}\sqrt{1+\frac{1}{c^{2}\,\gd_{n}^{2}}}
                           \end{array}\right)\quad \text{and}\qquad
 R_{n,H}:=\left(\begin{array}{cc}
                           \gd_n^{1/2} & 0\\
                   0 & \gd_n^{3/2}
                           \end{array}\right)\,, \qquad n\in \N.
\end{equation}
 Clearly,
     \begin{equation}\label{3.84}
(1 + c^{-2}d^{-2}_n)^{-1}\to 1\qquad \text{as}\qquad c\to\infty.
  \end{equation}
It follows from   \eqref{IV.1.1_19H}  and  \eqref{3.84} that
$\lim_{c\to\infty}R_n(c) = R_{n,H},\ n\in \N,$ and
   \begin{equation}\label{3.85}
 R_n(c)^{-1} :=\left(\begin{array}{cc}
                           \gd_n^{-1/2} & 0\\
                   0 & \gd_n^{-3/2}(1 + c^{-2}d^{-2}_n)^{-1/2}
                           \end{array}\right) \to R^{-1}_{n,H} \qquad  \text{as}\quad c\to\infty, \quad n\in \N.
\end{equation}

Next,  combining    \eqref{IV.1.1_09.2} and \eqref{IV.1.1_09.2.H}
with  \eqref{3.81}   and taking into account relations
\eqref{limk} and \eqref{3.85} we arrive at
%
%
   \begin{equation}\label{limMn}
 \lim_{c\to+\infty} M_{n,c}(z+c^2/2) = M_{n,H}(z) \,, \qquad n\in \N.
  \end{equation}

We emphasize however that the convergence in \eqref{3.84},  hence
the convergence in \eqref{3.85}  and \eqref{limMn},  is uniform in
$n\in\mathbb N$ if and only if $d_*(X)>0$.

(ii)
In this step we show that
     \begin{equation}\label{3.87AA}
s-\lim_{c\to+\infty}\left(\Theta_c - M_c(z+c^2/2)\right)^{-1} =
\left(\Theta - M_H(z)\right)^{-1}, \qquad z\in{\mathbb C}_+.
     \end{equation}
We consider the case of $m$-accumulative $\Theta$, the case of
$m$-dissipative $\Theta$  is treateded similarly.

Straightforward calculations show that the matrices
$M'_{n,c}(c^2/2)$ are uniformly positive in $n\in\mathbb N$ and
$c\in(0,\infty)$. Indeed, it easily follows from
\eqref{IV.1.1_08B} that the following  inequalities hold
    \begin{equation}\label{3.87}
M'_{n.c}(c^2/2) = \begin{pmatrix}
1 & \frac{1}{2}\left(1+\frac{1}{c^{2}\,\gd_n^2}\right)^{-1/2}\\
\frac{1}{2}\left(1+\frac{1}{c^{2}\,\gd_n^2}\right)^{-1/2} &
\frac{1}{3}\frac{3 + c^2\gd_n^2}{1 + c^{2}\,\gd_n^2}
\end{pmatrix} > \frac{1}{16}I_2,\qquad c\in(0,\infty),\quad
n\in\N.
    \end{equation}
Note that the Weyl function $M_c(\cdot)$ corresponding to the
triplet $\Pi(c)=\bigoplus_1^\infty\Pi_n(c)$ is
$M_c(\cdot)=\bigoplus^{\infty}_1 M_{n,c}(\cdot)\in[\cH].$
Combining this fact with inequalities \eqref{3.87} and the
integral representation of the $R[\cH]$-function $M_c(\cdot)$ (see
\eqref{WF_intrepr}) we obtain
     \begin{equation}\label{3.88A}
M'_c(c^2/2) = \bigoplus^{\infty}_{n=1} M_{n,c}'(c^2/2) =
\int_{{\mathbb
R}\setminus(-\alpha_c,\alpha_c)}\frac{1}{(t-c^2/2)^2}d\Sigma_c(t)
> \frac{1}{16}I_\cH,\qquad c\in(0,\infty),
   \end{equation}
where  $(-\alpha_c,\alpha_c) :=
\cap_{n=1}^\infty (-\alpha_{n,c},\alpha_{n,c})$ and $\alpha_c >
c^2/2$. Further, it follows from \eqref{spetcom} that with some
$\varepsilon_0>0$
    \begin{eqnarray}
\alpha_c - \frac{c^2}{2} = \sqrt{\frac{c^2\pi^2}{4 d^*(X)^2} +
\left(\frac{c^2}{2}\right)^2} - \frac{c^2}{2} = \frac{c^2}{2}
\left[\sqrt{\frac{\pi^2}{d^*(X)^2 c^2}+1}-1 \right] \nonumber  \\
= \frac{2^{-1}\pi^2 d^*(X)^{-2}}{\sqrt{\pi^2
d^*(X)^{-2}c^{-2}+1}+1}\ge\varepsilon_0\qquad \text{for}\qquad
c\ge 1.
    \end{eqnarray}
Hence $|t-c^2/2|>\alpha_c - c^2/2  \ge \varepsilon_0,\ t\in{\mathbb
R}\setminus\bigl(-\alpha_c, \alpha_c\bigr)$ and  $c>1.$ In turn,
this inequality yields
   \begin{equation}
\frac{(t-c^2/2)^2}{ (t-c^2/2)^2 + 1} \ge \varepsilon_1 :=
\frac{\varepsilon^2_0}{1+\varepsilon^2_0}, \qquad t\in{\mathbb
R}\setminus\bigl(-\alpha_c, \alpha_c\bigr), \quad c>1.
   \end{equation}
Combining this inequality with \eqref{3.88A}  and using the
integral representation \eqref{WF_intrepr} of the Weyl function $M_c(\cdot)$ we arrive at the
following  uniform estimate
      \begin{equation}\label{3.92}
\im M_c(i+ c^2/2)=  \int_{{\mathbb
R}\setminus(-\alpha_c,\alpha_c)}\frac{1}{(t-c^2/2)^2 +
1}d\Sigma_c(t) \ge \frac{\varepsilon_1}{16}I_\cH,\qquad  c\ge 1.
      \end{equation}
Since $\Theta_c$ is $m$-accumulative, we easily get from
\eqref{3.92} that
  \begin{eqnarray*}
\qquad  \left\|\left(\Theta_c - M_c(i+c^2/2)\right)h\right\|\cdot \|h\| \ge
\left|\left((\Theta_c-M_c(i+c^2/2))h,h\right)\right| \ge \left|\im \left((\Theta_c-
M_c(i+c^2/2))h,h\right)\right|
 \nonumber \\
= - \im\left(\Theta_c h,h\right) + \im \left(M_c(i+c^2/2)h,h\right)  \ge
\frac{1}{16}\varepsilon_1\|h\|^2,\qquad c\ge 1, \qquad h\in l^2_0(\N, \C^2). 
  \end{eqnarray*}
Since $l^2_0(\N, \C^2)$ is dense in $\cH = l^2(\N, \C^2)$, this
inequality yields
    \begin{equation}\label{3.88}
\left\|\left(\Theta_c - M_c(i + c^2/2)\right)^{-1}\right\| \le
16\varepsilon_1^{-1}, \qquad c\ge 1.
  \end{equation}
%
%
Further, relations  \eqref{limMn} immediately imply $
\lim_{c\to+\infty} M_c(z+c^2/2)h = M_H(z)h, \   h\in l^2_0(\N,
\C^2).$
Combining this relation  with \eqref{1.23.0}, yields
    \begin{equation}\label{3.89}
\lim_{c\to+\infty}\, \left(\Theta_c - M_c(z+c^2/2)\right)h = \left(\Theta -
M_H(z)\right)h,   \qquad h\in l^2_0(\N, \C^2), \quad z\in{\mathbb C}_+.
    \end{equation}
In turn, combining \eqref{3.89} with the uniform estimate  \eqref{3.88} and applying
\cite[Theorem 8.1.5]{Kato66}  we arrive at \eqref{3.87AA}.

(iii)  In this step we prove the following limit relation
  \begin{eqnarray}\label{1.24.1}
s-\lim_{c\to+\infty}\, \gamma_c(z +c^2/2)\left(\Theta_c - M_c(z +
c^2/2)\right)^{-1}\gamma_c^{*}(\overline{z} + c^2/2) \nonumber \\
=\left(\gamma_H(z)\left(\Theta - M_H(z)\right)^{-1}\gamma_H^{*}(\overline{z})\right)\bigotimes
\left(\begin{array}{cc} 1 & 0 \\0 & 0\end{array}\right)\,,\qquad
z\in \C_+.
   \end{eqnarray}
It follows from the first identity in \eqref{IV.1.1_09AAA} and
Definition \ref{def_Weylfunc} (see formula \eqref{II.1.3_01}) that
$$
\gamma_{n,c}(z)=\wt\gamma_{n,c}(z)R_n^{-1} \qquad\text{and}\qquad
\gamma_{n,H}(z)=\wt\gamma_{n,H}(z)R_{n,H}^{-1}, \qquad n\in \N.
$$
Combining  these identities  with  \eqref{1.21+} and
\eqref{1.21+.H2} and taking  \eqref{limk} and \eqref{3.85} into
account we obtain
%
%
\begin{equation}\label{limgn}
\lim_{c\to+\infty}\, \gamma_{n,c}(z+c^2/2)
=\left(\begin{array}{cc} (\gamma_{n,H}(z))_1 & (\gamma_{n,H}(z))_2
\\0 & 0\end{array}\right), \qquad  \lim_{c\to+\infty}\, \gamma^*_{n,c}(z+c^2/2)
=\left(\begin{array}{cc} (\gamma^*_{n,H}(z))_1 & 0
\\(\gamma^*_{n,H}(z))_2  & 0\end{array}\right).
\end{equation}
Here $(\gamma_{n,H}(z))_j$ denotes the $j$th component of the
vector function $\gamma_{n,H}(z)$ and the convergence is
understood in $L^2(x_{n-1},x_{n})\otimes \C^4$ and $\cH_n$,
respectively.

Next we prove that the family $\gamma_c(z + c^2/2) =
\bigoplus^{\infty}_{n=1}\gamma_{n,c}(i+c^2/2)$ is uniformly
bounded in $c>1.$ More precisely, assuming for simplicity that
$z=i$ we show that
      \begin{equation}\label{3.96}
\sup_{c>1}\|\gamma_c(c^2/2 \pm i)\| = \sup_{c>1}\|\gamma_c^*(c^2/2
\pm i)\| \le 8\sqrt 3.
    \end{equation}
It follows from   \eqref{II.1.3_02}  that
     \begin{equation}\label{3.97}
\im M_{c}(c^2/2 \pm i)= \im(c^2/2 \pm i)\gamma_{c}^*(c^2/2 \pm
i)\gamma_c(c^2/2 \pm i)= \pm \gamma_c^*(c^2/2 \pm
i)\gamma_c(c^2/2\pm i).
 \end{equation}
So, it suffices to estimate $\im M_c(c^2/2 \pm i)$ from above.

Taking into account formula \eqref{IV.1.1_08BB} where  ${\Delta(c^2/2)} :=
12(1+c^2 d^2_n)(12+c^2 d^2_n)^{-1}< 12$,  we obtain from \eqref{IV.1.1_08BB} and \eqref{3.87}  that
      \begin{equation}
(M_{n,c}'(c^{2}/{2}))^{-1} = \frac{1}{\Delta(c^2/2)}
    \begin{pmatrix}
\frac{1}{3}\frac{3 + c^2\gd_n^2}{1 + c^{2}\,\gd_n^2} &
-\frac{1}{2}\left(1+\frac{1}{c^{2}\,\gd_n^2}\right)^{-1/2}\\
-\frac{1}{2}\left(1+\frac{1}{c^{2}\,\gd_n^2}\right)^{-1/2} & 1
\end{pmatrix}
> \frac{1}{12}\cdot\frac{1}{16}, \quad   n\in\mathbb N,\  c>1.
  \end{equation}
Hence
    \begin{equation}\label{3.99}
{16}^{-1} < M'_{n,c}(c^2/2)<192 \qquad \text{and}\qquad
M'_c(c^2/2) = \bigoplus^{\infty}_{n=1} M_{n,c}'(c^2/2) <192 \cdot I_{\cH}.
  \end{equation}
On the other hand, it follows from \eqref{3.92} and  \eqref{3.88A}
with account of  \eqref{3.99}  that
    \begin{eqnarray}
\pm \im M_c(c^2/2 \pm i)=  \int_{{\mathbb
R}\setminus(-\alpha_c,\alpha_c)}\frac{1}{(t-c^2/2)^2 +
1}d\Sigma_c(t) \le \int_{{\mathbb
R}\setminus(-\alpha_c,\alpha_c)}\frac{1}{(t-c^2/2)^2}d\Sigma_c(t) \nonumber \\
  = M'_c(c^2/2) <192\cdot I_{\cH}, \quad c\ge 1.
    \end{eqnarray}
Combining this estimate with \eqref{3.97}  we arrive at \eqref{3.96}.

Further, note  that the convergence in \eqref{limgn} implies the convergence of
finite direct sums. Finally, combining this fact  with the uniform
estimate \eqref{3.96} we obtain
   \begin{equation}\label{limgnNEW}
s-\lim_{c\to+\infty}\, \gamma_{c}(c^2/2 \pm i)
=\left(\begin{array}{cc} (\gamma_{H}(\pm i))_1 & (\gamma_{H}(\pm
i))_2
\\0 & 0\end{array}\right), \qquad s-\lim_{c\to+\infty}\, \gamma^*_{c}(c^2/2 \pm i)
=\left(\begin{array}{cc} (\gamma^*_{H}(\pm i))_1 & 0
\\(\gamma^*_{H}(\pm i))_2  & 0\end{array}\right),
  \end{equation}
where $(\gamma_{H}(\cdot))_j,\ j\in \{1,2\},$ denotes the $j$th component of the
vector function $\gamma_{H}(\cdot)$. The convergence in  \eqref{limgnNEW} holds
in the spaces $L^2(a,b)\otimes \C^4$ and $\cH$, respectively.
Combining relations  \eqref{limgnNEW} with \eqref{3.87AA}  we arrive
at \eqref{1.24.1}.

(iv)  In this step we prove  formula \eqref{1.23} for
the operators  $D_X^c = \bigoplus_{n=1}^\infty D_{n}^c$
and $H_X = \bigoplus_{n=1}^\infty H_n$ assuming for the moment that the following limit
formula holds
  \begin{equation}\label{convn0New}
u-\lim_{c\to+\infty}\left(D_{n,0}^{c}- (z+c^{2}/2)\right)^{-1}=
\left(H_{n,0}-z\right)^{-1}\bigotimes \left(\begin{array}{cc} 1 & 0
\\0 & 0\end{array}\right)\,, \qquad n\in \N.
  \end{equation}
Here  $D^c_{n,0}:=D_n^{*}\upharpoonright\ker{\Gamma}^{(n)}_{0},$
 ${H}_{n,0}:=H_{n}^{*}\upharpoonright\ker{\Gamma}^{(n)}_{0,H}$ and
$$
\dom(D^c_{n,0}) = \ker{\Gamma}^{(n)}_{0} = \{\{f_1, f_2\}^{\tau}\in
W^{1,2}[x_{n-1},x_{n}]\otimes\C^{2}:  f_{1}(x_{n-1}+)=
f_{2}(x_{n}-)=0\} \quad \text{and} \quad
\dom({H}_{n,0}) = W^{2,2}_0[x_{n-1},x_{n}].
$$
The proof  of \eqref{convn0New}  is  postponed to the next step. Note that
convergence in \eqref{convn0New} is uniform in $L^2[x_{n-1}, x_{n}]\otimes
\C^2.$   \par

According to the Krein-type formula for resolvents (see
\eqref{II.1.4_01}) 
\begin{equation}\label{1.24}
\left(D_{X,\Theta_c}^c - z\right)^{-1} = \left(D^c_{X,0} -
z\right)^{-1} + \gamma_c(z)\left(\Theta_c -
M_c(z)\right)^{-1}\gamma_c^{*}(\overline{z})
\end{equation}
and
   \begin{equation}\label{1.24+}
\left(H_{X,\Theta}-z\right)^{-1}=\left(H_{X,0}-z\right)^{-1}+\gamma_H(z)\left(\Theta-M_H(z)\right)^{-1}\gamma_H^{*}(\overline{z})\,.
  \end{equation}
Here the realizations $D^c_{X,0}$ and $H_{X,0}$ are given by
      \begin{equation}\label{3.106}
D^c_{X,0} = \bigoplus_{n=1}^\infty D^c_{n,0}=
\bigoplus_{n=1}^\infty (D^c_{n,0})^* = (D^c_{X,0})^* \qquad \text{and} \qquad
H_{X,0}= \bigoplus_{n=1}^\infty H_{n,0} =
\bigoplus_{n=1}^\infty (H_{n,0})^* = H_{X,0}^*.
   \end{equation}
Combining  relations  \eqref{convn0New}  with  \eqref{3.106} and noting
that $\left\|\left(D_{n,0}^{c}- (z+c^{2}/2)\right)^{-1}\right\| \le |\im
z|^{-1}$ for any $n$ and  $c>0$, we obtain
  \begin{equation}\label{conv0New}
s-\lim_{c\to+\infty}\left(D_{X,0}^{c}- (z+c^{2}/2)\right)^{-1}=
\left(H_{X,0}-z\right)^{-1}\bigotimes \left(\begin{array}{cc} 1 & 0
\\0 & 0\end{array}\right)\,,\qquad z\in \C_+ .
  \end{equation}
Finally, combining this relation with \eqref{1.24.1} and applying the Krein
type formulae \eqref{1.24} and  \eqref{1.24+}  we arrive at
\eqref{1.23}.

\   (v) In this step we prove formula \eqref{convn0New} as well as the following formulas

  \begin{equation}\label{conv0NewForTwoPoints}
u-\lim_{c\to+\infty}\left(D_{\tau}^{c}- (z+c^{2}/2)\right)^{-1}=
\left(H_{\tau}-z\right)^{-1}\bigotimes \left(\begin{array}{cc} 1 & 0
\\0 & 0\end{array}\right)\,,\qquad \tau = a-, b+, \qquad z\in \C_+ .
  \end{equation}
All formulae can be obtained by direct calculations but we prefer to extract them
from the classical result for the "free" Dirac operator considered on the whole line.
To be precise denote  by  $D^c_{\rm  free}$ and $H_{\rm free}$  the "free" Dirac and
Schr\"{o}dinger operators generated by deferential expressions \eqref{1.2} and $-\frac{d^2}{dx^2}$ on $L^2(\R)\otimes\C^2$ and $L^2(\R)$, respectively.
By definition,  $\dom(D^c_{\rm free})=W^{1,2}(\R)\otimes\C^2$
and $\dom(H_{\rm free})=W^{2,2}(\R)$.
Then  according to the classical result (see e.g. \cite[Chapter 6]{Tha92})
  \begin{equation}\label{free}
u-\lim_{c\to+\infty}\left(D^{c}_{\rm free}-(z+c^{2}/2)\right)^{-1}=
\left(H_{\rm free}-z\right)^{-1}\bigotimes \left(\begin{array}{cc} 1 & 0 \\0 & 0\end{array}\right)\,.
  \end{equation}
\noindent
To this end we introduce  a two points  set $Y:= \{x_{n-1},x_n\} =: \{a,b\}$
and consider  the boundary triplet $\Pi_Y^c = \Pi^{(a-)}\bigoplus \wt\Pi^{(n)}\bigoplus \Pi^{(b+)}$  constructed  in Corollary \ref{s.a.Dirac} for the operator  $(D^c_Y)^* =  (D^c_{a^-}\bigoplus D^c_n\bigoplus D^c_{b^+})^*$.
In other words, $\Pi_{Y}^c =\{\mathbb C^4, \Gamma^c_{0},\Gamma^c_{1}\}$  where $\Gamma_{j}^c: = \Gamma^{(a-)}_{j}\bigoplus\wt{\Gamma}^{(n)}_{j}\bigoplus\Gamma^{(b+)}_{j},$\  $j\in\{0,1\},$ and $\Gamma^{(a-)}_{j},$ $\wt{\Gamma}^{(n)}_{j},$\   and   $\Gamma^{(b+)}_{j}$ are given by \eqref{triple.a}, \eqref{triple2} and \eqref{triple.b}, respectively.


It is easily seen that in the triplet $\Pi_Y^c$ the operator $D^c_{\rm free}$ is given by
     \begin{equation}\label{4.52}
\Gamma_{1}^cf = \begin{pmatrix}
f_1(a-)\\
icf_2(a+)\\
f_1(b-)\\
icf_2(b+)
\end{pmatrix}
=
   \begin{pmatrix}
0&1&0& 0\\
1& 0&0 &0\\
0&0& 0 &1\\
0& 0&1& 0
   \end{pmatrix}
\begin{pmatrix}
ic f_2(a-)\\
f_1(a+)\\
ic f_2(b-)\\
f_1(b+)
\end{pmatrix} =: \Theta_{\rm  free}\Gamma_{0}^cf, \quad f\in \dom\left((D^c_Y)^*\right),
   \end{equation}
i.e. $D^c_{\rm free} = (D^c_Y)^*\upharpoonright\ker(\Gamma^c_1 - \Theta_{\rm  free}\Gamma^c_0)$.
We emphasize that  despite of the dependence of the  triplets $\Pi_Y^c$ on $c$,  the boundary operators $\Theta_{\rm free} = \sigma_1 \oplus \sigma_1$ do not depend on $c.$ Here $\sigma_1 =
\begin{pmatrix}
0&1\\
1&0
\end{pmatrix}$\  (see definition \eqref{pauli}).

Alongside the triplet $\Pi_Y^c$ we  consider the boundary triplet  $\Pi_{Y,H}$ 
for the maximal Schr\"{o}dinger operator
$$
H^*_Y=H^*_{a_-}\bigoplus H^*_n\bigoplus H^*_{b^+},
$$
given in Remark \ref{RMH} (see also Theorem  \ref{th_bt_2.H}). Clearly,
$$
\Pi_{Y,H} =  \{\mathbb C^n,\Gamma_{0,H},\Gamma_{1,H}\} := \Pi_H^{(a-)}\bigoplus\Pi_H^{(n)}\bigoplus \Pi_H^{(b+)}\quad
\text{with}\quad\Gamma_{j,H} := \Gamma^{(a-)}_{j,H}\bigoplus\wt{\Gamma}^{(n)}_{j,H}\bigoplus\Gamma^{(b+)}_{j,H}.
$$
%
%

Here  $\Gamma^{(a-)}_{j,H}$, $\wt{\Gamma}^{(n)}_{j,H}$ and $\Gamma^{(b+)}_{j,H},\ j\in\{0,1\},$
 are given by \eqref{triple.a.H}, \eqref{triple2.H} and \eqref{triple.b.Hb},  respectively.
It is easily seen that in the boundary triplet  $\Pi_{Y,H}$ the free Schr\"{o}dinger operator $H_{\rm free}$  is given by $H_{\rm free} = H_Y^*\upharpoonright\ker(\Gamma_{1,H} - \Theta_{\rm  free}\Gamma_{0,H})$ with the same boundary operator $\Theta_{\rm free}$ as in \eqref{4.52}.

Consider formulae   \eqref{1.24}, \eqref{1.24+} and the limit relation \eqref{1.24.1}  with the set $Y=\{a,b\}$ in place  of $X$ and  $\Theta_c=\Theta=\Theta_{\rm free}$. In this case   \eqref{1.24.1}  holds in  the uniform sense  since $Y$ is finite.
Taking this relation into account and passing to the limit as $c\to\infty$ in  the
Krein type formulae   \eqref{1.24}, \eqref{1.24+}, with account of
 \eqref{free} we arrive at the identity
  \begin{equation}
u-\lim_{c\to+\infty}\left(D_{Y,0}^{c}- (z+c^{2}/2)\right)^{-1}=
\left(H_{Y,0}-z\right)^{-1}\bigotimes \left(\begin{array}{cc} 1 & 0
\\0 & 0\end{array}\right)\,,\qquad z\in \C_+ .
  \end{equation}
In turn, it implies \eqref{convn0New} as well as relations \eqref{conv0NewForTwoPoints}.
   \end{proof}
\begin{corollary}
Assume the conditions of Theorem  \ref{2.5}. Assume, in addition, that  $d_*(X)>0$ (in particular,  $X$ is finite)
and that  in place of \eqref{1.23.0}  the  uniform resolvent convergence holds, i.e.
\begin{equation}\label{1.23.0Unif}
\lim_{c\to +\infty}\|(\Theta_c -z)^{-1} - (\Theta -z)^{-1}\| = 0, \qquad z\in \C_+ \, (z\in \C_-).
  \end{equation}
Then in place of \eqref{1.23} the uniform resolvent convergence holds, i.e.
  \begin{equation}\label{1.23Unif}
u-\lim_{c\to
  +\infty}\left(D_{X,\Theta_c}^c-(z+{c^{2}}/{2})\right)^{-1}=
(H_{X,\Theta}-z)^{-1}\bigotimes \left(\begin{array}{cc} 1 & 0 \\0 &
0\end{array}\right), \qquad z\in \C_+ \, (z\in \C_-).
  \end{equation}
\end{corollary}
\begin{proof}
It  can be proved that in the case $d_*(X)>0$ condition \eqref{1.23.0Unif} implies uniform convergence in \eqref{3.87AA}. It can be done using uniform counterpart of \cite[Lemma 3.1]{DM95}.
Moreover, in this case the convergence in \eqref{limgnNEW}, hence the convergence in \eqref{1.24.1}, is uniform too. Besides, in the case $d_*(X)>0$ the convergence in \eqref{conv0New} is also uniform.  Finally, combining these  relations  and applying the Krein
type formulae \eqref{1.24} and  \eqref{1.24+}  we arrive at \eqref{1.23Unif}.
  \end{proof}
     \begin{remark}\label{2.5.1}
Theorem \ref{2.5} with its proof remains valid  in
the case of Dirac operators $D_X$ on the  line with interaction set  $X=\{x_k\}_{k\in\Z}$,
$x_{k-1}<x_k$.
Indeed, it can be adapted to this case  by using the boundary triplets defined in Remarks \ref{RMD} and \ref{RMH}.

In conclusion, note that Theorem \ref{2.5} comprises (see also  Theorem
\ref{nonrelGS} below) and extends  known results on
the non-relativistic limits of Dirac operators with point
interactions (see \cite{BD}, \cite{GS}, \cite[Appendix J]{Alb_Ges_88} and references therein).
\end{remark}


\section{Gesztesy-\v{S}eba realizations }

Following \cite{GS} (see also \cite{Alb_Ges_88}) we define two
families of symmetric extensions, which  turn out to be closely
related to their non-relativistic counterparts $\delta$- and
$\delta'-$interactions.
 First we consider the case of a finite or infinite
interval $\cI = (a, b) \subseteq \R_a,$\ $-\infty <a.$ Let, as in the previous sections,
$$
X=\{x_n\}_{n\in \N}\,,\quad -\infty < a=:x_0 <x_1 < \ldots <
x_{n}<x_{n+1}<\ldots \,,\quad \lim_{n\to+\infty}x_n = b\le
\infty\,,\quad
$$
and let
$$
\alpha := \{\alpha_{n}\}_{n=1}^\infty\subset\R\cup
\{+\infty\}\,,\qquad \beta := \{\beta_{n}\}_{n=1}^\infty
\subset\R\cup\{+\infty\}\,.
$$
Then the two families of  Gesztesy-\v{S}eba operators (in short,
GS-operators or GS-realizations) on the interval $(a,b)$ are
defined to be the closures of the operators
   \begin{equation}\label{delta}
\begin{split}
 D_{X,\alpha}^0=& D\upharpoonright \dom(D_{X,\alpha}^0),\\
\dom(D_{X,\alpha}^0) = &\Big\{f\in W^{1,2}_{\comp}(\cI \backslash X)\otimes
\C^{2}: f_{1}\in AC_{\loc}(\cI),\ f_{2}\in AC_{\loc}(\cI\backslash X);\\&
f_2(a+)=0\,,\quad f_{2}(x_{n}+)-f_{2}(x_{n}-)
=-\frac{i\alpha_{n}}{c}f_{1}(x_{n}),\,\,n\in\N\Big\},
\end{split}
  \end{equation}
and
  \begin{equation}
\begin{split}\label{deltap}
D_{X,\beta}^0 = & D \upharpoonright \dom(D_{X,\beta}^0), \\
\dom(D_{X,\beta}^0) = & \Big\{f\in W^{1,2}_{\comp}(\cI \backslash X)
\otimes \C^{2}: f_{1}\in AC_{\loc}(\cI\backslash X),\
f_{2}\in AC_{\loc}(\cI);\\&
f_2(a+)=0\,,\quad
f_{1}(x_{n}+)-f_{1}(x_{n}-) = i\beta_{n}cf_{2}(x_{n}),\,\,n\in\N\Big\},
\end{split}
\end{equation}
respectively, i.e. $D_{X,\alpha} =  \overline{ D_{X,\alpha}^0}$ and  $D_{X,\beta} =  \overline{ D_{X,\beta}^0}$.

It is easily seen that both operators  $D_{X,\alpha}$ and  $D_{X,\beta}$ are  symmetric,  but not necessarily self-adjoint, in general. However, both $D_{X,\alpha}$ and  $D_{X,\beta}$ are either symmetric or self-adjoint only simultaneously.
 If $D_{X,\alpha}$ and $D_{X,\beta}$ are  self-adjoint, then their  domains are described explicitly  (see Theorem \ref{1}(i)). Moreover,  the character feature of GS realizations $D_{X,\alpha}$ and  $D_{X,\beta}$ is that they are always  self-adjoint provided that $\cI = \R_{\pm}, \R,$
 (see Proposition \ref{cor_delta_carleman} and Theorem \ref{1}(ii)).
    \begin{remark}\label{remarkGS} \item $(i)$
Originally the GS-realizations  $D_{X,\alpha}$ and
$D_{X,\beta}$ have been introduced (cf. \cite{GS}) in the case of
point interactions distributed on the line $\R.$
In this case  $X=\{x_k\}_{k\in\Z}$,
$\alpha=\{\alpha_k\}_{k\in\Z}$, $\beta=\{\beta_k\}_{k\in\Z}$, and
$\lim_{n\to-\infty}x_n = -\infty$ and $\lim_{n\to +\infty}x_n =
+ \infty$. Moreover, in this case boundary conditions in
\eqref{delta} and \eqref{deltap} are labelled by $n\in \Z$ and the
condition $f_2(a+)=0$ is  dropped.

\item $(ii)$  Note also that if $\alpha_n=\infty$ ($\beta_n =
\infty)$ for some $n\in \N$, then  the $n$-th boundary condition \eqref{delta} (resp. \eqref{deltap}) takes
the form
   \begin{equation}\label{delta_infty}
f_1(x_n) = 0 \qquad (\text{resp. } f_2(x_n) = 0).
   \end{equation}
\end{remark}

\noindent 
In what follows we call conditions (\ref{delta}), (\ref{deltap})
by Gesztesy-\v{S}eba boundary conditions (in short
GS-conditions).\par
%
%
%
%

To investigate Gesztesy-\v{S}eba realizations $D_{X,\alpha}$ in
the framework of boundary triplets approach we first  find
 boundary relations (operators) $\Theta$ that parameterize
operators $D_{X,\alpha}$ according to Corollary \ref{s.a.Dirac}.
It turns out that, as in the Schrodinger case (cf. \cite{KM}), the
boundary operator corresponding to $D_{X,\alpha}$
in the boundary triplet  constructed in Theorem \ref{th_bt_2},  is Jacobi
matrix.

\subsection{GS-realizations $D_{X,\gA}$: parametrization by Jacobi matrices} 

Consider the boundary triplet $\Pi=\{\cH,\Gamma_0,\Gamma_1\}$ for
$D_X^*$  constructed in Theorem \ref{th_bt_2}, (cf. formulae
\eqref{IV.1.1_12}, \eqref{IV.1.1_12.1}).  By Corollary
\ref{s.a.Dirac}(i),  the realization $D_{X,\gA}$ admits the
representation (cf. \eqref{31d})
      \begin{equation}\label{IV.1.1_12A}
D_{X,\gA}  = D_{\Theta(\gA)}:= D_X^* \lceil
\dom(D_{\Theta(\gA)}),\quad \dom(D_{\Theta(\gA)})=\{f\in
W^{1,2}(\cI \backslash X)\otimes\C^{2}:
\ \{\Gamma_0f,\Gamma_1f\}\in\Theta(\gA)\}.   
     \end{equation}

Since  for any $\alpha$ the realizations  $D_{X,\gA}$ and $D_{X,0}
:= D^*_{X}\upharpoonright \ker(\Gamma_0)$ are disjoint,
$\Theta(\gA)$ is a (closed) operator in $\cH=l^2(\N)\otimes \C^2$,
$\Theta(\gA)\in \mathcal{C}(\cH)$.  We show that
$\Theta(\gA)$  is a Jacobi matrix.  More precisely, consider the
Jacobi matrix
   \begin{equation}\label{IV.2.1_01}
B_{X,\gA}=\left(
\begin{array}{cccccc}
  0& -\frac{\nu(\gd_{1})}{\gd_1^{2}}& 0 & 0& 0  &  \dots\\
   -\frac{\nu(\gd_{1})}{\gd_1^{2}} &  -\frac{\nu(\gd_{1})}{\gd_1^{2}}&\frac{\nu(\gd_{1})}{\gd_1^{3/2}\gd_2^{1/2}}& 0 & 0&  \dots\\
  0 & \frac{\nu(\gd_{1})}{\gd_1^{3/2}\gd_2^{1/2}} & \frac{\alpha_1}{\gd_2} & -\frac{\nu(\gd_{2})}{\gd_2^{2}}& 0&   \dots\\
  0 & 0 & -\frac{\nu(\gd_{2})}{\gd_2^{2}}&  -\frac{\nu(\gd_{2})}{\gd_2^{2}}&\frac{\nu(\gd_{2})}{\gd_2^{3/2}\gd_3^{1/2}}&  \dots\\
  0 & 0 & 0 & \frac{\nu(\gd_{2})}{\gd_2^{3/2}\gd_3^{1/2}}& \frac{\alpha_2}{\gd_3}&   \dots\\
\dots& \dots&\dots&\dots&\dots&\dots\\
 \end{array}%
\right)\,,
   \end{equation}
where
   \begin{equation}\label{IV.2.1_01NU}
\nu(x):=\frac{1}{\sqrt{1+\frac{1}{c^2x^2}}}\,.
   \end{equation}
Let $\tau_{X,\gA}$ be the second order difference expression associated with (\ref{IV.2.1_01}).
One defines the corresponding minimal symmetric operator in $l^2(\N)\otimes \C^2$ by (see \cite{Akh, Ber68})
  \begin{equation}\label{IV.2.1_02}
B^0_{X,\gA}f:=\tau_{X,\gA}f,\qquad f\in\dom(B^0_{X,\gA}) := l^2_0(\N)\otimes \C^2, 
\quad \mathrm{and}\quad B_{X,\gA}=\overline{B^0_{X,\gA}}.
  \end{equation}
%
 Recall that $B_{X,\gA}$\footnote{Usually  we will identify the Jacobi matrix with (closed) minimal symmetric operator  associated with it. Namely, we denote by $B_{X,\gA}$ the Jacobi matrix (\ref{IV.2.1_01})  as well as the minimal closed symmetric operator (\ref{IV.2.1_02}).}
 has equal deficiency indices and $\mathrm{n}_+(B_{X,\gA})=\mathrm{n}_-(B_{X,\gA})\leq 1$.

Note that $B_{X,\gA}$ admits a representation
      \begin{equation}\label{IV.2.2_01A}
B_{X,\gA}=R^{-1}_X(\widetilde{B}_{\gA}-Q_X)R^{-1}_X,\quad \text{where}\quad \widetilde{B}_{\gA} := \left(%
\begin{array}{cccccc}
  0& 0 & 0 & 0& 0  & \dots\\
  0 & 0& 1& 0 & 0&  \dots\\
  0 & 1 & \alpha_1 &   0& 0& \dots\\
  0 & 0 & 0 & 0 & 1&  \dots\\
  0 & 0 & 0 & 1 & \alpha_2&  \dots\\
\dots& \dots&\dots&\dots&\dots&\dots\\
 \end{array}
\right),
  \end{equation}
and $R_X= \bigoplus_{n=1}^\infty R_n,\ Q_X=\bigoplus_{n=1}^\infty Q_n$ and $R_n,$  $Q_n$,
are defined by \eqref{IV.1.1_19}.
     \begin{proposition}\label{prop_IV.2.1_01}
Let $\Pi=\{\cH,\Gamma_0,\Gamma_1\}$ be the boundary triplet for
$D_{X}^*$ constructed in Theorem \ref{th_bt_2} and let $B_{X,\gA}$
be the minimal Jacobi operator defined by \eqref{IV.2.1_01}. Then
$\Theta(\gA) = B_{X,\gA}$, i.e.,
   \begin{equation*}
D_{X,\gA}=D_{B_{X,\gA}}=D_{X}^*\upharpoonright\dom(D_{B_{X,\gA}}),\qquad
\dom(D_{B_{X,\gA}})=\{f\in W^{1,2}(\cI\setminus X)\otimes \C^2:
\Gamma_1f=B_{X,\gA}\Gamma_0f\}.
     \end{equation*}
     \end{proposition}
  \begin{proof}
 Let $f\in W_{\comp}^{1,2}(\cI\setminus X)\otimes \C^2 = \bigl( W^{1,2}(\cI\setminus X)\cap L_{\comp}^{2}(\cI)\bigr)\otimes \C^2$. Then $f\in \dom(D_{X,\gA})$ if and only
 if $\wt{\Gamma}_1f=\wt{B}_\alpha\wt{\Gamma}_0f.$
 Here
$\wt{\Gamma}_j :=\bigoplus_{n\in \N}\wt{\Gamma}_j^{(n)}$ where
$\wt{\Gamma}_j^{(n)},\ j\in \{0,1\},$ are defined  by \eqref{triple2}                    
and $\wt{B}_\alpha$ is given  by \eqref{IV.2.2_01A}.
 Combining  \eqref{IV.1.1_09AAA}, \eqref{IV.1.1_19}  with  \eqref{IV.2.2_01A},
we rewrite the equality
$\wt{\Gamma}_1f=\wt{B}_\alpha\wt{\Gamma}_0f$  as
$\Gamma_1f=B_{X,\gA}\Gamma_0f$, $f\in W_{\comp}^{1,2}(\cI\setminus X)\otimes \C^2$.
 Taking the closures and applying Corollary
\ref{s.a.Dirac}(i) one completes  the
proof.
      \end{proof}
   \begin{remark}\label{remark4.4}
Note that the matrix \eqref{IV.2.1_01} has negative off-diagonal
entries, although, in the classical theory of Jacobi operators,
off-diagonal entries are assumed to be positive. But it is known
(see, for instance, \cite{Tes_98}) that the (minimal) operator
$B_{X,\alpha}$ is unitarily equivalent to the minimal Jacobi operator associated
with the matrix
   \begin{equation}\label{IV.2.1_01'}
B'_{X,\alpha} :=\left(%
\begin{array}{cccccc}
 0& \frac{\nu(\gd_{1})}{\gd_1^{2}}& 0 & 0& 0  &  \dots\\
   \frac{\nu(\gd_{1})}{\gd_1^{2}} &  -\frac{\nu(\gd_{1})}{\gd_1^{2}}&\frac{\nu(\gd_{1})}{\gd_1^{3/2}\gd_2^{1/2}}& 0 & 0&  \dots\\
  0 & \frac{\nu(\gd_{1})}{\gd_1^{3/2}\gd_2^{1/2}} & \frac{\alpha_1}{\gd_2} & \frac{\nu(\gd_{2})}{\gd_2^{2}}& 0&   \dots\\
  0 & 0 & \frac{\nu(\gd_{2})}{\gd_2^{2}}&  -\frac{\nu(\gd_{2})}{\gd_2^{2}}&\frac{\nu(\gd_{2})}{\gd_2^{3/2}\gd_3^{1/2}}&  \dots\\
  0 & 0 & 0 & \frac{\nu(\gd_{2})}{\gd_2^{3/2}\gd_3^{1/2}}& \frac{\alpha_2}{\gd_3}&   \dots\\
\dots& \dots&\dots&\dots&\dots&\dots\\
 \end{array}
\right).
    \end{equation}
In the sequel we will identify the operators $B_{X,\alpha}$ and
$B'_{X,\alpha}$ when investigating those spectral properties of the operator $\rH_{X,\gA}$, which are invariant under unitary transformations.
        \end{remark}

\subsubsection{Self-adjointness}

{\bf 1. Boundary triplets approach.}
\noindent First we  study self-adjointness of GS-realizations $D_{X,\gA}$
using the parametrization by means of the  Jacobi matrices
$B_{X,\gA}$ 

 Combining Corollary  \ref{s.a.Dirac}(ii)  with Propositions
\ref{prop_IV.2.1_01} we arrive at the following statement.
     \begin{proposition}\label{th_delta_sa}
The GS-operator $D_{X,\gA}$ has equal deficiency indices and
$\mathrm{n}_+(D_{X,\gA})=\mathrm{n}_-(D_{X,\gA})\leq 1$. Moreover,
$\mathrm{n}_\pm(D_{X,\gA})=\mathrm{n}_\pm(B_{X,\gA})$, where
$B_{X,\gA}$ is the minimal Jacobi operator associated with the
Jacobi
matrix 
\eqref{IV.2.1_01}. In
particular, $D_{X,\gA}$ is self-adjoint if and only if $B_{X,\gA}$
is.
      \end{proposition}

Using Carleman's criterium of self-adjointness of Jacobi matrices (see
e.g. \cite{Akh, Ber68, KosMir99, KosMir01}), we obtain
sufficient conditions for the operator $D_{X,\gA}$ to be
self-adjoint  in $L^2(\cI, \C^2)$.
   \begin{proposition}\label{cor_delta_carleman}
Let $\cI$ be  an infinite interval, i.e.
either  $\cI=\R_{\pm}$ or  $\cI=\R$.
Then the GS-realization  $D_{X,\gA}$  is self-adjoint for any sequence
$\gA=\{\alpha_n\}_{n=1}^\infty\subset\R\cup \infty.$
      \end{proposition}
       \begin{proof}
Let $B_{X,\gA}'$ be the minimal Jacobi operator of the form
\eqref{IV.2.1_01'}. By Carleman's test  (see \cite{Akh},
\cite[Chapter VII.1.2]{Ber68}), $B_{X,\gA}'$ is self-adjoint
provided that
\begin{equation}\label{IV.2.2_02}
\sum_{n=1}^\infty \left(\gd_n^2+\gd_n^{3/2}\gd_{n+1}^{1/2}\right)\sqrt{1+\frac{1}{c^{2}\,\gd_n^2}} = \frac{1}{c}\sum_{n=1}^\infty \bigl(\gd_n + \gd_n^{1/2}\gd_{n+1}^{1/2}\bigr)\sqrt{1 + {c^{2}\gd_n^2}}    = \infty.
     \end{equation}

Since $d_n < \gd_n + \gd_n^{1/2}\gd_{n+1}^{1/2}\le \frac
32 d_n + \frac 12 d_{n+1},$ the series in the left-hand side of
\eqref{IV.2.2_02} diverges only simultaneously  with the series
    \begin{equation}\label{IV.2.2_02BB}
\sum_{n=1}^\infty \gd_n^2\sqrt{1 + \frac{1}{c^{2}\,\gd_n^2}} =
\frac{1}{c}\sum_{n=1}^\infty \gd_n\sqrt{1 + {c^{2}\gd_n^2}}.
   \end{equation}
The later  series  diverges if and only if
$\sum_{n=1}^{\infty}d_{n}=+\infty,$ i.e. if and only if the interval $\cI$ is
infinite.  Thus, the minimal Jacobi operator  $B_{X,\gA}$ is self-adjoint
whenever  the interval  $\cI$ is infinite.
It remains to apply Proposition \ref{th_delta_sa}.
   \end{proof}
   \begin{remark} Note that the condition $\sum_{n=1}^{\infty}d_{n}=+\infty$   
is equivalent to the following one 
  \begin{equation}\label{carle_eq}
\sum_{n=1}^{\infty}d_{n}\,\sqrt{d_{n}^{2}+\frac{1}{c^{2}}} =
+\infty\,.
  \end{equation}
The formal (non-relativistic) limit in \eqref{carle_eq} as
$c\to\infty$ leads to the condition
$\sum_{n=1}^{\infty}d^2_{n}=+\infty$, coinciding with that of
\cite[Proposition 5.7]{KM}. The latter guaranties  the self-adjointness
of   the Schr\"{o}dinger operator with point interactions.
\end{remark}
Next we present sufficient  conditions for GS-operators $D_{X,\alpha}$ on a finite interval to be self-adjoint.
%
%
  \begin{proposition}\label{interval}
Assume that $|\cI| < \infty$. Then the $GS$-realization $D_{X,\alpha}$ in $L^2(\cI, \C^2)$ is selfadjoint provided that
\begin{equation}\label{4.14}
\sum_{n\in \N} \sqrt{d_{n}d_{n+1}}\,|\alpha_n|=+\infty\,.
 \end{equation}
  \end{proposition}
  \begin{proof} By Proposition  \ref{th_delta_sa},  it suffices to show that  the minimal Jacobi operator $B_{X,\gA}'$ associated with the Jacobi matrix  \eqref{IV.2.1_01'}
   is  selfadjoint. By the Dennis-Wall test (see \cite{Akh}, Problem 2, p.25), $B_{X,\alpha}$ is self-adjoint whenever
  \begin{equation}\label{4.15}
\sum_{n=1}^\infty\frac{d_{n+1}^{3/2}}{\nu(d_{n+1})}\left(\frac{d_{n}^{3/2}|\alpha_n|}{\nu(d_n)} + d_{n+2}^{1/2}\right)=+\infty\,.
  \end{equation}
The  condition $|\cI| < \infty$ is equivalent to
$\sum_{n=1}^\infty d_n<+\infty$. The latter  implies
$\nu(d_{n})\thicksim c\,d_n$. Hence  $\frac{d_{n}^{3/2}}{\nu(d_n)} \thicksim c^{-1}\,d_n^{1/2}$.
Taking these relations into account and noting that
$$
2\sum_{n\in \N} \sqrt{d_{n+1}}\,\sqrt{d_{n+2}}\le
 \sum_{n\in \N} (d_{n+1} + d_{n+2}) < +\infty,
$$
one concludes that the series \eqref{4.14} and \eqref{4.15} diverge only simultaneously.
\end{proof}
     \begin{example}
Let $\cI := (0,1)$ and let the sequence  $X=\{x_n\}_{n=1}^\infty\subset(0,1)$ be given by $x_n = 1- 1/2^n$, so that $d_n=1/2^n$. Let also $\alpha=\{\alpha_n\}_1^\infty$ be given by
$\alpha_n = (-3)2^n+1,\ n\in \N.$
By Proposition \ref{interval},  the $GS$-operator
$D_{X,\alpha}$ on $L^2(0,1)\otimes \C^2$ is selfadjoint since the series $\sum_{n=1}^\infty{\alpha_n}/{2^n}$ diverges.

On the other hand, it is easily seen that
$$
\{d_n\}_1^\infty \in l^2(\N)\,,\quad d_{n-1}d_{n+1}=\frac{1}{2^{2n}}=d_n^2\quad \text{and}\quad
\sum_{n=1}^\infty d_{n+1}\left|\alpha_n+\frac{1}{d_n}+\frac{1}{d_{n+1}}\right|=
\sum_{n=1}^\infty\frac{1}{2^{n+1}}=\frac12.
$$
Therefore, by \cite[Proposition 5.9]{KM}, the corresponding  Schr\"odinger operator $H_{X,\alpha}$ on $L^2(0,1)$ is not self-adjoint: it is symmetric with
the deficiency  indices $n_{\pm}(H_{X,\alpha}) =1$.
     \end{example}
%
\noindent
{\bf 2. The classical approach.}
\noindent   Now we show, by a direct proof, that in the case $\cI=\R$,
$X=\{x_k\}_{k\in\Z}$, $\alpha=\{\alpha_k\}_{k\in\Z}$ and
$\beta=\{\beta_k\}_{k\in\Z}$ (see Remark \ref{remarkGS}(i)) the
Gesztesy-\v{S}eba operators are always self-adjoint. This proof can readily be extended for
other realizations  as well as for Dirac operators $D_{X,\alpha}(Q)$ with unbounded
potential matrix $Q$.
\begin{theorem}\label{1}
Let $D_{X,\alpha}$ and $D_{X,\beta}$ be  GS-realizations of the Dirac operator in $L^2(\cI,\C^2)$.
Then:

\item $(i)$  The operator $D_{X,\alpha}^* := (D_{X,\alpha})^*$ adjoint to the symmetric operator $D_{X,\alpha}$  is given by
   \begin{equation}\label{delta_adjoint}
\begin{split}
 D_{X,\alpha}^*=& D\upharpoonright \dom(D_{X,\alpha}^*),\\
\dom(D_{X,\alpha}^*) = &\Big\{f\in W^{1,2}(\cI \backslash X)\otimes
\C^{2}: f_{1}\in  W^{1,2}(\cI),\ f_{2}\in AC_{\loc}(\cI\backslash X);\\&
f_2(a+)=0\,,\quad f_{2}(x_{n}+)-f_{2}(x_{n}-)
=-\frac{i\alpha_{n}}{c}f_{1}(x_{n}),\,\,n\in\N\Big\}.
\end{split}
  \end{equation}
Similarly, the operator $D_{X,\beta}^*$ adjoint  to $D_{X,\beta}$ is given by the expression \eqref{deltap}  with  $W^{1,2}_{\comp}(\cI \backslash X)$ replaced by  $W^{1,2}(\cI \backslash X)$.

\item $(ii)$  If $|\cI|= \infty$ (i.e. either  $\cI = \R_{\pm}$ or  $\cI = \R$)  
then  both $D_{X,\alpha}$ and $D_{X,\beta}$ are selfadjoint, i.e.
\begin{equation}
   D_{X,\alpha}^{*} =  D_{X,\alpha}
\qquad  \textrm{and}\qquad
D_{X,\beta}^{*} = D_{X,\beta}.
  \end{equation}
   \end{theorem}
      \begin{proof}
(i) Denote the right hand side of  \eqref{delta_adjoint} by $W^{1,2}_{\alpha}(\cI \backslash X).$
Then for any  $f\in W^{1,2}_{\alpha}(\cI \backslash X)$ integrating by parts
one arrives at the identity
%
%
 \begin{equation}\label{4.13}
  \left(D^{*}_{X,\alpha}\,f,\varphi\right) = \left(f,D_{X,\alpha}\,\varphi\right), \qquad
  \varphi= \binom{\varphi_{1}}{\varphi_{2}}   \in \dom(D_{X,\alpha}),
\end{equation}
proving the inclusion  $W^{1,2}_{\alpha}(\cI \backslash X)\subset \dom(D_{X,\alpha}^*).$

Let us prove the converse inclusion.  Since  $D_X \subset D_{X,\alpha}$ and $D_X$ is  symmetric,  one has
%
%
   \begin{equation}\label{4.11}
\dom(D^{*}_{X,\alpha}) \subset \dom(D_X^*) = W^{1,2}({\cI}\setminus X)\otimes \mathbb C^2 =
\bigoplus_{n\in \N} W^{1,2}[x_{n-1},x_{n}]\otimes\C^{2}.
     \end{equation}
Let $f= \binom{f_1}{f_2} \in  \dom\left(D^{*}_{X,\alpha}\right)$.
Then,  by definition, \eqref{4.13} holds.
According  to  the "regularity property" \eqref{4.11}  we can integrate by parts
in \eqref{4.13} over  any interval $[x_{n-1},x_{n}],\ n\in \N$.
Substitute in \eqref{4.13}  vector functions  $\varphi$
supported on a small neighborhood $(x_{j} - \varepsilon,\,x_{j} +
\varepsilon)$ of $x_{j}$ and  integrating by parts  we get
%
%
    \begin{equation*}
[f_2(x_{j}+)-f_2(x_j-) + i\alpha_jc^{-1} f_1(x_j+)]
\overline{\varphi_1(x_j)} + [f_1(x_{j}+) - f_1(x_j-)]
\overline{\varphi_2(x_j-)} = 0.
    \end{equation*}
Since $\varphi_1(x_j)$ and $\varphi_2(x_j -)$ are arbitrary, the latter equality holds if and only if $f_1(x_{j}+) = f_1(x_j-)$  and $f_2(x_{j}+)-f_2(x_j-)  = -i\alpha_jc^{-1} f_1(x_j)$. Since
$j\in \N$ is arbitrary, $f$ satisfies boundary conditions in \eqref{delta_adjoint} and
$\dom(D_{X,\alpha}^*) \subset W^{1,2}_{\alpha}(\cI \backslash X).$  Noting that the opposite   inclusion is already proved,  we arrive at \eqref{delta_adjoint}.

(ii)  For definiteness we assume that $\cI = \R.$ It suffices to show that   $D^{*}_{X,\alpha}$
 is symmetric. Let $f = \binom{f_1}{f_2},\, g = \binom{g_1}{g_2} \in
 \dom\left(D^{*}_{X,\alpha}\right)$.

Choosing $a,b\in\R\backslash X,\, a<b$, such that
$x_{p-1}<a<x_{p}<x_{p+1}<...<x_{q}<b<x_{q+1}$ and integrating by
parts we get
    \begin{equation}\label{4d}
  \begin{split}
  &\int_{a}^{b}D^{*}_{X,\alpha}f(x)\overline{g(x)}dx-\int_{a}^{b}f(x)\overline{D^{*}_{X,\alpha}g(x)}dx\\
  &=-i\,c\,\int_{a}^{b}\left[f'_{2}(x)\overline{g_{1}(x)}+f'_{1}(x)\overline{g_{2}(x)}\right]dx-i\,c\,\int_{a}^{b}\left[f_{1}(x)\overline{g'_{2}(x)}+f_{2}(x)\overline{g'_{2}(x)}\right]dx\\
  &=-i\,c\,\left.\left[f_{2}(x)\overline{g_{1}(x)}+f_{1}(x)\overline{g_{2}(x)}\right]\right|_{a}^{x_{p}-}-i\,c\,\left.\left[f_{2}(x)\overline{g_{1}(x)}+f_{1}(x)\overline{g_{2}(x)}\right]\right|_{x_{q}+}^{b}\\
  &-i\,c\,\sum_{k=p+1}^{q}\left.\left[f_{2}(x)\overline{g_{1}(x)}+f_{1}(x)\overline{g_{2}(x)}\right]\right|_{x_{k-1}+}^{x_{k}-}\\
  &=i\,c\,\left[f_{2}(a)\overline{g_{1}(a)}+f_{1}(a)\overline{g_{2}(a)}\right]-i\,c\,\left[f_{2}(b)\overline{g_{1}(b)}+f_{1}(b)\overline{g_{2}(b)}\right]
  \\
  &+i\,c\,\sum_{k=p}^{q}\left[f_{2}(x_{k}+)-f_{2}(x_{k}-)\right]\overline{g_{1}(x_{k})}+i\,c\,\sum_{k=p}^{q}f_{1}(x_{k})\left[\overline{g_{2}(x_{k}+)}-\overline{g_{2}(x_{k}-)}\right]\\
  &=i\,c\,\left[f_{2}(a)\overline{g_{1}(a)}+f_{1}(a)\overline{g_{2}(a)}\right]-i\,c\,\left[f_{2}(b)\overline{g_{1}(b)}+f_{1}(b)\overline{g_{2}(b)}\right].
  \end{split}
  \end{equation}
Since $f_{j},g_{j}\in L^{2}(\R),\,j\in\{1,2\}$, there exist (non unique) sequences $\{a_{n}\},\{b_{n}\}\subset\R$ such that $a_{n}\rightarrow-\infty$, $b_{n}\rightarrow\infty$ as $n\rightarrow\infty$ and
\begin{equation}\label{5d}
 \lim_{n\to\infty}\big\{|f_{1}(a_{n})|+|f_{2}(a_{n})|+|g_{1}(a_{n})|+|g_{2}(a_{n})|\big\}=0
\end{equation}
and
\begin{equation}\label{6d}
      \lim_{n\to\infty}\big\{|f_{1}(b_{n})|+|f_{2}(b_{n})|+|g_{1}(b_{n})|+|g_{2}(b_{n})|\big\}=0.
      \end{equation}
\noindent
Without loss of generality, we can assume
that $\{a_{n}\},\{b_{n}\}\subset\R\backslash X$ since \eqref{5d}
and \eqref{6d} remain valid with $\{a_{n}\pm\varepsilon_{n}\}$ and
$\{b_{n}\pm\varepsilon_{n}\}$ in place of $\{a_{n}\}$ and
$\{b_{n}\}$, respectively, provided that  $\varepsilon_{n},\ n\in \N,$ are small
enough.

 According to (\ref{4d})
      \begin{equation}\label{7d}
  \begin{split}
  &\int_{a_{n}}^{b_{n}}D^{*}_{X,\alpha}f(x)\overline{g(x)}dx-\int_{a_{n}}^{b_{n}}f(x)\overline{D^{*}_{X,\alpha}g(x)}dx\\
  &=i\,c\,\left[f_{2}(a_{n})\overline{g_{1}(a_{n})}+f_{1}(a_{n})\overline{g_{2}(a_{n})}\right]-i\,c\,\left[f_{2}(b_{n})
  \overline{g_{1}(b_{n})}+f_{1}(b_{n})\overline{g_{2}(b_{n})}\right].
  \end{split}
  \end{equation}
   \noindent
Passing here  to the limit  as $n\to\infty$ with account of  relations \eqref{5d},
\eqref{6d} we arrive at the identity
  \begin{equation}\label{8d}
  \left(D^{*}_{X,\alpha}\,f,g\right)=\left(f,D^{*}_{X,\alpha}\,g\right),\qquad f,g\in\dom(D^{*}_{X,\alpha}),
  \end{equation}
showing that   $D^{*}_{X,\alpha}$ is symmetric,
$D^{*}_{X,\alpha}\subseteq D^{*\,*}_{X,\alpha}=D_{X,\alpha}$.
Since $D_{X,\alpha}$ is also symmetric, one has
$D_{X,\alpha}= D^{*}_{X,\alpha}.$

\noindent The case of GS realizations $D_{X,\beta}$  is
considered in much  the same way.
   \end{proof}
   \begin{remark} \item $(i)$ In the case $d_{*}(X)>0$ this result is stated in \cite{GS} (see also
\cite[Appendix J ]{Alb_Ges_88}).

\item $(ii)$  The proof of Theorem \ref{1} remains valid for general Dirac operators 
with arbitrary potential matrix $Q\in L_{\loc}^2(\R)\otimes \C^{2\times 2}$,
  \begin{equation}\label{4.9}
D_{X,\alpha}(Q) := -i\,c\,\frac{d}{dx}\otimes\sigma_{1} +
\frac{c^{2}}{2}\otimes\sigma_{3} + Q(x),\qquad Q(x)=Q(x)^{*},
\end{equation}
subject to GS-boundary conditions \eqref{delta}, \eqref{deltap}.

Moreover, the GS--boundary conditions  \eqref{delta} and
\eqref{deltap} can be replaced by certain  other  boundary conditions.
For instance, Theorem \ref{1} as well as its proof  remains valid for
operators $D_{X,\gamma}(Q)$  generated by differential expression
\eqref{4.9} subject to the boundary conditions
   \begin{equation}\label{10d}
\begin{split}
&f_{1}(x_{j}+)=\cos(\gamma_{j})\,f_{1}(x_{j}-)-i\sin(\gamma_{j})\,f_{2}(x_{j}-), \\
&f_{2}(x_{j}+)=\cos(\gamma_{j})\,f_{2}(x_{j}-)-i\sin(\gamma_{j})\,f_{1}(x_{j}-)\,,
\end{split}
   \end{equation}
with $\gamma_{j}\in \R$, $j\in \Z.$ Note  that realizations $D_{X,\gamma}$ have been studied in numerous papers under the assumption $d_{*}(X)>0$ (see for instance \cite{Gumbs}, \cite{Lapidus},  \cite{AvrGro},\cite{DavSte},  as well \cite[Appendix J ]{Alb_Ges_88} and the references therein).


Our proof of Theorem \ref{1} generalizes the known proof of selfadjointness in $L^2(\R)\otimes \C^2$ of the Dirac operator $D(Q)$ with a continuous potential matrix $Q$ (see \cite[Chapter 8]{LevSar88}).
        \end{remark}
Next we complete Proposition \ref{interval}  providing  sufficient conditions for GS realization $D_{X,\alpha}$ on
a finite interval $\cI$ to have non-trivial deficiency  indices $n_{\pm}(D_{X,\alpha})=1$.
\begin{theorem}\label{not-s.a.}
Let $|\cI|<+\infty$ and let $D_{X,\alpha}$ be  the GS realization  of the Dirac
expression on $\cI$. Then $D_{X,\alpha}$ is symmetric with $n_{\pm}(D_{X,\alpha}) =1$ provided that
\begin{equation}\label{non-self_cond}
\sum_{n=2}^{\infty}d_{n}\prod_{k=1}^{n-1}\left(1+\frac 1c\,|\alpha_{k}|\right)^{2} < +\infty\,.  
\end{equation}
\end{theorem}
  \begin{proof}
We  examine the  operator  %
\begin{equation}
T_{X,\alpha}:=D_{X,\alpha}-\frac{c^{2}}{2}\otimes\left(
\begin{matrix}1&0\\
0&-1\end{matrix}\right)=-i\,c\,\frac{d}{dx}\otimes\left(
\begin{matrix}0&1\\
1&0\end{matrix}\right)
\end{equation}
since obviously $n_{\pm}(D_{X,\alpha}) = n_{\pm}(T_{X,\alpha}).$
It suffices to show that under the  assumption \eqref{non-self_cond} the equation
$(T^{*}_{X,\alpha}+i)f=0$ has a non-trivial $L^2(\cI, \C^2)$-solution. The equation  is equivalent to
the system  $cf'_{1}=f_{2}$ and $cf'_{2}=f_{1}$ which has the following 
piecewise smooth solutions
 \begin{eqnarray}\label{5.27}
f_{1}=\bigoplus_{n=1}^{\infty}f_{1,n}\,,\quad
f_{1,n}(x)=a_{n}e^{-(x_{n}-x)/c}+b_{n}e^{(x_{n}-x)/c}\,,\quad x\in [x_{n-1},x_{n}]\,, \nonumber \\
f_{2}=\bigoplus_{n=1}^{\infty}f_{2,n}\,,\quad
f_{2,n}(x)=a_{n}e^{-(x_{n}-x)/c}-b_{n}e^{(x_{n}-x)/c}\,,\quad x\in [x_{n-1},x_{n}]\,.
 \end{eqnarray}
Let us find  sequences $\{a_{n}\}_{1}^{\infty}\subset\C$, $\{b_{n}\}_{1}^{\infty}\subset\C$  such that $f\in {\dom}(T^{*}_{X,\alpha}) = {\dom}(D^{*}_{X,\alpha})$.
According to the description of ${\dom}(D^{*}_{X,\alpha})$ (see Theorem  \ref{1}(i))  $\binom{f_1}{f_2}$ should satisfy boundary conditions \eqref{delta_adjoint}.
The  condition $f_{2}(x_{0}+)=0$ yields  $a_{1}e^{-d_{1}/c}-b_{1}e^{d_{1}/c}=0$. Further,  the  condition
$$
f_{1,n}(x_{n}+)=f_{1,n}(x_{n}-), \qquad n\in \N,
$$
is transformed into
 \begin{equation}\label{5.28A}
a_{n+1}e^{-d_{n+1}/c}+b_{n+1}e^{d_{n+1}/c}=a_{n}+b_{n}\,, \qquad n\in \N.
 \end{equation}
Moreover,  the jump condition
$$
f_{2,n}(x_{n}+)-f_{2,n}(x_{n}-)=-i\,\frac{\alpha_{n}}{c}\,f_{1,n}(x_{n}),\qquad n\in \N,
$$
is equivalent to
  \begin{equation}\label{5.28B}
a_{n+1}e^{-d_{n+1}/c}-b_{n+1}e^{d_{n+1}/c}-(a_{n}-b_{n})=-i\,\frac{\alpha_{n}}{c}\,(a_{n}+b_{n}), \qquad n\in \N\,.
  \end{equation}
Clearly,  relations \eqref{5.28A} and \eqref{5.28B}
are equivalent  to  the following recursive equations
  \begin{eqnarray}\label{5.28}
a_{n+1}=\left(a_{n}-i\,\frac{\alpha_{n}}{2c}\,(a_{n}+b_{n}\right)\,e^{d_{n+1}/c}\,,\qquad n\in \N, \nonumber \\
b_{n+1}=\left(b_{n}+i\,\frac{\alpha_{n}}{2c}\,(a_{n}+b_{n}\right)\,e^{-d_{n+1}/c}\,,\qquad n\in \N,
\end{eqnarray}
for sequences $\{a_{n}\}_{1}^{\infty}$ and $\{b_{n}\}_{1}^{\infty}$
with the following initial data
$$
a_{1}= e^{d_{1}/c} \quad \text{and} \quad  b_{1}= e^{-d_{1}/c}\,.
$$
It remains to check that under condition  \eqref{non-self_cond} the inclusion $f_{1}, f_{2}\in L^{2}(\cI)$ holds.
It follows from  \eqref{5.27} that
   \begin{align*}
\|f_{k}\|_{2}^{2}=&\sum_{n=1}^{\infty}\|f_{k,n}\|^{2}_{2}\le{2}
\sum_{n=1}^{\infty}\int_{0}^{d_{n}}(|a_{n}|^{2}e^{-2x/c}+ |b_{n}|^{2}e^{2x/c})dx\\
= &\
c\ \sum_{n=1}^{\infty}\left(|a_{n}|^{2}(1-e^{-2d_{n}/c})+|b_{n}|^{2}(e^{2d_{n}/c}-1)\right),\qquad k\in \{1,2\} \,.
\end{align*}
Since $\sum_{n=1}^{\infty}d_{n}=|\cI|<+\infty$, $d_{n}\to0$ and therefore   $(1-e^{-2d_{n}/c})\sim (e^{2d_{n}/c}-1)\sim 2d_{n}/c$ as $n\to \infty$.
This implies inequality  $\|f_{k}\|_{2}<+\infty$ whenever
  \begin{equation}\label{5.30}
\sum_{n=1}^{\infty}\left(|a_{n}|^{2}+|b_{n}|^{2}\right)d_{n}<+\infty\,.
  \end{equation}

Let us prove by induction the following estimates
    \begin{equation}\label{5.29}
|a_{n+1}|, \ |b_{n+1}| \le \exp\left(\frac{d_1+ \ldots + d_{n+1}}{c}\right)\cdot\prod^n_{k=1}\left(1+\frac{|\alpha_k|}{c}\right), \qquad n\in \N.
  \end{equation}

For $n=1$ these estimates are obvious. Assume that inequalities \eqref{5.29} are
proved for $n\le m-1$.
Then for $n=m$ we obtain from \eqref{5.28}  and  \eqref{5.29} that
    \begin{eqnarray*}
|a_{m+1}|\le \left(|a_m|\left(1+\frac{|\alpha_m|}{2c}\right) + \frac{|\alpha_m|}{2c}|b_m|\right) e^{d_{m+1}/c}  \nonumber \\
\le \prod^{m-1}_{k=1}\left(1+\frac{|\alpha_k|}{c}\right) \left[\left(1+\frac{|\alpha_m|}{2c}\right) + \frac{|\alpha_m|}{2c}\right] \exp\left(\frac{d_1+ \ldots + d_m +d_{m+1}}{c}\right) \nonumber \\
=  \exp\left(\frac{d_1+ \ldots + d_{m+1}}{c}\right)\cdot
\prod^m_{k=1}\left(1+\frac{|\alpha_k|}{c}\right).
   \end{eqnarray*}
This inequality proves  the inductive hypothesis  \eqref{5.29} for $a_{n}$.
The estimate for $b_{m+1}$ is proved similarly. Thus, both  inequalities \eqref{5.29} are established.
Combining \eqref{5.30} with  \eqref{5.29} and  the assumption \eqref{non-self_cond}  we conclude  that $f_1, f_2 \in L^2(\cI)$. This  completes the proof.
   \end{proof}
%

Next we extract from Theorem  \ref{not-s.a.} certain simple sufficient conditions for the equality $n_{\pm}(D_{X,\alpha}) =1$ to hold.  First we present  such conditions involving $\alpha$ and not depending
on $X=\{x_n\}_1^\infty$.
%
%
  \begin{corollary}\label{C1}
 The GS realization $D_{X,\alpha}$ on a finite interval $\cI$ is symmetric with $n_{\pm}(D_{X,\alpha}) =1$ whenever $\alpha = \{\alpha_{n}\}_1^{\infty}\in l^1(\N).$
  \end{corollary}
  \begin{proof}
Clearly,  for any positive sequence $\{p_{k}\}_{1}^{\infty}$
$$
\prod_{k=1}^{\infty}(1+p_{k})\le \exp\left(\sum_{k=1}^{\infty}p_{k}\right)\,.
$$
It follows with account of the inclusion $\alpha \in l^1(\N)$ that
\begin{align*}
\sum_{n=2}^{\infty}d_{n}\prod_{k=1}^{n-1}\left(1 + \frac 1c\,|\alpha_{k}|\right)^{2} \le \exp \left(\frac {2}{c}\sum_{k=1}^{\infty}
|\alpha_{k}|\right) \sum_{n=2}^{\infty}d_{n}
\le|\cI|\exp \left(\frac {2}{c}\sum_{k=1}^{\infty}
|\alpha_{k}|\right)\,.
\end{align*}
%
%
It remains to apply  Theorem \ref{not-s.a.}.
     \end{proof}
Our next test involves  both $X$ and $\alpha$.

\begin{corollary}\label{C2}
Let $|\cI|<+\infty$. Then the GS realization $D_{X,\alpha}$ is symmetric with $n_{\pm}(D_{X,\alpha}) =1$ whenever
  \begin{equation}\label{5.31}
\limsup_{n\to\infty}\,\frac{d_{n+1}}{d_{n}}\left(1+\frac{|\alpha_{n}|}c
\right)^{2}<1\,.
   \end{equation}
In particular,  $n_{\pm}(D_{X,\alpha}) =1$  provided that one of the following  conditions is satisfied
\item $(i)$   $\limsup_{n\to\infty}(d_{n+1}/d_{n})=0$ and the  sequence $\alpha = \{\alpha_n\}_1^\infty$ is  bounded;
\item $(ii)$  $\limsup_{n\to\infty}(d_{n+1}/d_{n})=:(1/d)$ with  $d>1$ and
$\sup_{n\in \N} \alpha_n < c(\sqrt{d}-1).$
   \end{corollary}
  \begin{proof}
By the ratio test condition \eqref{5.31} yields  the convergence of the series \eqref{non-self_cond}. It remains to apply  Theorem \ref{not-s.a.}.
  \end{proof}
\begin{remark} Note that the condition $\limsup_{n\to\infty}(d_{n+1}/d_{n})\le 1$ is always satisfied whenever the interval $\cI$ is finite. Indeed, it is implied by  the convergence of the series  $\sum_{n=1}^{\infty}d_{n}=|\cI|<+\infty$.

The following cases are more complicated and  require  more  detailed  analysis:

\item $(i)$ $\limsup_{n\to\infty}(d_{n+1}/d_{n})=0$ and the sequence $\{\alpha_{n}\}_1^\infty$ is unbounded;

\item $(ii)$  $\limsup_{n\to\infty}(d_{n+1}/d_{n})=1$  although  $\lim_{n\to\infty}\alpha_{n}=0$,

 We discuss the case (i) in the following example.
\end{remark}

 \begin{example}
\item $(i)$ Let  $|\alpha_{n}|\sim \alpha_{0}/n^{s}$,  $s>0$, and $\alpha_{0}>0$.   Then, by Corollary \ref{C1}, $n_{\pm}(D_{X,\alpha}) =1$  for any $X=\{x_n\}_1^\infty$ whenever $s>1$.

Next let $X=\{x_n\}_1^\infty$ with $x_{n}=1-1/d^{n}$, $d>1$, $n\in \N$.
Then, by Corollary \ref{C2}, $n_{\pm}(D_{X,\alpha}) =1$  for  $s\in (0,1]$ and any  $\alpha_{0}\in \R_+$ as well as for $s=0$ whenever  $\alpha_{0}< c(\sqrt d-1)$.

\item $(ii)$   Let $X=\{x_n\}_1^\infty$ with $x_{n}=1-1/n!$,  $|\alpha_{n}|\sim \alpha_{0}n^{s}$,\  $n\in \N$, $s\in \R$.  Then,  by Corollary \ref{C2}, $n_{\pm}(D_{X,\alpha}) =1$
for any  $\alpha_{0}$ whenever $s<1/2$,
and for $\alpha_{0}<c$ whenever $s=1/2$.

On the other hand, if $\alpha_n\ge (n-1)!$, then, by Proposition \ref{interval}, the operator $D_{X,\alpha}$ is self-adjoint.
  \end{example}
%
   \begin{remark}
Comparing Proposition \ref{interval}  with Theorem  \ref{not-s.a.}  one might say  that very roughly speaking
$D_{X,\alpha}$ is  self-adjoint on a finite interval  whenever the sequence $\{\alpha_n\}_1^\infty$ grows  faster than the sequence $\{d_n\}_1^\infty$ decays.
   \end{remark}

\subsubsection{Continuous spectrum and resolvent comparability}

   \begin{proposition}\label{Prop_rescompar}
Let $D_{X,\alpha^{(k)}}$  be the Gesztesy-\v{S}eba realization of
Dirac operator on the half-line $\R_+$ given  by \eqref{delta}
with $\alpha^{(k)} := \{\alpha^{(k)}_n\}_{n\in{\N}} (\subset \R),\
k\in \{1,2\} $. Let also $B_{X,\alpha^{(k)}}$  be the  Jacobi
operator defined on $\cH = l^2(\N)\otimes \C^2$ by the matrix  \eqref{IV.2.1_01} with
$\alpha^{(k)}$ in place of $\alpha$.
Then $D_{X,\alpha^{(k)}}= D^*_{X,\alpha^{(k)}}$, \  $k\in \{1,2\}$, and for any
$p\in(0,\infty]$  the inclusion      
    \begin{eqnarray}\label{rescompar1}
(D_{X, \alpha^{(1)}} - z)^{-1} - (D_{X, \alpha^{(2)}} - z)^{-1}
\in\mathfrak S_p(\gH), \qquad z\in\rho(D_{X,\alpha^{(1)}}) \cap
\rho(D_{X,\alpha^{(2)}})
    \end{eqnarray}
is equivalent to the inclusion
        \begin{eqnarray}\label{rescompar2}
(B_{X, \alpha^{(1)}} - \zeta)^{-1} - (B_{X, \alpha^{(2)}} -
\zeta)^{-1} \in\mathfrak S_p(\cH), \qquad \zeta \in\rho(B_{X,\alpha^{(1)}})
\cap \rho(B_{X,\alpha^{(2)}}).
    \end{eqnarray}
      \end{proposition}
  \begin{proof}
Consider the boundary triplet $\Pi = \{\cH, \Gamma_0, \Gamma_1\}$ defined in Theorem
\ref{th_bt_2}. Then, by Proposition \ref{prop_IV.2.1_01},  $H_{X,
\alpha^{(k)}} = H_{B_{X, \alpha^{(k)}}}$ where
$B_{X,\alpha^{(k)}}, \ k\in \{1,2\},$ is the corresponding Jacobi
operator. Since $\cI =\R_+$,  both operators $D_{X,\alpha^{(k)}}$ and
$B_{X,\alpha^{(k)}}$, $k\in \{1,2\},$ are  selfadjoint, by Proposition
\ref{cor_delta_carleman}. Therefore the resolvents $(D_{X,
\alpha^{(k)}} - z)^{-1}$ and $(B_{X,
\alpha^{(k)}} - \zeta)^{-1},$ \ $k\in \{1,2\},$ are  well defined for
any $z, \zeta\in \C_+$ and relations \eqref{rescompar1} and \eqref{rescompar2} have sense. One completes the proof by applying Proposition \ref{prop_II.1.4_02}(i).
 \end{proof}
    \begin{corollary}\label{cor4.12}
Assume the conditions of Proposition \ref{Prop_rescompar}. Assume,
in addition, that either
   \begin{equation}\label{4.27}
  \left\{\frac{\alpha^{(1)}_n -
\alpha^{(2)}_n}{d_{n+1}}\right\}_{n=1}^\infty \in l^p(\N),\quad p\in(0,\infty)
\qquad \text{or}\qquad \left\{\frac{\alpha^{(1)}_n -
\alpha^{(2)}_n}{d_{n+1}}\right\}_{n=1}^\infty \in c_0(\N).
   \end{equation}
Then the inclusion \eqref{rescompar1} holds with $p\in(0,\infty)$
and $p = \infty$,  respectively.
   \end{corollary}
   \begin{proof}
Let  $B_{X,\alpha^{(k)}}, \ k\in \{1,2\},$ be the Jacobi operator
given  by \eqref{IV.2.1_01}. Clearly,
$l^2_{0}(\N)\otimes C^2 \subset\dom(B_{X,\alpha^{(1)}})\cap
\dom(B_{X,\alpha^{(2)}}).$   It follows from representation
\eqref{IV.2.2_01A} for $B_{X, \alpha}$ and formula
\eqref{IV.1.1_19}  for $R_n$ that
%
%
\[
B_{X,\gA^{(1)}}f - B_{X,\gA^{(2)}}f =
R^{-1}_X\bigl(\widetilde{B}_{\gA^{(1)}} -
\widetilde{B}_{\gA^{(2)}}\bigr)R^{-1}_X f=
\bigoplus_{n=1}^\infty\left(\begin{array}{cc}
                                                            0          &0\\
                                                            0&                  \frac{\alpha_n^{(1)}-\alpha_n^{(2)}}{\gd_{n+1}}
                             \end{array}\right)f, \qquad f\in
l^2_{0}(\N).
\]
Due to the  assumption \eqref{4.27} the operator $B_{X,\gA^{(1)}}
- B_{X,\gA^{(2)}}$ admits the  closure and
$\overline{B_{X,\gA^{(1)}} - B_{X,\gA^{(2)}}} \in
\mathfrak{S}_p(\cH)\subset [\cH]$. Hence $\dom(B_{X,\gA^{(1)}}) =
\dom(B_{X,\gA^{(2)}})$ and,  by Proposition
\ref{prop_II.1.4_02}(ii), the inclusion \eqref{rescompar2} holds.
It remains to apply Proposition \ref{Prop_rescompar}.
   \end{proof}

Next we slightly generalize Corollary \ref{cor4.12}  allowing one
of the sequences $\gA^{(k)} = \{\gA^{(k)}_n\}_{n\in \N}$ to take infinite values. Moreover, in the case
$d_*(X)>0$ we can drop dependence on $d_n$ in \eqref{4.27}.

To state the result we set $(i+\infty)^{-1}:=0$.
%
      \begin{corollary}\label{col_rc_1}
Let $\alpha^{(1)}=\{\alpha^{(1)}_n\}^{\infty}_1\subset{\R}$ and
$\alpha^{(2)}=\{\alpha^{(2)}_n\}^{\infty}_1\subset{\R}\cup\{\infty\}$.
Then

\item $(i)$ The inclusion
    \begin{equation}\label{IV.2.3.02AA}
\left\{\left({\alpha_n^{(1)}/d_{n+1} - \I}\right)^{-1} - \left({\alpha_n^{(2)}/d_{n+1} -
\I}\right)^{-1}\right\}_{n=1}^\infty \in l^p(\N), \qquad
p\in(0,\infty),\qquad (\in c_0(\N),\quad\text{for}\quad p=\infty),
  \end{equation}
yields the inclusion \eqref{rescompar1}.

\item $(ii)$   If in addition $0<d_*(X)\leq d^*(X)<\infty$,  then
\eqref{rescompar1}  is equivalent to the inclusion
    \begin{equation}\label{IV.2.3.02}
\left\{({\alpha_n^{(1)}-\I})^{-1} -
({\alpha_n^{(2)}-\I})^{-1}\right\}_{n=1}^\infty \in
l^p(\N),\qquad 
p\in(0,\infty),\qquad (\in c_0(\N),\quad\text{if}\quad p=\infty).
  \end{equation}
Moreover, if $\{\alpha_n^{(j)}\}_{n=1}^\infty\in l^\infty(\N),\ j\in
\{1,2\},$ then \eqref{IV.2.3.02} is equivalent to the inclusion
$\{\alpha_n^{(1)}-\alpha_n^{(2)}\}_{n=1}^\infty\in l^p(\N)$\  (resp. $c_0(\N))$.
\end{corollary}
The proof is similar to that of Corollary \ref{cor4.12} and is omitted.

To state the next result we recall the  definition of the essential spectrum.
   \begin{definition}
It is said that $\lambda_0 = \overline{\lambda}_0$ belongs to the essential spectrum of the operator $T= T^* \in\cC(\gH)$
$\bigl($in short $\lambda_0\in\sigma_{\ess}(T)\bigr)$ if there exists a bounded non-compact sequence $f_n\in\mathfrak H,\  n\in\mathbb N$, such that
     \begin{equation}
     (T-\lambda_0)f_n\to 0\qquad \text{as}\qquad  n\to\infty.
     \end{equation}
\end{definition}

In the following theorem we investigate continuous and  absolutely continuous spectra
of GS-realizations comparing them with the  Neumann  realization $D_N$ of the Dirac expression
given by
     \begin{equation}\label{4.27A}
 D_N = D\upharpoonright \dom(D_N), \qquad \dom(D_N) = \{f=\binom{f_1}{f_2}\in W^{1,2}(\mathbb R_+)\otimes\mathbb C^2:\ f_2(x_0+)=0\}.
\end{equation}
We also investigate singular spectrum of GS-realizations.
   \begin{theorem}\label{th_cont_spectr}
Let $X = \{x_n\}_{n=1}^\infty (\subset \R_+)$,  $\alpha = \{\alpha_n\}^{\infty}_1\subset{\mathbb R}$ and let $D_{X,\alpha}$ be the corresponding Gesztesy-\v{S}eba operator on the
half-line ${\R}_+$. Then the following holds:

\item $(i)$  If $\{\alpha_n/d_{n+1}\}_1^\infty \in c_0(\N)$, i.e.
$\lim_{n\to\infty}\alpha_n/d_{n+1} = 0,$  then
  \begin{equation}\label{4.26}
\sigma_{\ess}(D_{X,\alpha}) = \sigma_{\ess}(D_{N}) = {\R}\setminus (-c^2/2, c^2/2),
    \end{equation}

\item $(ii)$  Assume that $\{\alpha_n/d_{n+1}\}_1^\infty \in l^1(\N)$, i.e.
$\sum_{n\in \N}|\alpha_n/d_{n+1}| < \infty.$  Then the $ac$-part
$D^{ac}_{X,\alpha}$ of  $D_{X,\alpha}$ is unitarily
equivalent to the Neumann realization $D_N$.
In particular,
      \begin{equation}\label{4.28}
\sigma_{\ac}(D_{X,\alpha}) = \sigma_{\ac}(D_{N}) = {\R}\setminus (-c^2/2, c^2/2).
    \end{equation}
%
%

\item $(iii)$  Assume that
    \begin{equation}\label{4.29}
\limsup_{n\to\infty}\frac{|\alpha_n|}{d_{n+1}} =\infty.
    \end{equation}
Then the GS-operator ${D}_{X,\alpha}$ is purely singular,
i.e.
   \begin{equation}\label{4.32}
\sigma_{ac}({D}_{X,\alpha})=\emptyset.
   \end{equation}

\item $(iv)$  Assume in addition that $d_*(X) > 0$. Then  the above
assumptions on the sequence $\{\alpha_n/d_{n+1}\}_1^\infty$ in (i), (ii), and \eqref{4.29} can be replaced by
  $$
\{\alpha_n\}^{\infty}_1\in c_0(\N), \qquad \{\alpha_n\}^{\infty}_1\in
l^1(\N) \qquad\text{and}\qquad
\limsup_{n\to\infty}{|\alpha_n|} =\infty,
  $$
respectively.
      \end{theorem}
  \begin{proof}
(i) We choose $\alpha^{(2)}:= \mathbf  0 =\{0\}^{\infty}_1$ to be a
zero sequence and set $\alpha^{(1)} := \alpha =
\{\alpha_n\}^{\infty}_1$. It is easily seen that ${D}_{X,\alpha^{(2)}} = {D}_{X, \mathbf  0}$ coincides with the Neumann realization $D_N$ given by \eqref{4.27A}.
Moreover, noting  that $D_N = D_{b+,0}$ with $b=0$ (see formula \eqref{3.20}),
Lemma \ref{2.3+}(ii) yields  $\sigma_{\ess}(D_N) = {\R}\setminus
(-c^2/2, c^2/2).$
On the other hand, due to the assumption on $\{\alpha_n/d_{n+1}\}_1^\infty$,
the above sequences $\alpha^{(1)}$ and $\alpha^{(2)}$ satisfy the second condition
in \eqref{4.27},  hence, by Corollary  \ref{cor4.12},  
$(D_{X,\alpha}-i)^{-1}  - (D_N - i)^{-1} \in\mathfrak S_\infty(\gH)$.
By the Weyl theorem (see \cite[Corollary XIII.4.1]{RedSim78}), the later inclusion implies
$\sigma_{\ess}(D_{X,\alpha}) = \sigma_{\ess}(D_N) = {\R}\setminus (-c^2/2, c^2/2).$

(ii) Now, by Corollary  \ref{cor4.12}, the condition $\{\alpha_n/d_{n+1}\}_1^\infty \in l^1(\N)$ implies  $(D_{X,\alpha} - i)^{-1}  - (D_N - i)^{-1} = (D_{X, \alpha^{(1)}} - i)^{-1} - (D_{X, \alpha^{(2)}} - i)^{-1} \in\mathfrak S_1(\gH)$.  By the Kato-Rosenblum theorem (see \cite[Chapter 10.4]{Kato66}, \cite[Theorem XI.9]{RedSim79}), the
$\ac$-part $D^{\ac}_{X,\alpha}$ of $D_{X,\alpha}$ is unitarily
equivalent to $D^{\ac}_N = D_N$. It remains to apply Lemma
\ref{2.3+}(ii) and note that $D_N = D_{b+,0}$ with $b=0$.

(iii)  According to  \eqref{4.29} there exists a subsequence
$\{\gA_{n_k}\}$ such that
      \begin{equation}\label{eq:56.04A}
\lim_{n\to\infty}\frac{|\gA_{n_k}|}{\gd_{n_k}}=\infty.
    \end{equation}
Set
    \begin{equation}\label{4.33}
\widetilde{\gA}_n := \begin{cases}
                      \gA_n,\quad n\notin \{n_k\},\\
                      \infty,\quad n\in\{n_k\}.
                      \end{cases}
   \end{equation}
and $\widetilde{\gA} := \{\widetilde{\gA}_n\}_1^{\infty}.$
Without loss of generality we assume that the subsequence $\{\alpha_{n_k}\}_{k=1}^\infty$ satisfies
     \begin{equation}\label{4.33A}
\sum^{\infty}_{k=1}d_{n_k+1}|\alpha_{n_k}|^{-1}<\infty,
  \end{equation}
i.e. $\{d_{n_k+1}\alpha^{-1}_{n_k}\}_{k\in\N}\in l^1(\mathbb N)$. Otherwise we replace $\{\alpha_{n_k}\}_{k\in\N}$ by its appropriate  subsequence.
It follows that
    \begin{equation}\label{4.34}
\sum^{\infty}_{n=1}|(\alpha_n d^{-1}_{n+1} - i)^{-1} - (\widetilde{\alpha}_n d^{-1}_{n+1} - i)^{-1}| = \sum^{\infty}_{n=1}|(\alpha_{n_k}d^{-1}_{n_k+1} - i)^{-1}|<\infty.
    \end{equation}
By  Corollary \ref{col_rc_1}(i), this relation  yields
%
%
  \begin{equation}\label{eq:56.04B}
(D_{X,\gA} - i)^{-1}- (D_{X,\widetilde\gA} - i)^{-1}\in
\mathfrak{S}_1(\mathfrak H).
    \end{equation}
According to Remark \ref{remarkGS} (see formula \eqref{delta_infty}) it follows from \eqref{4.33})
that the operator $D_{X,\widetilde\gA}$ admits the following orthogonal  decomposition
    \begin{equation}\label{eq:56.05}
D_{X,\widetilde\gA}  = \bigoplus_{k=1}^\infty D_{n_k},\qquad
L^2(\R_+,  \C^2) = \bigoplus_{k = 1}^\infty L^2\left([x_{n_{k-1}},x_{n_k}],  \C^2\right),
  \end{equation}
where
   \begin{eqnarray*}
 D_{n_k} = D\upharpoonright  \dom(D_{n_k}),\qquad
\dom(D_{n_k}) = \Big\{f\in \bigoplus_{j= n_{k-1}}^{n_k-1} W^{1,2}\left([x_j,x_{j+1}], \C^{2}\right): f_{1}(x_{j}+) = f_{1}(x_{j}-), \\  
f_{2}(x_{j}+)-f_{2}(x_{j}-) = - ic^{-1}
{\alpha_{j}}f_{1}(x_{j}),\,\, x_j\in  X\cap (x_{n_{k-1}},x_{n_k}),
\quad f_{1}({x_{n_{k-1}} ) = f_{1}({x_{n_k}}}) = 0  \Big\}.
       \end{eqnarray*}
Clearly, $D_{n_k}$ is a selfadjoint extension of the minimal
operator  $D_{n_k}' := D_{n_k, \min}$ given by
   \begin{eqnarray*}
D_{n_k}' =  D\upharpoonright\dom(D_{n_k}'),     \\
\dom(D_{n_k}') = \Big\{f\in
W^{1,2}_0([x_{n_{k-1}},x_{n_k}], \C^{2}):\  f(x_j) = 0,
\quad x_j\in X\cap [x_{n_{k-1}},x_{n_k}] \Big\}.
  \end{eqnarray*}
Clearly,  the operator $D'_{n_k}$ admits a self-adjoint extension $\bigoplus^{n_k}_{j=n_{k-1}+1}D_{n_j,0}$ with discrete spectrum (see Lemma 3.1(ii)) where $D_{n_j,0}$ is given by \eqref{3.8}.   Since $D'_{n_k}$ is a symmetric operator with finite deficiency indices, each its self-adjoint extension  has also discrete spectrum.  In particular, $D_{n_k},\ k\in \N,$ has discrete spectrum.

Therefore due to the representation \eqref{eq:56.05}   the spectrum of the operator $D_{X,\widetilde\gA}$ is purely point, in particular  $\sigma_{\ac}(D_{X,\widetilde\gA}) = \emptyset$.
On the other hand, it follows from  \eqref{eq:56.04B} and  the Kato-Rosenblum theorem
that the $\ac$-parts $D^{\ac}_{X,\alpha}$ and $D_{X,\widetilde\gA}^{\ac}$  of the operators
$D_{X,\alpha}$ and  $D_{X,\widetilde\gA}$ are unitarily equivalent.
In particular, $\sigma_{\ac}(D_{X,\alpha}) = \sigma_{\ac}(D_{X,\widetilde\gA}) = \emptyset.$

(iv) This statement is immediate from the previous ones.
     \end{proof}
%
%

Next we extend Theorem \ref{th_cont_spectr} to the case of $GS$-realizations of the Dirac
expression   $D(Q)$ with a bounded  potential matrix $Q\in L^{\infty}({\cI})\otimes{\C}^{2\times 2}.$ 
Namely, consider  differential expression
  \begin{equation}\label{4.9BB}
D(Q) := D^c(Q) := -i\,c\,\frac{d}{dx}\otimes\sigma_{1} +
\frac{c^{2}}{2}\otimes\sigma_{3} + Q(x),\qquad Q(x)=Q(x)^{*},
\end{equation}
and denote by $D_{X}(Q) := D_{X}^c(Q)$ the minimal operator associated on $\cI\setminus X$ with the expression $D^c(Q)$.
As in \eqref{3.27} one has
 \begin{equation}\label{4.39A}
D_X(Q) = D(Q)\upharpoonright \dom(D_X(Q)), \quad \dom(D_X(Q)) = W^{1,2}_0\left({\cI}\setminus X, \mathbb C^2\right)  =
\bigoplus_{n=1}^\infty W^{1,2}_0\left([x_{n-1},x_{n}], \C^{2}\right).
  \end{equation}

Further, let $D_{X,\alpha}(Q) := D_{X,\alpha} + Q$ be the GS realization of $D(Q)$.  If  $\alpha := \mathbf  0 =\{0\}^{\infty}_1$ is a zero sequence
we set  $D_N(Q) := {D}_{X, \mathbf  0}(Q)$ and note that $D_N(Q)$, the Neumann realization  of $D(Q)$, is given by the expression \eqref{4.9BB}
on  the domain  \eqref{4.27A}, i.e.
$$
 D_N(Q) = D(Q) \upharpoonright \dom(D_N(Q)), \quad \dom(D_N(Q)) = \dom(D_N) = \{f\in W^{1,2}(\mathbb R_+)\otimes\mathbb C^2:\ f_2(x_0+)=0\}.
$$

   \begin{proposition}\label{prop4.16}
Assume that $Q\in L^{\infty}({\R}_+)\otimes{\C}^{2\times 2},$\
$Q(x)=Q^*(x)$ for a.e. $x\in{\R}_+,$ and
$\alpha=\{\alpha_n\}^{\infty}_1\subset{\R}$. Then the following
holds

\item $(i)$ If $\{\alpha_n/d_{n+1}\}_1^\infty  \in c_0(\N)$, then
    \begin{equation}\label{4.44A}
\sigma_{\ess}\bigl(D_{X,\alpha}(Q)\bigr) =
\sigma_{\ess}\bigl(D_N(Q)\bigr).
  \end{equation}
Moreover, if in addition,  $Q(x)\to 0$ as $x\to\infty$, then
       \begin{equation}\label{4.44}
\sigma_{\ess}\bigl(D_{X,\alpha}(Q)\bigr) =
\sigma_{\ess}\bigl(D_N(Q)\bigr)  = \R\setminus(-c^2/2, c^2/2).
       \end{equation}

\item $(ii)$  If $\{\alpha_n/d_{n+1}\}_1^\infty \in  l^1(\N)$, then
    \begin{equation}\label{4.45A}
\sigma_{\ac}\bigl(D_{X,\alpha}(Q)\bigr) =
\sigma_{\ac}\bigl(D_N(Q)\bigr).
     \end{equation}
Moreover, if  additionally,  $Q\in L^1(\mathbb R_+)\otimes\mathbb C^{2\times 2}$, then
    \begin{equation}\label{4.45}
\sigma_{\ac}\bigl(D_{X,\alpha}(Q)\bigr) =
\sigma_{\ac}\bigl(D_N(Q)\bigr) = \R\setminus(-c^2/2, c^2/2).
     \end{equation}

\item $(iii)$ If condition \eqref{4.29}  is satisfied, then the spectrum of $D_{X,\alpha}(Q)$ is purely singular, i.e.
    \begin{equation}
\sigma_{\ac}\bigl(D_{X,\alpha}(Q)\bigr) = \emptyset.
    \end{equation}

\item $(iv)$   Assume in addition that $d_*(X) > 0$. Then  the above
assumptions can be replaced by
  $$
\{\alpha_n\}^{\infty}_1\in c_0(\N), \qquad \{\alpha_n\}^{\infty}_1\in
l^1(\N) \qquad\text{and}\qquad
 \limsup_{n\to\infty}
{|\alpha_n|} =\infty,
  $$
respectively.
    \end{proposition}
    \begin{proof}
(i)-(ii). Let $\Pi =\{\cH,\Gamma_0,\Gamma_1\}$ be a boundary triplet for the operator  $D_{X}^*$
defined in Theorem \ref{th_bt_2}.
Since $Q = Q^* \in L^{\infty}({\R}_+)\otimes {\C}^{2\times 2},$
$\Pi$ is also  the boundary triplet  for the operator  $D_{X}(Q)^*$.
Moreover, due to the inclusion  $Q\in L^{\infty}({\R}_+)\otimes {\C}^{2\times 2},$  one has
$\dom(D_{X, \alpha}(Q)) = \dom(D_{X, \alpha})$. Therefore the boundary operator
for the extension $D_{X, \alpha}(Q)$ of   $D_{X}(Q)$ coincides with  the boundary operator for the extension
$D_{X, \alpha}$ of   $D_{X}.$ Thus,  by Proposition \ref{prop_IV.2.1_01},
%
%
   \begin{equation}\label{4.30}
\dom(D_{X, \alpha}(Q)) = \dom(D_{X, \alpha}) = \{f\in W^{1,2}\left(\R_+
\setminus X,  \C^2\right): \Gamma_1f = B_{X,\gA}\Gamma_0f\},
  \end{equation}
where  $B_{X,\gA}$ is  the Jacobi operator  given  by
\eqref{IV.2.1_01}.   Therefore due to Proposition  \ref{prop_II.1.4_02}(i)        inclusion
\eqref{rescompar1} is equivalent to the inclusion
   \begin{equation}\label{4.36}
(D_{X, \alpha^{(1)}}(Q) - i)^{-1} - (D_{X, \alpha^{(2)}}(Q) -
i)^{-1} \in\mathfrak S_p(\gH), \qquad p\in (0,\infty].
     \end{equation}
We compare the realizations  $D_{X, \alpha}(Q)$  and
$D_{X,\mathbf {0}}(Q) = D_N(Q)$ using the inclusion \eqref{4.36} with $\alpha^{(1)} =\alpha$
and $\alpha^{(2)}= \mathbf {0}=\{0\}_1^\infty$.
Namely, the inclusion \eqref{4.36} with $p =\infty$ yields  \eqref{4.44A}  by applying the
Weyl theorem  \cite[Corollary XIII.4.1]{RedSim78}).
Similarly, the inclusion \eqref{4.36} with $p = 1$ yields equality \eqref{4.45A} by applying
the  Kato-Rosenblum theorem  (\cite{Kato66}, \cite[Theorem XI.9]{RedSim79}).

It is well known that $\sigma_{\ess}\bigl(D_N(Q)\bigr)  = \sigma_{\ess}\bigl(D_N\bigr) = \R\setminus(-c^2/2, c^2/2)$ provided that  $Q(x)\to 0$ as $x\to\infty$.
Combining this relation with \eqref{4.44A}  we arrive at \eqref{4.44}.

Further, according to  \cite[Theorem 9.1.1]{LevSar88},  $\sigma_{\ac}\bigl(D_N(Q)\bigr) = \sigma_{\ess}\bigl(D_N(Q)\bigr)  = \R\setminus(-c^2/2, c^2/2)$ whenever  $Q\in L^1(\mathbb R_+)\otimes\mathbb C^{2\times 2}$.  Combining this fact with  \eqref{4.45A} we arrive at \eqref{4.45}.

(iii)  As in the proof of Theorem \ref{th_cont_spectr} (iii) we define a sequence $\widetilde{\gA} := \{\widetilde{\gA}_n\}_1^{\infty}$  by formula \eqref{4.33} and  find  a subsequence $\{\widetilde{\gA}_{n_k}\}_{k=1}^{\infty}$   satisfying  \eqref{4.33A}.
Alongside \eqref{4.30} we have the following representation of the domain $\dom(D_{X, \widetilde{\gA}}(Q))$,
   \begin{equation}\label{4.30A_Domain}
\dom(D_{X, \widetilde{\gA}}(Q)) = \dom(D_{X, \widetilde{\gA}}) = \{f\in W^{1,2}\left(\R_+
\setminus X, \C^2\right): \{\Gamma_1 f, \Gamma_0f\} \in \Theta_{X,\widetilde{\gA}}\},
     \end{equation}
where the boundary relation $\Theta_{X,\widetilde{\gA}}$
corresponding to $D_{X, \widetilde{\gA}}(Q)$ does not depend on $Q$.
As it is shown in the proof of Theorem \ref{th_cont_spectr}(iii) conditions \eqref{4.29} yields the inclusion \eqref{eq:56.04B}. In turn,  combining  relations \eqref{4.30} and  \eqref{4.30A_Domain} with Proposition  \ref{prop_II.1.4_02}(i) we get that   the inclusion \eqref{eq:56.04B} yields
(in fact,  is equivalent to) the inclusion
  \begin{equation}\label{eq:56.04BBCC}
(D_{X,\gA}(Q) - i)^{-1} - (D_{X,\widetilde\gA}(Q) - i)^{-1}\in
\mathfrak{S}_1(\mathfrak H).
    \end{equation}
The rest of the proof coincides with the proof of  Theorem \ref{th_cont_spectr}(iii).
  \end{proof}
     \begin{remark}
In the case of $Q\not \equiv 0$ an explicit description of $\sigma_{\ess}\bigl(D_N(Q))$ and  $\sigma_{ac}\bigl(D_{N}(Q)\bigr)$ is known also for some  non-decaying potentials.
For instance,  if $Q$ is periodic, $Q(x + \tau)= Q(x),\  x\in \R,$ then the essential spectrum of the operator $D(Q)$ in $L^2(\mathbb R,\mathbb C^2)$ is absolutely continuous and  has zone-band structure.
This fact allows one  to complete the statement (ii) for periodic $Q$.
     \end{remark}
%
%
       \begin{remark}
Note that analogs of the main results of this section are known
for Schr\"{o}dinger operators  $H_{X,\alpha}$ with $\delta$-interactions.
For instance, in the case $d_*(X)>0$, the resolvent comparability criterion  for Schr\"{o}dinger operators $H_{X,\alpha}$ (i.e. analogs of Corollaries \ref{col_rc_1} and \ref{cor4.12}) was  obtained in \cite{Koc_89} (see also \cite{Mih_94a}).
Moreover, the statements similar to Theorem \ref{th_cont_spectr}(iv) and Proposition \ref{prop4.16}(iv)  have also  been obtained for operators $H_{X,\alpha}$ in \cite{Mih_94a, Mih_96} and \cite{Koc_89}. These authors have also applied boundary triplets technique  to the operator $H_{X}^*$ with $d_*(X)>0$.
Other results on absolutely continuous and singular spectrum of $H_{X,\alpha}$ in the case $d_*(X)>0$  can be found in \cite{Chr_Sto_94}.

In the case $d_*(X) = 0$  Schr\"{o}dinger operators $H_{X,\alpha}$ were treated in detail in \cite{KM} where one can find  analogs of Proposition \ref{Prop_rescompar}, Corollary \ref{cor4.12} and Theorem \ref{th_cont_spectr}(i)-(ii).

Note also that the  proof of Theorem \ref{th_cont_spectr}(iii) is similar to that presented in \cite{Chr_Sto_94, Mih_96}. However  the idea of the proof goes back to the paper \cite{SimSpe89}  where it is applied to 1-D Schr\"{o}dinger operators with  $L^1_{\loc}(\R_+)$-potentials.
In connection with Theorem \ref{th_cont_spectr}(iii) we mention also a recent interesting  paper \cite{Lot11}.
In particular, it is shown in \cite{Lot11} that the Schr\"{o}dinger operator with point interactions on a sparse set has purely point continuous  spectrum.
         \end{remark}
     \begin{remark}
Another proof of Theorem \ref{th_cont_spectr}(iii)  can also be extracted from \cite[Theorem 1.1]{MalNei11}.  It is based on an explicit block-diagonal form ${M}(\cdot) = \bigoplus_{n=1}^\infty {M}_{n}(\cdot)$ of the Weyl function ${M}(\cdot)$.
      \end{remark}
%
%

\subsubsection{Discrete spectrum}

\noindent Here we  investigate the discreteness property of proper extensions  
of the minimal Dirac operator $D_X(Q)$ defined by \eqref{4.39A}
and  associated in $\mathfrak H = L^2({\R}_+)\otimes{\C}^{2}$ with  the differential expression  \eqref{4.9BB}
on $\R_+\setminus X$.                             
In particular, we show that in the case $d_{*}(X)>0$
there are no proper extensions with discrete spectrum.

First we investigate the discreteness property for the
minimal Dirac operator  $D_X := D_X(0)$ with zero potential $Q=0$.
   \begin{theorem}\label{deltadiscr}
Let $X = \{x_n\}_1^\infty (\subset \R_+)$ and
let $\Pi = \{\mathcal{H}, \Gamma_0, \Gamma_1\}$  be the  boundary triplet for $D_X^*$ defined in Theorem \ref{th_bt_2}
and let $\Theta\in\widetilde \cC(\cH)$ with $\rho(\Theta)\not = \emptyset.$
Then:
%
%
\item $(i)$ $D_X$ has compact inverse $(D_X)^{-1}(\in [\ran(D_X),\mathfrak H])$
if and only if $\lim_{n\to+\infty}d_{n}=0$.
\item $(ii)$ A proper  extension
$\widetilde{D_X} = D_{X,\Theta}$ of  $D_X$
has discrete spectrum if and only if
$\lim_{n\to+\infty}d_{n}=0$ and
$\Theta\in \widetilde\cC(\cH)$ has discrete spectrum.
  \end{theorem}
        \begin{proof} (i) \emph{Sufficiency}.
Let $\lim_{n\to+\infty}d_{n}=0$.
According to the construction,  the boundary triplet $\Pi =\{\cH,\Gamma_0,\Gamma_1\}$   for $D_X^{*}$
is the direct sum, $\Pi=\bigoplus_{n=1}^\infty\Pi_{n}$ (see  Theorem \ref{th_bt_2},
formulae \eqref{IV.1.1_12}, \eqref{IV.1.1_12.1})  and
  \begin{equation}\label{32Ad}
A_0 := D_X^{*}\upharpoonright \ker\Gamma_{0}
=\bigoplus_{n=1}^\infty A_{n,0},\qquad
A_{n,0} = D_{n,0}:=D^{*}_{n}\upharpoonright\ker(\Gamma_{0}^{(n)}).
   \end{equation}
%
%
Combining \eqref{32Ad} with Lemma \ref{2.3}(ii) we get
   \begin{equation}\label{33Ad}
\sigma(A_{0})=\bigcup_{n=1}^\infty\sigma(D_{n,0})=
\bigcup_{n,j=1}^\infty\{\lambda_{n,j}^{\pm}\}
  \end{equation}
where for any fixed $n$
   \begin{equation}\label{33d}
\lambda^{\pm}_{n,j} =
\pm\sqrt{\frac{
c^2\pi^2}{d^2_n}\,\left(j+\frac12\right)^{2}+\left(\frac{c^{2}}{2}\right)^{2}}
\sim\pm
\frac{\pi\,c\,}{d_{n}}\left(j+\frac12\right),    \qquad\textrm{as}\qquad j\to\infty.
   \end{equation}

\noindent It follows that any non-zero (finite or infinite) accumulation point of the sequence
$\{d_{n}\}_{n=1}^\infty$ generates  countably many accumulation points
for\  the sequence $\{\lambda^{\pm}_{n,j}\}_{n,j\in \N}$. Thus, the spectrum
$\sigma(A_{0})$ is discrete, i.e. $A_{0}^{-1}\in \mathfrak S_{\infty}(\mathfrak H)$ if and only if
$\lim_{n\to+\infty}d_{n}=0$. In particular, the later condition  yields compactness of  $(D_X)^{-1} = A_{0}^{-1} \upharpoonright \ran (D_X) $.

\emph{Necessity.} Assume that $d_{n}$ does not converge  to
zero, so that we can find a subsequence $\{d_{n_{k}}\}_{k=1}^{\infty}$
and $\varepsilon>0$ such that $d_{n_{k}}\geqslant\varepsilon>0$, $k\in\N$.
Choose a function $\varphi = \binom{\phi}{\phi}\in W^{1,2}_{0}(\R_+, \mathbb C^2)$ such that
\[\phi(x) = \left\{\begin{array}{cc}1, &\qquad \varepsilon/4\leqslant x\leqslant 3\varepsilon/4, \\0, &\qquad x\notin [0,\varepsilon],
\end{array}\right.\]
and put
 \[\varphi_k(x):= \varphi(x-x_{n_{k}}),\qquad k\in\N.\]
 Clearly, $\varphi_{k}\in\dom(D_X)$, $k\in\N$. Moreover, there exist constants $C_1$ and $C_2$
 such that
  \begin{equation}\label{34d}
  \|\varphi_{k}\|_{L^{2}(\R_+)}= C_1 \qquad
\textrm{and}\qquad\|D_X\varphi_{k}\|_{L^{2}(\R_+, \C^2)} \le C_2, \qquad k\in \N.
  \end{equation}
 \noindent
 Since the functions $\varphi_{k}$ have disjoint supports, the sequence $\{\varphi_{k}\}_{1}^{\infty}$ is not compact in $L^{2}(\R_+)\otimes \mathbb C^2$. Therefore it follows from  \eqref{34d} that the operator $(D_X)^{-1}$ is not compact.

(ii) Let the spectrum  $\sigma(D_{X,\Theta})$ of  $D_{X,\Theta}$ be discrete,
i.e. $\rho(D_{X,\Theta})\neq\emptyset$ and
$(D_{X,\Theta} - z)^{-1}\in{\mathfrak S}_\infty(\mathfrak{H})$
for $z\in\rho(D_{X,\Theta})$.
Then $z\in\widehat{\rho}(D_X)$ and the operator
  $$
\left(D_{X} - z\right)^{-1} = (D_{X,\Theta} - z)^{-1}\upharpoonright \ran\left(D_{X} - z\right)
$$
is also compact. By (i)  $\lim_{n\to+\infty}d_{n}=0$. Therefore it follows from (\ref{33d}) and
(\ref{33Ad}) (and was already mentioned) that the spectrum
$\sigma(A_{0})$ is discrete. Since both operators $D_{X,\Theta}$  and $A_{0}$
have compact resolvents,  it follows from  Proposition \ref{prop_II.1.4_02}(i), 
that  $(\Theta - \zeta)^{-1}\in{\mathfrak S}_\infty(\cH)$ for   $\zeta\in \rho(\Theta)$,
i.e.  the spectrum $\sigma(\Theta)$ of $\Theta$ is discrete too.

Conversely, let $\lim_{n\to+\infty}d_{n}=0$
and let the spectrum $\sigma(\Theta)$ be discrete.
Then, by (i), the condition $\lim_{n\to+\infty}d_{n}=0$ yields discreteness of the spectrum  of $A_0$.
Finally,  by Proposition \ref{prop_II.1.4_02}(i), $(D_{X,\Theta} - z)^{-1}\in
{\mathfrak S}_\infty(\mathfrak{H})$ since both resolvents $(A_0 -z)^{-1}$  and  $(\Theta - \zeta)^{-1}$ are compact.
      \end{proof}
     \begin{corollary}\label{cor_discreteness}
Assume the conditions of Theorem \ref{deltadiscr}.  Let also  $Q(\cdot) = Q^*(\cdot) \in L^2_{\loc}({\R}_+)\otimes{\C}^{2\times 2},$
and let  $D_X(Q)$ be a  minimal Dirac operator  on $\R_+\setminus X$ given by  \eqref{4.39A}.
Assume in addition  that the multiplication operator $f\to Qf$ in $L^2(\R_+, \C^2)$ is strongly subordinated to the Dirac operator $D_X^*$, i.e. $\dom(D_X^*)\subset \dom(Q)$ and there exist constants
$a\in (0,1), \   b>0$, such that
     \begin{equation}\label{5.58}
\|Qf\|_{L^2(\R_+,\C^2)} \le  a\|D_X^*f\|_{L^2(\R_+,\C^2)} + b\|f\|_{L^2(\R_+,\C^2)},\quad 0< a <1, \quad f\in \dom(D_X^*).
     \end{equation}
Then:
\item $(i)$  $D_X(Q)$ has compact inverse $(D_X(Q))^{-1}(\in [\ran(D_X(Q)),\mathfrak H])$
if and only if  $\lim_{n\to+\infty}d_{n}=0$.
\item $(ii)$ A proper  extension
$\widetilde{D_X}(Q) = D_{X,\Theta}(Q) := D_{X,\Theta} +Q$ of  $D_X(Q)$
has discrete spectrum if and only if
$\lim_{n\to+\infty}d_{n}=0$ and
$\Theta (\in \widetilde\cC(\cH))$ has discrete spectrum.
In particular, both statements are satisfied whenever  $Q(\cdot) = Q^*(\cdot) \in L^\infty({\R}_+)\otimes{\C}^{2\times 2}.$
\end{corollary}
    \begin{proof}
(i)  Since $Q$ is strongly subordinated to $D_X^*$ it is also strongly  subordinated to
its restriction  $A_0 := D_X^{*}\upharpoonright \ker\Gamma_{0} = A_0^*,$  (see \eqref{32Ad}).
The latter yields boundedness of the operator $Q(A_0 - i)^{-1}$. Moreover, by the Kato-Rellich theorem (\cite[Theorem 5.4.3]{Kato66}),  $A_0 + Q$ is self-adjoint.

Further, since $Q$ is strongly subordinated to $D_X^*$ it is also subordinated to $D_X^*(Q)= D_X^* + Q$
  (see \cite[Chapter 4.1]{Kato66}), hence  $Q$ is also  subordinated to  $A_0(Q) = A_0+Q,$ the restriction of
$D_X^*(Q)$, with the $(A_0+Q)$-bound not exceeding $a(1-a)^{-1}$, i.e.
 $$
\|Qf\|_{L^2(\R_+,\C^2)} \le  (1-a)^{-1}\left(a \|(A_0 + Q)f\|_{L^2(\R_+,\C^2)} + b\|f\|_{L^2(\R_+,\C^2)}\right),\qquad  f\in \dom(A_0 + Q).
$$
Since  $A_0 + Q$ is self-adjoint,  the latter is amount to saying that  the operator
 $Q(A_0 + Q -i)^{-1}$ is  bounded. Therefore it follows from the identity
   \begin{eqnarray}\label{5.59}
(A_0 - i)^{-1} - (A_0(Q) - i)^{-1} =   \left(A_0(Q) - i\right)^{-1} Q\,(A_0 - i)^{-1} =  (A_0 - i)^{-1}Q\,  \left(A_0(Q) - i\right)^{-1}
     \end{eqnarray}
that the operators $(A_0-i)^{-1}$ and $\bigl(A_0(Q)-i)^{-1}$ are compact only simultaneously. It remains to apply Theorem \ref{deltadiscr}(i).

(ii) This statement is immediate from Theorem \ref{deltadiscr}(ii) and formula \eqref{5.59} with
$A_0$ and $A_0(Q)$ replaced by $D_{X,\Theta}$ and $D_{X,\Theta}(Q)$, respectively.
%
  \end{proof}

Next we stand the "individual" version of Corollary \ref{cor_discreteness}.
     \begin{corollary}\label{cor_discretenessNew}
Assume the conditions of Theorem \ref{deltadiscr} and let   $Q(\cdot) = Q^*(\cdot) \in L^2_{\loc}({\R}_+)\otimes{\C}^{2\times 2}.$
Assume in addition  that the multiplication operator $f\to Qf$ in $L^2(\R_+, \C^2)$
is strongly subordinated to a realization $D_{X, \Theta} = D_{X, \Theta}^*$,
i.e. $\dom(D_{X, \Theta})\subset \dom(Q)$ and the following estimate holds
     \begin{equation}\label{5.58NEW}
\|Qf\|_{L^2(\R_+,\C^2)} \le  a\|D_{X, \Theta}f\|_{L^2(\R_+,\C^2)} + b\|f\|_{L^2(\R_+,\C^2)},\quad 0< a <1, \quad f\in \dom(D_{X, \Theta}).
     \end{equation}
Then  a proper  extension
$D_{X,\Theta}(Q) := D_{X,\Theta} + Q$ of the minimal Dirac operator $D_X(Q)$ on $\R_+\setminus X$ (see  \eqref{4.39A}) has discrete spectrum if and only if
$\lim_{n\to+\infty}d_{n}=0$ and $\Theta (\in \widetilde\cC(\cH))$ has discrete spectrum.
    \end{corollary}
The proof is similar to that of Corollary \ref{cor_discreteness} and is omitted.
%
    \begin{remark}\label{rem5.28}
\item $(i)$  Sufficiency in Theorem \ref{deltadiscr}(i) can  easily be proved directly.
Indeed, since $\dom(D_X)=W^{1,2}_{0}(\R\backslash X)\otimes\C^{2}$,
it suffices to show that the identical embedding $W_{0}^{1,2}(\R\backslash
X, \C^{2}) \hookrightarrow  L^{2}(\R, \C^{2})$ is compact provided that $\lim_{n\to\infty}\gd_n=0$. Let $f$ belong to the unit ball of
$W_{0}^{1,2}(\R\backslash X)$.  One has
    \begin{equation}\label{4.40A}
\left|f(x)\right|^{2}=\left|f(x)-f(x_{n})\right|^{2}=\left|\int_{\Omega_{n}}f'(t)dt\right|^{2}\leqslant
d_n\int_{\Omega_{n}}\left|f'(t)\right|^{2}dt,\quad
x\in\Omega_{n}=[x_{n-1},x_{n}].
   \end{equation}
\noindent Choosing  any $\varepsilon>0$ we find $N\in\N$ such that
$d_{n}\leqslant\varepsilon$. Therefore it follows from \eqref{4.40A} that
\[
\sum_{|n|\geqslant
N}\int_{\Omega_{n}}\left|f(t)\right|^{2}dt\leqslant
\sum_{|n|\geqslant
N}d_{n}^{2}\|f'\|^{2}_{L^{2}(\Omega_{n})}\leqslant\varepsilon^{2}\|f'\|^{2}_{L^{2}(\R_+)}\leqslant\varepsilon^{2}.
\]
\noindent Thus the ``tails'' of functions $f$ running through  the unit ball
of $W^{1,2}_{0}(\R\backslash X)$ are uniformly small in  $L^{2}(\R)$. It remains
to note that the embedding $W^{1,2}[a,b]\otimes\C^{2}\hookrightarrow  L^{2}[a,b]\otimes\C^{2}$ is
compact for any finite interval $[a,b]$.

\item $(ii)$  As it is clear from the proof, Corollary \ref{cor_discreteness} remains valid under weaker assumptions.
Namely, condition  \eqref{5.58}  can be replaced by the following one:
$Q$ is subordinated (in the sense of  \cite[Chapter 4.1]{Kato66}) to both operators $D_X^*$ and $D_X^* +Q.$

Note also that an alternative proof of Corollary \ref{cor_discreteness}  can be obtained as follows. Equipping $\dom(D_{X,\Theta} +Q)$ with the graph norm one obtains the Hilbert space $\mathfrak H_+(\Theta, Q)$. It follows from  estimate  \eqref{5.58}  that the Hilbert spaces  $\mathfrak H_+(\Theta, Q)$ and $\mathfrak H_+(\Theta):= \mathfrak H_+(\Theta, 0)$ coincide algebraically and topologically (see  \cite[Chapter 4.1]{Kato66}). Thus, both embeddings $\mathfrak H_+(\Theta, Q)\hookrightarrow\mathfrak H$ and $\mathfrak H_+(\Theta) \hookrightarrow\mathfrak H$ are compact only simultaneously. But the compactness of the embedding $\mathfrak H_+(\Theta, Q)\hookrightarrow\mathfrak H$ is equivalent to the discreteness of the spectrum of  $D_{X,\Theta}$.
   \end{remark}

Theorem \ref{deltadiscr} establishes a connection  between the
discreteness property of extensions  $D_{X,\Theta}(Q)$ of $D_X(Q)$
and the same property of the corresponding  boundary relations $\Theta$
with respect to the  boundary triplet for $D_X^*$ defined in  Theorem \ref{th_bt_2}.
 Now we are in position to investigate discreteness property for $GS$-realizations $D_{X,\gA}^c(Q)$ in terms of the corresponding  distances $d_n$ and intensities $\alpha_n$.
To this end we exploit a connection between the $GS$-realizations and Jacobi matrices
on the one hand and the  known results on discreteness property of Jacobi matrices on the other hand.
   \begin{proposition}\label{prop_IV.2.4_01}
Let  $X = \{x_n\}_1^\infty (\subset \R_+)$,  $\alpha=\{\alpha_j\}^{\infty}_1\subset{\mathbb R}$
and let $Q(\cdot) = Q^*(\cdot)\in L^2_{\loc}({\R}_+)\otimes{\C}^{2\times 2}$ be strongly subordinated to the GS-realization $D_{X,\gA}^c = D_{X,\gA}^c(0)$ on $\R_+$.
Assume also that $\lim_{n\to\infty}\gd_n=0$  and
%
    \begin{equation}\label{prop_chihara_1}
\lim_{n\to\infty}\frac{|\gA_n|}{\gd_n}=\infty\qquad
\text{and}\qquad \lim_{n\to\infty}\frac{c}{\alpha_n} > -
\frac{1}{4}.
  \end{equation}
Then the $GS$-operator $D_{X,\gA}^c(Q)$ on the half-line $\R_+$ has discrete spectrum.
     \end{proposition}
   \begin{proof}
First we consider the case of the Dirac operator $D_{X,\gA}^c$ with zero potential $Q=0$.
Since $\lim_{n\to\infty}\gd_n=0$, one has
   \begin{equation}\label{eq1}
\lim_{n\to\infty}\frac{1}{\alpha_n\sqrt{d_n^2+1/c^2}} =
\lim_{n\to\infty}\frac{c}{\alpha_n}.
   \end{equation}
By the Carleman test (see Proposition \ref{cor_delta_carleman}), the Jacobi matrix
$B_{X,\alpha}'$ given by \eqref{IV.2.1_01'}  is self-adjoint.
Therefore,  by  \cite[Theorem 8]{Chi62}, the operator $B'_{X,\gA}$  has discrete spectrum  provided that  $\lim_{n\to\infty}\gd_n=0$ and conditions \eqref{prop_chihara_1}
are satisfied.

Next we consider the boundary triplet $\Pi=\{\cH,\Gamma_0,\Gamma_1\}$ for the operator $D_X^*$ constructed in Theorem \ref{th_bt_2}.  By Proposition \ref{prop_IV.2.1_01}, the boundary operator corresponding to the $GS$-realization $D_{X,\alpha}= D^c_{X,\alpha}$ is given by the Jacobi operator $B_{X,\alpha}$ of the form \eqref{IV.2.1_01},  \eqref{IV.2.2_01A}.
 Since the operators $B_{X,\alpha}$ and $B'_{X,\alpha}$ are unitarily equivalent (see Remark \ref{remark4.4}), the spectrum of $B_{X,\alpha}$ is discrete too.
To prove the discreteness property of the operator $D_{X,\alpha}$ it remains to  apply  Theorem \ref{deltadiscr}.

Let now $Q\not = 0.$ Since   $Q(\cdot)$ is strongly subordinated to $D_{X,\gA}^c$,
general case is reduced to the previous one by applying Corollary \ref{cor_discretenessNew}.
       \end{proof}

To apply Proposition \ref{prop_IV.2.4_01} to GS operators $D_{X,\gA}^c(Q)$  with
certain unbounded potentials we establish the following analog of the classical Hardy inequality.
%
%
\begin{lemma}\label{lemma5.30}
Assume that  $d^*(X)<\infty.$ Then for any $f\in W^{1,2}(\mathbb R_+\setminus X)$ and satisfying  $f(0)=0$ the following
inequality  holds
%
%
    \begin{equation}\label{5.74}
\int_0^{\infty}\frac{|f(x)|^2}{x^2}dx\le\frac{1}{4}\int_0^{x_1}|f'(x)|^2 dx + \frac{2}{x^2_1}\left( 3d^*(X)^2\int^{\infty}_{x_1}|f'(x)|^2 dx + 2\int^{\infty}_{x_1}|f(x)|^2 dx\right).
    \end{equation}
 \end{lemma}
 \begin{proof}
Indeed, by the classical Hardy inequality,
    \begin{equation}\label{5.75}
\int^{x_1}_0\frac{|f(x)|^2}{x^2}dx\le\frac{1}{4}\int^{x_1}_0|f'(x)|^2 dx, \qquad f\in W^{1,2}[0, x_1],\quad  f(0)=0.
    \end{equation}
%
%
Further, since $f\in W^{1,2}[x_k, x_{k+1}]$ for any $k\in\mathbb N$, one easily gets
    \begin{eqnarray*}
\int^{x_{k+1}}_{x_k}\frac{|f(x)|^2}{x^2}dx \le 2  d_k \left(\frac{|f(x_k -)|^2}{x_k x_{k+1}} + \int^{x_{k+1}}_{x_k}\frac{dx}{x^2}\int^x_{x_k}|f'(t)|^2dt\right)
\le \frac{2d_k}{x_k x_{k+1}}\left(|f(x_k-)|^2 + d_k\|f'\|^2_{L^2(\Delta_k)}\right)\\  \le \frac{2}{x_k x_{k+1}}\left(d_k|f(x_k-)|^2 + d^*(X)^2\|f'\|^2_{L^2(\Delta_k)}\right), \qquad k\ge 1.
    \end{eqnarray*}
Taking a sum of these inequalities and applying Proposition \ref{prop3.6}(ii) 
(see  formula \eqref{3.45}) we obtain
    \begin{equation}\label{5.76a}
\int^{\infty}_{x_1}\frac{|f(x)|^2}{x^2}dx  \le \frac{2}{x^2_1}\left( 2d^*(X)^2\int^{\infty}_{x_1}|f'(x)|^2 dx + 2\int^{\infty}_{x_1}|f(x)|^2 dx + d^*(X)^2\int^{\infty}_{x_1}|f'(x)|^2 dx\right).
\end{equation}
Combining \eqref{5.75} with \eqref{5.76a} we  arrive at \eqref{5.74}.
 \end{proof}
  \begin{example}\label{example5.30}
Let us present an example of GS operator  $D_{X,\gA}^c(Q) = D_{X,\gA}^c + Q$ with an unbounded  potential matrix  $Q(\cdot) = Q^*(\cdot)\notin L^\infty({\R}_+)\otimes{\C}^{2\times 2}$
satisfying conditions of Proposition \ref{prop_IV.2.4_01}.
Let $Q(\cdot)=\diag\bigl(q_1(\cdot),q_2(\cdot)\bigr)$ with  $q_1(\cdot)\in L^{\infty}(\mathbb R_+)$, and an unbounded  measurable function $q_2(\cdot)$ satisfying
    \begin{equation}\label{5.76}
|q_2(x)|\le\frac{C_0}{x^{\gamma}},\qquad x\in \R_+,\qquad \gamma\in (0,1],\qquad  C_0>0.
    \end{equation}
Let us show  that for sufficiently small $C_0$ the multiplication operator with the matrix $Q(\cdot)$ is strongly subordinated to the operator $D^c_{X,\alpha}$, i.e. that $\dom(D_{X,\gA}^c(Q)) \subset\dom(Q)$  and inequality  \eqref{5.58NEW} holds for  $f=
\begin{binom}
{f_1}{f_2}
\end{binom} \in\dom(D_{X,\alpha})$.
Since  $q_1(\cdot)\in L^{\infty}(\mathbb R_+)$, it suffices to estimate $\|q_2f_2\|_{L^2(\R_+)}$.
Noting that  $f_2(0)=0$ for  $f= \begin{binom} {f_1}{f_2} \end{binom} \in\dom(D_{X,\alpha})$,  and
combining inequality   \eqref{5.76} with Lemma \ref{lemma5.30}, we get 
$$
\|q_2f_2\|_{L^2(\R_+)}\le \widetilde C_0\|f'_2\|^2_{L^2(\R_+)} +  4C_0 x_1^{-2}\|f_2\|^2_{L^2(\R_+)},
$$
where $\widetilde C_0 =  C_0\max\{4^{-1}, 6x_1^{-2}d^*(X)^2\}.$
Since $q_1(\cdot)\in L^{\infty}(\mathbb R_+)$, this estimate implies \eqref{5.58}  whenever $C_0$ is sufficiently small.

Thus, the operator  $D^c_{X,\alpha}(Q) = D^c_{X,\alpha} +Q$ is self-adjoint and has discrete  spectrum provided that  $C_0$ is  small enough and conditions \eqref{prop_chihara_1} are satisfied. Note  that strong subordination of the operator $Q$ holds although $q_2\notin L^2(0, \varepsilon)$ for $\gamma\ge 1/2$.
   \end{example}
   \begin{remark}
For any fixed $c>0$  conditions \eqref{prop_chihara_1} are weaker than the
corresponding conditions  for the discreteness of Schr\"odinger operator $\rH_{X,\alpha}$  from \cite[Propositions 5.18]{KM} that read as follows
%
%
    \begin{equation}\label{prop_chihara_1H}
\lim_{n\to\infty}\frac{|\gA_n|}{\gd_n}=\infty\qquad
\text{and}\qquad \lim_{n\to\infty}\frac{1}{d_n\alpha_n} > -
\frac{1}{4}.
  \end{equation}
They can be obtained by taking the formal limit as $c\to +\infty$ in the left-hand side of  \eqref{eq1} with account of \eqref{prop_chihara_1}.

Note also that if $\alpha$ is negative, conditions \eqref{prop_chihara_1} do not guaranty  discreteness for the whole family $D_{X,\alpha}^c, \ c>0$,  of GS-realizations.
\end{remark}
  \begin{example}
Let $\cI =\R_+$ and let the sequence $X=\{x_n\}_{n=1}^\infty$ be given by $x_n=\log (n+1)$, so that
$d_n=\log(1+\frac1n) \thicksim \frac1n$. By Proposition \ref{cor_delta_carleman}, the $GS$-operator $D_{X,\alpha}^c$ is self-adjoint for any sequence $\alpha=\{\alpha_n\}_1^\infty\subset\R\cup\infty$. By Proposition \ref{prop_IV.2.4_01},
the $GS$-operator  $D_{X,\gA}^c$ has discrete spectrum whenever
\begin{equation}\label{example1}
\lim_{n\to\infty}n\,|\gA_n|=\infty\qquad
\text{and}\qquad 4c\lim_{n\to\infty}{\alpha_n^{-1}} > -1.
  \end{equation}

 It is interesting to compare  the $GS$-operator  $D_{X,\gA}^c$ with the corresponding Schr\"odinger operator  $\rH_{X,\gA}$ with $\delta$-interactions (see formula \eqref{Hdelta} below).
 Since $\{d_n\}_1^\infty\in l^2(\N)$,  the selfadjointness  of   $\rH_{X,\gA}$ depends on $\alpha= \{\alpha_n\}_1^\infty$
 (see \cite[Example 5.12]{KM}). Moreover, the pair of discreteness conditions  \eqref{prop_chihara_1H} for $\rH_{X,\gA}$  turn into
$\lim_{n\to\infty}n\,|\gA_n|=\infty\  \text{and}\  4 \lim_{n\to\infty}{n}{\alpha_n^{-1}} > -1\,.$

Note that if $\alpha$ is positive, i.e. $\alpha= \{\alpha_n\}_1^\infty\subset \R_+$, both pairs of conditions in \eqref{prop_chihara_1} and  \eqref{prop_chihara_1H} are reduced to the first common condition $\lim_{n\to\infty}\frac{\gA_n}{\gd_n}=\infty.$
At the same time, if $\alpha$ is negative conditions \eqref{prop_chihara_1} and  \eqref{prop_chihara_1H} are quite different.
For instance, $\rH_{X,\gA}$   is discrete  (and self-adjoint)
whenever  $\alpha_n = -(4 + \varepsilon)(n + \frac12) + O(\frac1n)$, $\varepsilon>0$ (see \cite[Example 5.12 (ii) and Proposition 5.18]{KM} $\rH_{X,\gA}$).
On the other hand, $D_{X,\gA}^c$ is discrete  provided that $\alpha_n = -(4c + \varepsilon)+O(\frac1n)$.
   \end{example}
%
%


\subsection{GS-realizations $D_{X,\beta}$: parametrization by Jacobi matrices and spectral properties}

Let  $X =\{x_n\}_1^\infty (\subset \R_+)$ be as above  and let $D_X$ be the minimal Dirac operator
given by \eqref{3.26}, \eqref{3.27}.
In this section we discuss some spectral properties  of GS-realizations $D_{X,\beta}$.
We  compute the corresponding boundary operator $B_{X,\beta}$ parameterizing  $D_{X,\beta}$ in the boundary triplet  $\Pi=\{\cH,\Gamma_0,\Gamma_1\}$ for $D_X^*$  constructed in Theorem \ref{th_bt_2},
and show that the spectral properties of $D_{X,\beta}$ are similar to that  of the GS-operators  $D_{X,\alpha}$.
In what follows we confine ourselves to the case of operators $D_{X,\beta}$ only, although the most part of the results  remains valid for operators $D_{X,\beta}(Q)$ depending on a potential matrix $Q\not =0.$

Consider the boundary triplet $\Pi=\{\cH,\Gamma_0,\Gamma_1\}$ for $D_X^* = \bigoplus_{n=1}^\infty D_{n}^{*}$  constructed in Theorem \ref{th_bt_2}
(see \eqref{IV.1.1_12} for the definitions of  $\Gamma_0$ and $\Gamma_1$).
Since $\beta_n\neq 0, \ n\in\N$, the operator $D_{X,\gB}$ is disjoint with the self-adjoint extension  $D^*_{X}\upharpoonright \ker(\Gamma_0)$ of the minimal Dirac operator $D_X$.
Here $\Gamma_0$ and $\Gamma_1$ are determined by \eqref{IV.1.1_12}.  Therefore, the boundary  relation $\Theta$
parameterizing  $D_{X,\gB}$ in  the triplet $\Pi$
is a closed operator, $\Theta\in \cC(\cH)$.

Consider the following 
Jacobi matrix
\begin{equation}\label{IV.3.1_04}
B_{X,\gB}:=\left(\begin{array}{ccccccc}
0&-\frac{\nu(\gd_{1})}{\gd_{1}^{2}}&0&0&0
&0
&\dots\\
-\frac{\nu(\gd_{1})}{\gd_{1}^{2}}& - \frac{\nu^{2}(\gd_{1})}{d_{1}^{3}}\left(\beta_{1} + \gd_{1} \right) &\frac{\nu(\gd_{1})}{\gd_{1}^{3/2}\,{\gd_2}^{1/2}}&0&0
&0
&\dots\\
0 & \frac{\nu(\gd_{1})}{\gd_{1}^{3/2}{\gd_2}^{1/2}} & 0 &-\frac{\nu(\gd_{2})}{\gd_{2}^{2}}&0
&0
&\dots\\
0&0&-\frac{\nu(\gd_{2})}{\gd_{2}^{2}}& - \frac{\nu^{2}(\gd_{2})}{d_{2}^{3}}\left(\beta_{2}+ \gd_{2} \right) & \frac{\nu(\gd_{2})}{\gd_{2}^{3/2}{\gd_{3}^{1/2}}}
&0
&\dots\\
0&0&0&\frac{\nu(\gd_{2})}{\gd_{2}^{3/2}{\gd_{3}^{1/2}}}
&0
&-\frac{\nu(\gd_{3})}{\gd_{3}^{2}}
&\dots\\
0&0&0&0&-\frac{\nu(\gd_{3})}{\gd_{3}^{2}}
&-\frac{\nu^{2}(\gd_{3})}{d_{3}^{3}}\left(\beta_{3} + \gd_{3}\right)&\dots\\
\dots&\dots&\dots&\dots&\dots&\dots&
\end{array}\right),
 \end{equation}
where \  $\nu(x):={1}/{\sqrt{1+\frac{1}{c^2x^2}}}.$
    Note that
\begin{equation}\label{IV.3.1_02B}
B_{X,\gB} = R_X^{-1}(\widetilde{B}_{\gB}-Q_X)R_X^{-1},\qquad \ \
\widetilde{B}_{\gB}=\left(\begin{array}{cccccc}
0 & 0 & 0& 0&  0& \cdots\\
0 & -\beta_1 & 1 & 0& 0&  \dots\\
0 & 1  & 0 & 0& 0&  \dots\\
0 & 0 & 0& -\beta_2 & 1&  \dots\\
0 & 0  & 0 & 1& 0&  \dots\\
\dots & \dots & \dots&  \dots & \dots
\end{array}\right),
\end{equation}
where  $R_X = \bigoplus_{n_1}^\infty R_n$  and  $Q_X = \bigoplus_{n=1}^\infty Q_n$ are determined by \eqref{IV.1.1_19}.

We also denote by $B_{X,\gB}$ the minimal (closed) Jacobi operator\  associated in  $l^2(\N,\C^2)$ with the Jacobi  matrix \eqref{IV.3.1_04}. Clearly,  $B_{X,\gB}$ is symmetric and according to general
properties of Jacobi operators $\mathrm{n}_+(B_{X,\beta})=\mathrm{n}_-(B_{X,\beta})\leq 1$.

%
%
  \begin{proposition}\label{prop_IV.3.1_01}
Let $D_X = \bigoplus_{n=1}^\infty D_{n}$  be the minimal Dirac operator in  $L^{2}(\cI,  \C^{2})$.
Let also $\Pi=\{\cH,\Gamma_0,\Gamma_1\}$ be the boundary triplet for $D_X^* = \bigoplus_{n=1}^\infty D_{n}^{*}$ constructed in Theorem \ref{th_bt_2} and let $B_{X,\gB}$ be the
minimal Jacobi operator associated in  $l^2(\N,\C^2)$  with the
matrix \eqref{IV.3.1_04}. Then  the boundary operator corresponding to the GS-realization
$D_{X,\gB}$ in the triplet  $\Pi$, is the Jacobi operator $B_{X,\gB}$, i.e.
  \begin{equation}\label{IV.3.1_03}
D_{X,\gB} = D_{B_{X,\gB}}:=D_{X}^*\upharpoonright\dom(D_{B_{X,\gB}}),
\qquad \dom(D_{B_{X,\gB}}) :=\{f\in\dom(D_{X}^*):\ \Gamma_1f = B_{X,\gB}\Gamma_0f\}.
  \end{equation}
  \end{proposition}
     \begin{corollary}\label{th_delta_prime}
The GS-operator $D_{X,\beta}$ has equal deficiency indices and
$\mathrm{n}_+(D_{X,\beta})=\mathrm{n}_-(D_{X,\beta})\leq 1$.
Moreover,
$\mathrm{n}_\pm(D_{X,\beta})=\mathrm{n}_\pm(B_{X,\beta})$. In particular,
$D_{X,\beta}$ is selfadjoint if and only if $B_{X,\beta}$ is.
      \end{corollary}
\begin{proof}
The proof is implied  by combining Proposition  \ref{prop_IV.3.1_01}
with Corollary  \ref{s.a.Dirac}(ii)  and the known properties of Jacobi matrices \cite{Akh}, \cite{Ber68}.
 \end{proof}
    \begin{proposition}\label{carlemandeltaprime}
Let $D_X$  be the minimal Dirac operator in  $L^{2}(\cI, \C^{2})$
and let   $D_{X,\gB}$ be the GS-realization of $D_X.$ 

\item $(i)$ If  $|\cI|=+\infty$, then $D_{X,\gB}$ is self-adjoint.  

\item $(ii)$ If $|\cI| < \infty$  and
\begin{equation}\label{b}
\sum_{n=1}^\infty |\beta_n|\sqrt{d_nd_{n+1}}=+\infty\,,
\end{equation}  
then $D_{X,\gB}$ is self-adjoint.
%
%
    \end{proposition}
\begin{proof}
(i)  By Corollary \ref{th_delta_prime},
$\mathrm{n}_\pm(D_{X,\beta})=\mathrm{n}_\pm(B_{X,\beta}).$
Alongside  the  Jacobi  matrix  $B_{X,\gB}$ of the form  \eqref{IV.3.1_04} we
consider  Jacobi matrix $B_{X,\gB}'$ obtained from   $B_{X,\gB}$ by replacing its off-diagonal entries  by their modulus  (cf. with  construction of the matrix $B_{X,\gA}'$ of the form \eqref{IV.2.1_01'}). The matrices $B_{X,\gB}$  and $B_{X,\gB}'$  are unitarily equivalent.
Self-adjointness of the operator  $B_{X,\gB}'$  follows from the Carleman test. In fact, the proof coincides with  that  of Proposition \ref{cor_delta_carleman} since the off-diagonal entries  of the Jacobi matrices  $B_{X,\beta}'$ and $B_{X,\alpha}'$ coincide.  \par

(ii) The proof is similar to that of Proposition   \ref{interval}.
Applying the  Dennis-Wall test (see \cite[Chapter 1,  Problem 2]{Akh}) to the Jacobi matrix $B_{X,\beta}'$
yields  self-adjointness of $B_{X,\alpha}'$ provided that
\begin{equation}\label{d+b}
\sum_{n=1}^\infty |d_{n}+\beta_n|\sqrt{d_nd_{n+1}}=+\infty\,.
\end{equation}
Since $|\cI|<+\infty$, one has $\sum_{n=1}^\infty d_{n}\sqrt{d_nd_{n+1}}< 2|\cI|^2 < +\infty$. Thus, the series in \eqref{d+b} and the series  in \eqref{b}
diverge only simultaneously.
      \end{proof}

All previous results on spectral properties of the GS-operator $D_{X,\alpha}$ have their counterparts for the realization $D_{X,\beta}$. They can be proved directly in a much the same way but we prefer another way described as follows.

Alongside the GS-realization $D_{X,\alpha}$ we introduce another GS-realization $\widehat D_{X,\alpha}$ being the closure of the operator
  \begin{equation*}
\begin{split}
 \widehat D_{X,\alpha}^0 = & D\upharpoonright \dom( \widehat D_{X,\alpha}^0),\\ \dom( \widehat D_{X,\alpha}^0)
= &\Big\{f\in W^{1,2}_{\comp}(\cI \backslash X)\otimes
\C^{2}: f_{1}\in AC_{\loc}(\cI),\ f_{2}\in AC_{\loc}(\cI\backslash X);\\&
f_1(a+)=0\,,\quad f_{2}(x_{n}+)-f_{2}(x_{n}-)
=-\frac{i\alpha_{n}}{c}f_{1}(x_{n}),\,\,n\in\N\Big\},
\end{split}
  \end{equation*}
i.e. $\widehat D_{X,\alpha} = \overline{\widehat D_{X,\alpha}^0}.$ The following  statement  is immediate from the  previous definitions.
    \begin{proposition}\label{prop6.4}
Let $\alpha = c^2\beta$.  Then the  realizations $D_{X,\beta}$ and $- \widehat D_{X,\alpha}$ are unitarily equivalent. More precisely, the
   following identity holds
\begin{equation*}
U^{-1}D_{X,\beta}U = -  \widehat D_{X,\alpha}\,, \qquad  \dom( \widehat D_{X,\alpha}) = U^{-1}\dom(D_{X,\beta}),
\qquad\alpha = c^{2}\beta \,,
\end{equation*}
where  $U$  is the unitary operator,
$$
U:\  L^2(\cI)\otimes \C^2\to L^2(\cI)\otimes \C^2\,,
\qquad \ \  U:=1\otimes\sigma_{2}\,,
$$
and  $\sigma_{2}$ is one of the Pauli matrices given  by \eqref{pauli}.
     \end{proposition}
   \begin{proposition}\label{prop6.5}
Let $\alpha = \beta c^2$. Then the  GS-realizations $D_{X,\beta}$ and $D_{X,\alpha}$  are selfadjoint only simultaneously. Moreover, the spectrum $\sigma(D_{X,\beta})$ of $D_{X,\beta}$ is either discrete or purely singular if and only if so is the spectrum $\sigma(D_{X,\alpha})$ of $D_{X,\alpha}$.

Besides, the $ac$-parts of the operators $D_{X,\alpha}$ and $D_{X,\beta}$ are unitarily equivalent.
  \end{proposition}
\begin{proof}
Assume that  $D_{X,\alpha}$  is selfadjoint. Then its restriction
   \begin{equation*}
S:= D_{X,\alpha}\upharpoonright \dom(S), \quad \dom(S) := \dom( \widehat D_{X,\alpha})\cap
 \dom(D_{X,\alpha}) =
 \Big\{f\in \dom(D_{X,\alpha})\,:\,f_{1}(a+)=0\Big\},
 \end{equation*}
is  a closed symmetric operator with the deficiency indices
$n_{\pm}(S) = 1$.  Therefore, by the second von Neumann formula,
 $\dim(\dom( \widehat D_{X,\alpha})/\dom(S)) = 1$ and $ \widehat D_{X,\alpha}$ being a symmetric operator is  selfadjoint too.
Moreover, since $n_{\pm}(S) = 1$,  the resolvent difference
\begin{equation}\label{difference}
( \widehat D_{X,\alpha} - z)^{-1} - ( D_{X,\alpha}-z)^{-1}
\end{equation}
is  rank-one operator. Therefore the operators $ \widehat D_{X,\alpha}$ and $D_{X,\alpha}$
have either discrete spectrum or purely singular spectrum  only simultaneously.
Moreover, by the Kato-Rosenblum theorem their $ac$-parts are unitarily equivalent.
To complete the proof it remains to apply Proposition \ref{prop6.4}.
     \end{proof}
        \begin{remark}
Let  $\tau: \dom(D_{X,\alpha}) \to \C$ be the trace mapping given by
$\tau(f):=f_1(a+)$. Clearly,  it is  continuous and surjective.
Since  $\dom( \widehat D_{X,\alpha})\cap \dom(D_{X,\alpha})=\dom(D_{X,\alpha})\cap ker (\tau)$  is dense in $L^{2}(\cI)$,
the operator  $ \widehat D_{X,\alpha}$ can be treated as a singular perturbation of $D_{X,\alpha}$ in the sense of \cite{pos}.
This leads to an alternative  proof of Proposition \ref{prop6.5}.
       \end{remark}

Combining Propositions  \ref{cor_delta_carleman} and \eqref{interval}  with  Proposition \ref{prop6.4}
yields an alternative proof of  Proposition  \ref{carlemandeltaprime}.
Moreover,  using  Proposition \ref{prop6.4} one can
obtain the counterparts of the results of Sections 5.4 and 5.5
on spectral properties of $D_{X,\alpha}$.

We demonstrate this possibility by stating  the following result on discreteness of  $D_{X,\beta}$.

%
%
   \begin{proposition}
Let  $X = \{x_n\}_1^\infty (\subset \R_+)$,  $\beta=\{\beta_j\}^{\infty}_1\subset{\mathbb R}$.
Assume that $\lim_{n\to\infty}\gd_n=0$  and
%
    \begin{equation}
\lim_{n\to\infty}\frac{|\beta_n|}{\gd_n}=\infty\qquad
\text{and}\qquad \lim_{n\to\infty}\frac{1}{c\beta_n} > -
\frac{1}{4}.
  \end{equation}
Then the (self-adjoint) $GS$-operator $D_{X,\beta}$ on the half-line $\R_+$ has discrete spectrum.
     \end{proposition}
   \begin{proof}
The statement is immediate by combining  Proposition \ref{prop_IV.2.4_01} with Proposition \ref{prop6.4}.
   \end{proof}

\subsection{Non-relativistic limit of Gesztesy-\v{S}eba operators}
Let, as in Section 3.2,  $X = \{x_n\}_{n=1}^{\infty} (\subset \R_+)$ be a discrete set and $\alpha = \{\alpha_n\}_1^\infty,$ $\beta = \{\beta_n\}_1^\infty \subset \R.$

In this section we consider the non-relativistic limits of Gesztesy-\v{S}eba operators $D_{X,\alpha}^c := D_{X,\alpha}$ and $D_{X,\beta}^c := D_{X,\beta}$
using their parameterizations with respect to the boundary triplet $\Pi = \{\cH, \Gamma_0, \Gamma_1\}$  constructed in Theorem \ref{th_bt_2}.
By Propositions  \ref{prop_IV.2.1_01}  and  \ref{prop_IV.3.1_01},
  \begin{equation}\label{IV.3.1_03NRLimit}
D_{X,\gA}^c = D_{B_{X,\gA}^c}:=D_{X}^*\upharpoonright \ker(\Gamma_1 - B_{X,\gA}^c\Gamma_0)\quad \text{and}\quad
D_{X,\gB}^c = D_{B_{X,\gB}^c}:= D_{X}^*\upharpoonright \ker(\Gamma_1 - B_{X,\gB}^c\Gamma_0),
  \end{equation}
where  $B_{X,\gA}^c$ and $B_{X,\gB}^c$ are the Jacobi matrices given by \eqref{IV.2.1_01} and \eqref{IV.3.1_04}, respectively.
Here we indicate explicitly the dependence of all operators on the parameter $c$ by writing
$D^c_{X,\alpha}$, $D^c_{X,\beta}$, $B^c_{X,\alpha}$ and $B^c_{X,\beta}$ in place of $D_{X,\alpha}$, $D_{X,\beta}$, $B_{X,\alpha}$ and $B_{X,\beta}$, respectively.

Following \cite{Alb_Ges_88} we recall the definitions of the operators
which describe Schr\"{o}dinger operators  with $\delta$ and $\delta'$ interactions, respectively (cf. also \cite[Sections 5, 6]{KM}).
Let
   \begin{equation}\label{Hdelta}
\begin{split}
\rH_{X,\gA}^0 = & -\frac{d^2}{dx^2}\upharpoonright \dom(\rH_{X,\gA}^0),\\
\dom(\rH_{X,\gA}^0)=&\Big\{f\in W^{2,2}_{\comp}(\cI \backslash X): f\in AC_{\loc}(\cI),\ f'\in AC_{\loc}(\cI\backslash X);\\&
f'(a+)=0\,,\quad f'(x_{n}+)-f'(x_{n}-)
=\alpha_{n}f(x_{n}),\,\,n\in\N\Big\},
\end{split}
  \end{equation}
  \begin{equation}
\begin{split}\label{Hdeltap}
\rH_{X,\beta}^0 = & -\frac{d^2}{dx^2}\upharpoonright \dom(\rH_{X,\gB}^0), \\
\dom(\rH_{X,\beta}^0)=& \Big\{f\in W^{2,2}_{\comp}(\cI \backslash X)
\,:\, f\in AC_{\loc}(\cI\backslash X),\ f'\in AC_{\loc}(\cI);\\&
f'(a+)=0\,,\quad
f(x_{n}+)-f(x_{n}-)=\beta_{n}f'(x_{n}),\,\,n\in\N\Big\}.
\end{split}
\end{equation}
Then the operators  $\rH_{X,\gA}$ and $\rH_{X,\beta}$ are defined to be the closures of $\rH_{X,\gA}^0$ and  $\rH_{X,\beta}^0$, respectively.

If $\cI =\R_+$, the operator $H_{X,\beta}$ is  selfadjoint in $L^2(\R_+)$ for any $\beta$ (\cite[Theorem 4.7]{BSW95}, \cite[Theorem 6.3]{KM},  although  $\rH_{X,\alpha}$  is only symmetric with equal deficiency indices $n_{+}(\rH_{X,\alpha}) = n_{-}(\rH_{X,\alpha}) \le 1$ (see \cite{BSW95}).
Moreover,  $\rH_{X,\alpha}$ may have non-trivial deficiency indices $n_{\pm}(\rH_{X,\alpha})=1$ (see \cite{Chr_Sto_94}, \cite[Section 5.2]{KM}). However, the operator $\rH_{X,\alpha}$ is selfadjoint,  $\rH_{X,\alpha} = \rH^*_{X,\alpha}$, provided that it is semibounded below
(\cite[Theorem 1]{AKM_10}, see also the recent publication \cite{HryMyk12} for another proof).

It is shown in \cite{KM} that certain  spectral properties of  $\rH_{X,\alpha}$ closely correlate  with the corresponding
properties of the following Jacobi matrix
\begin{equation}\label{IV.2.1_01A}
B_{X,\gA}(\rH) := \left(%
\begin{array}{cccccc}
  0& -\gd_1^{-2} & 0 & 0& 0  &  \dots\\
  -\gd_1^{-2} & -\gd_1^{-2}& \gd_1^{-3/2}\gd_2^{-1/2}& 0 & 0&  \dots\\
  0 & \gd_1^{-3/2}\gd_2^{-1/2} & \alpha_1\gd_2^{-1} & -\gd_2^{-2} & 0&   \dots\\
  0 & 0 & -\gd_2^{-2} & -\gd_2^{-2} & \gd_2^{-3/2}\gd_3^{-1/2}&  \dots\\
  0 & 0 & 0 & \gd_2^{-3/2}\gd_3^{-1/2} & \alpha_2\gd_3^{-1}&   \dots\\
\dots& \dots&\dots&\dots&\dots&\dots\\
 \end{array}%
\right).
\end{equation}
As usual we  identify the Jacobi matrix $B_{X,\gA}(\rH)$ with (closed) minimal symmetric operator
 associated with it and denote it by the same letter  (cf. \eqref{IV.2.1_02}).
Recall that $B_{X,\gA}(\rH)$ has equal deficiency indices and $\mathrm{n}_{\pm}(B_{X,\gA}(\rH)) \leq 1$.

The Jacobi matrix  $B_{X,\alpha}(\rH)$ coincides with the matrix  $B^\infty_{X,\alpha}$ given by  \eqref{IV.2.1_01}   with  $\nu(x)\equiv 1$, i.e. with $c=\infty$. Note however that in the case $\sum_{k\in \N}d_k =\infty$ the matrix  $B^c_{X,\alpha},\ c<\infty,$ is  always  selfadjoint though
the matrix $B_{X\alpha}(\rH)$ might be only symmetric (see  \cite[Section 5.2]{KM}).

Similarly, according to \cite{KM}  certain  spectral properties of  $\rH_{X,\beta}$ closely correlate  with the corresponding properties of the  Jacobi matrix $B_{X,\beta}(\rH)$ given by

\begin{equation}\label{IV.3.1_04H}
B_{X,\gB}(H) := \left(\begin{array}{cccccc}
0 & -\gd_1^{-2} & 0& 0& 0 & \dots\\
-\gd_1^{-2} & -(\beta_1+\gd_1)\gd_1^{-3} & \gd_1^{-3/2}\gd_2^{-1/2}& 0& 0 & \dots\\
0 & \gd_1^{-3/2}\gd_2^{-1/2} & 0& -\gd_2^{-2}& 0 & \dots\\
0 & 0 & -\gd_2^{-2}& -(\beta_2+\gd_2)\gd_2^{-3}& \gd_2^{-3/2}\gd_3^{-1/2} & \dots\\
0 & 0 & 0& \gd_2^{-3/2}\gd_3^{-1/2} & 0 & \dots\\
\dots & \dots & \dots& \dots& \dots & \dots
\end{array}\right).
\end{equation}

Denote also by $B_{X,\beta}(\rH)$ the Jacobi matrix defined by \eqref{IV.3.1_04} with $\nu(x)\equiv 1$, i.e. with $c=\infty$.

The Jacobi matrices $B_{X,\alpha}(\rH)$ and $B_{X,\beta}(\rH)$ first appeared for the parametrization of Schr\"{o}dinger operators  $\rH_{X,\gA}$ and  $\rH_{X,\gB}$ with $\delta$- and $\delta'$-interactions, respectively  (cf. \cite[Proposition 5.1]{KM} and \cite[Proposition 6.1]{KM}).
Let us recall these statements:
    \begin{proposition}\label{prop_IV.2.1_01H}
Let $\Pi=\{\cH,\Gamma_0,\Gamma_1\}$ be the boundary triplet
for $\rH_{\min}^*$ constructed in Theorem \ref{th_bt_2.H}. Then the boundary operators corresponding to the realizations $\rH_{X,\gA}$ and $\rH_{X,\gB}$ coincide with
the minimal Jacobi operators  $B_{X,\gA}:= B_{X,\gA}(\rH)$ and  $B_{X,\gB} := B_{X,\gB}(\rH)$, respectively, i.e.
%
%
     \begin{eqnarray*}
\rH_{X,\gA} = \rH_{B_{X,\gA}}=\rH_{\min}^* \upharpoonright \dom(\rH_{B_{X,\gA}}),\qquad \dom(\rH_{B_{X,\gA}})=
\{f\in W^{2,2}(\cI\setminus X):\Gamma_1f = B_{X,\gA}\Gamma_0f\},    \\
   \rH_{X,\gB}=\rH_{B_{X,\gB}}=\rH_{\min}^* \upharpoonright \dom(\rH_{B_{X,\gB}}),\qquad \dom(\rH_{B_{X,\gB}})=
   \{f\in W^{2,2}(\cI\setminus X):\Gamma_1f = B_{X,\gB}\Gamma_0f\}.
    \end{eqnarray*}
Moreover,  $n_{\pm}(\rH_{X, \alpha}) = n_{\pm}(B_{X,\alpha}(\rH)) \le 1$ and $n_{\pm}(\rH_{X, \beta}) = n_{\pm}(B_{X,\beta}(\rH)) \le 1.$
     \end{proposition}
Then the following results on the non relativistic limit hold:
  \begin{theorem}\label{nonrelGS}\
Let $X= \{x_n\}_1^{\infty} (\subset \R_+)$ be a discrete set and $\alpha = \{\alpha_n\}_1^\infty,$
$\beta = \{\beta_n\}_1^\infty \subset \R.$
 \par\noindent

\item $(i)$ Assume that $\cI =\R_+$ and $\rH_{X, \alpha}$ is selfadjoint. Then
  \begin{equation}\label{1.23GS}
s-\lim_{c\to
  +\infty}\left(D_{X,\alpha}^c - (z + {c^{2}}/{2})\right)^{-1}=
(\rH_{X, \alpha} -z)^{-1}\bigotimes \left(\begin{array}{cc} 1 & 0 \\0 &
0\end{array}\right)\,.
  \end{equation}
In particular,  \eqref{1.23GS} holds provided that $\rH_{X, \alpha}$ is semibounded below.

\item $(ii)$ Assume that   $\cI =\R_+$. Then  the operators  $D_{X,\beta}^c,\ c<\infty,$   and  $\rH_{X, \beta}$  are selfadjoint and the following
relation holds
  \begin{equation}\label{1.23GSbeta2}
s-\lim_{c\to
  +\infty}\left(D_{X,\beta}^c - (z + {c^{2}}/{2})\right)^{-1}=
(\rH_{X, \beta} - z)^{-1}\bigotimes \left(\begin{array}{cc} 1 & 0 \\0 &
0\end{array}\right)\,.
  \end{equation}

\item $(iii)$  Assume that  $\cI = (0, b)$ with $b<\infty$. Assume also that
  \begin{equation}\label{1.23GSbetaSelfad}
\sum_{n=1}^\infty |\beta_n| \sqrt{d_nd_{n+1}}=+\infty \qquad\text{and}\qquad
\sum_{n=1}^\infty \bigl(\gd_{n+1} \bigl|\sum_{i=1}^n (\beta_i+\gd_i)\bigr|^2\bigr) = \infty.
     \end{equation}\label{1.23GSbeta}
Then relation  \eqref{1.23GSbeta2} holds.
   \end{theorem}
 \begin{proof} (i) Firstly, by  Theorem   \ref{1}, the operator $D^c_{X,\alpha},\ c<\infty,$ is   selfadjoint for any $\alpha\subset \R$ and any $c>0$ since $\cI=\R_+$.
Therefore, by  Proposition  \ref{prop_IV.2.1_01},  $B^c_{X,\alpha},\ c<\infty,$ is selfadjoint too. Further,  let $B_{X,\alpha}(\rH)$ be the minimal Jacobi operator associated with the matrix \eqref{IV.2.1_01A}. By Proposition \ref{prop_IV.2.1_01H}, $\rH_{X,\alpha} = \rH_{X,B_{X,\alpha}}$. Moreover,  by Proposition  \ref{prop_IV.2.1_01H},  $B_{X,\alpha}(\rH)= B_{X,\alpha}(\rH)^*$ since  $\rH_{X,\alpha} = \rH_{X,\alpha}^*$.
 Combining  \eqref{IV.2.1_01}, \eqref{IV.2.1_01NU} and  \eqref{IV.2.1_01A} with  the obvious
relation $\lim_{c\to\infty} \nu(cx) =1$ we get
  \begin{equation}\label{4.71}
\lim_{c\to+\infty} B^c_{X,\alpha}h = B_{X,\alpha}(\rH)h  \qquad  \text{for\ \ all}   \qquad
h\in l^2_{0}(\N,\C^2).
    \end{equation}
Note also that,  by the definition of a minimal Jacobi operator, $l^2_{0}(\N,\C^2)$ is a core for both (selfadjoint) Jacobi operators   $B_{X,\alpha}(\rH)$ and  $B^c_{X,\alpha}$, $c\in (0, \infty)$. Applying  Theorem \ref{2.5}  we arrive at the relation \eqref{1.23GS}.
To complete  the proof it remains to note that $\rH_{X,\alpha}$ is selfadjoint,  $\rH_{X,\alpha} = \rH^*_{X,\alpha}$, provided that it is semibounded below (see  \cite[Theorem 1]{AKM_10}).

(ii)  Let $\cI =\R_+.$   By \cite{BSW95}  (see also \cite[Theorem 6.3(i)]{KM}), $\rH_{X,\beta} = \rH_{X,\beta}^*$. Combining this relation with  Proposition  \ref{prop_IV.2.1_01},  yields $B_{X,\beta}(\rH)= B_{X,\beta}(\rH)^*.$

On the other hand, by Proposition \ref{carlemandeltaprime}(i),   $D_{X,\beta}^c$ is selfadjoint too, $D_{X,\beta}^c = (D_{X,\beta}^c)^*,\ c<\infty$.
Further,  by Proposition  \ref{prop_IV.3.1_01},
$D_{X,\gB}^c = D_{B_{X,\gB}}^c = D_{X}^* \upharpoonright \ker (\Gamma_1 - B_{X,\gB}^c\Gamma_0).$
Therefore, by Proposition   \ref{s.a.Dirac}(ii),  $B_{X,\beta}^c$ is selfadjoint too,  $B_{X,\beta}^c= (B_{X,\beta}^c)^*.$
Note that alongside \eqref{4.71}  we obtain from   \eqref{IV.3.1_04}  and  \eqref{IV.3.1_04H} a  similar relation
$\lim_{c\to+\infty} B^c_{X,\beta}h = B_{X,\beta}(\rH)h,$\   $h\in l^2_{0}(\N,\C^2).$

Again, by the definition of the  minimal Jacobi operators $B^c_{X,\beta}$ and  $B_{X,\beta}(\rH)$,  $l^2_{0}(\N,\C^2)$ is a core for both of them.
To arrive at \eqref{1.23GSbeta2} it remains to apply Theorem \ref{2.5}.

(iii) Let $|\cI| < \infty.$  By Proposition \ref{carlemandeltaprime}(ii),
selfadjointness of the  $GS$-operators $D_{X,\beta}^c$, \  $c\in (0,\infty),$
is implied by the first of conditions  \eqref{1.23GSbetaSelfad}. Combining  this fact with Propositions \ref{prop_IV.3.1_01} and \ref{s.a.Dirac}(ii),  we get $B_{X,\beta}^c = (B_{X,\beta}^c)^*,\  c<\infty.$
Further, by \cite[Theorem 6.3(ii)]{KM}, the second of conditions  \eqref{1.23GSbetaSelfad}  yields $\rH_{X,\beta} = \rH_{X,\beta}^*$.  In turn, by Proposition  \ref{prop_IV.2.1_01H},  $B_{X,\beta}(\rH)= B_{X,\beta}(\rH)^*.$
  The  proof is completed in much the same way as  the proof of statement (i).
   \end{proof}
 \begin{remark}
\item $(i)$  Note that condition $H_{X,\alpha}=H^*_{X,\alpha}$ hence the conclusion  \eqref{1.23GS}  of Theorem \ref{nonrelGS}  is satisfied  for any sequence $\alpha=\{\alpha_n\}^{\infty}_1\subset \R$  provided that  $\sum_j d^2_j=\infty$ \cite[Proposition 5.7]{KM} (see recent publications \cite{IsmKost10},  \cite{MirzSaf11}  for other proofs and generalizations). In particular, the non-relativistic limit  \eqref{1.23GS}  is valid whenever  $d_*(X)>0$. Under the  latter assumption both statements \eqref{1.23GS} and  \eqref{1.23GSbeta2} were  originally obtained by Gesztesy and {\v S}eba \cite{GS} (see also \cite[Appendix J]{Alb_Ges_88}). In this case the uniform  convergence in \eqref{1.23GS}, \eqref{1.23GSbeta2} holds.

\item $(ii)$ Clearly, conditions \eqref{1.23GSbetaSelfad}  can be replaced by the assumptions of selfadjointness of operators  $D_{X,\beta}^c$, $c>0,$  and $\rH_{X,\beta}$.
 \end{remark}
\vskip10pt\noindent
\section*{Acknowlegments}
\addcontentsline{toc}{section}{Acknowlegments}%
\markboth{Acknowlegments}{Acknowlegments}%

We  are indebted to A. Kostenko for useful discussion. We are also indebted  to the anonymous referee for useful remarks.

A part of this work was done when  one of the authors (M. Malamud) was visiting
the Department of Physics and Mathematics at Insubria University in November-December 2010. His  stay was supported by a grant from Cariplo Foundation - Centro Volta.
\addcontentsline{toc}{section}{References}%


\begin{thebibliography}{00}


\bibitem{Akh}
\emph{Akhiezer N.I.} The classical moment problem and some related questions in analysis,  Oliver and Boyd Ltd, Edinburgh, London, 1965.

\bibitem{Akh_Glz}
\emph{Akhiezer N.I., Glazman I.M.} Theory of Linear Operators in Hilbert Spaces,  Nauka, 1978.


\bibitem{Alb_Ges_88}
\emph{Albeverio S., Gesztesy F., Hoegh-Krohn R., Holden H.}
Solvable Models in Quantum Mechanics, Sec. Edition, AMS Chelsea Publ., 2005.

\bibitem{AKM_10}
\emph{Albeverio S., Kostenko A., Malamud M.M.}
{ Spectral theory of semi-bounded  Sturm-Liouville  operators with local interactions on a discrete set}, J. Math. Phys., \textbf{51}, 102102, (2010), 24 pp.


\bibitem{Alb_Kur_00}
\emph{Albeverio S., Kurasov P.} Singular Perturbations of Differential Operators and  Schr\"{o}dinger Type Operators, Cambridge Univ. Press, 2000.

\bibitem{Alb_Niz_Tar_04}
\emph{Albeverio S., Nizhnik L., Tarasov V.} Inverse scattering problem for a Dirac system with nonstationary point interactions, \emph{Inverse Problems}, \textbf{20}, no. 3, (2004), 799-813.

\bibitem{AvrGro}
\emph{Avron J.E., Grosman A.} The relativistic Kronig-Penney Hamiltonian, \emph{Phys. Let.} \textbf{56 A}, (1976), 5-57.

\bibitem{Ben}
\emph{Benvegn\`u S.} Relativistic point interaction with Coulomb potential in one dimension, \emph{J. Math. Phys.}, \textbf{38}, no. 2, (1997), 556-570.


\bibitem{BD}
\emph{Benvegn\`u S., D\c{a}browski L.} Relativistic point interaction,
\emph{Lett. in Math. Physics}, \textbf{30}, (1994), 159-167.


\bibitem{Ber68}
\emph{Berezanskii Ju.~M.} Expansions in eigenfunctions of selfadjoint operators,
\newblock Translated from Russian by R. Bolstein, J. M. Danskin, J. Rovnyak
and L. Shulman. Translations of Mathematical Monographs, Vol. 17., American
Mathematical Society, Providence, R.I., 1968.



\bibitem{BraMalNei02}
\emph{Brasche J. F., Malamud M.M., Neidhardt H.} Weyl function and spectral properties of self-adjoint extensions, \emph{Integr. Eq. Oper. Theory}, \textbf{43}, (2002), 264-289.

%
\bibitem{BGW08}
\emph{B.\ M.\ Brown, G.\ Grubb, and I.\ G.\ Wood}
{$M$-functions  for closed extensions of adjoint pairs of operators with applications to elliptic
boundary problems}, \emph{Math. Nachr.}, \textbf{282}, (2009) 314-347.
%

\bibitem{Brze_03}
\emph{Brze\'{z}niak Z., Jefferies B.} Renormalization of Coulomb interactions for the 1D Dirac equation, \emph{J. Math. Phys.}, \textbf{44}, no. 4, (2003), 1638-1659.


\bibitem{BGP07}
\emph{Bruening J., Geyler V., Pankrashkin K.} Spectra of self-adjoint extensions and applications to solvable  Schr\"odinger operators, \emph{Rev. Math. Phys.}, \textbf{20}, (2008), 1-70.

\bibitem{BSW95}
\emph{Buschmann D., Stolz G., Weidmann J.} One-dimensional Schr\"odinger operators with local point interactions, \emph{J. Reine Angew. Math.}, \textbf{467}, (1995), 169-186.

\bibitem{Chi62}
\emph{Chihara T.} Chain sequences and orthogonal polynomials, \emph{Trans. AMS} \textbf{104}, (1962), 1--16.


\bibitem{DavSte}
\emph{Davison S.G., Steslicka M.} Relativistic treatment of localized states. A review. \emph{Intern. J. Quant. Chem.}, \textbf{4}, (1971), 45-453.

\bibitem{DelAnt2011}
\emph{Dell'Antonio G.} Aspetti Matematici della Meccanica Quantistica, Bibliopolis, 2011.

\bibitem{DHMS06}
\emph{Derkach V.A., Hassi S., Malamud M.M. and H.S.V. de Snoo.},
Boundary relations and their Weyl families, \emph{Trans. Amer.
Math. Soc.}, \textbf{358}, 12, (2006), 5351-5400.

\bibitem{DM91}
\emph{Derkach V.A., Malamud, M.M.} Generalised Resolvents and
the boundary value problems for Hermitian Operators with gaps,
\emph{J. Funct. Anal.}, \textbf{95}, (1991), 1-95.

\bibitem{DerMal92}
\emph{Derkach V.A., Malamud, M.M.}  Characteristic functions  of almost solvable extensions of Hermitian operators, \emph{Ukrainian Math. J.}, \textbf{44} (1992), 379-401.

\bibitem{DM95}
\emph{Derkach V.A., Malamud M.M.}, Generalised Resolvents and
the boundary value problems for Hermitian Operators with gaps,
\emph{J. Math. Sci.}, \textbf{73}, No 2, (1995), 141-242.

\bibitem{Dit_Ex_Seb89}
\emph{Dittrich J., Exner P., \v{S}eba P.} Dirac Hamiltonian with contact interaction on a sphere. Schr\"odinger  operators, standard and nonstandard (Dubna, 1988), 190-204, World Sci. Publ., Teaneck, NJ, 1989.

\bibitem{Dit_Ex_Seb8}
\emph{Dittrich J., Exner P., \v{S}eba P.}  Dirac operators with a spherically symmetric $\delta$-shell interaction, \emph{J. Math. Phys.}, \textbf{30}, 12, (1989), 2875-2882.

\bibitem{Dit_Ex_Seb}
\emph{Dittrich J., Exner P., \v{S}eba P.} Dirac Hamiltonian with Coulomb potential and spherically symmetric shell contact interaction. With a Russian summary, Communications of the Joint Institute for Nuclear Research. Dubna, E2-90-531. Joint Inst. Nuclear Res., Dubna, \emph{Integral Equations and Operator Theory}, \textbf{46}, Birkhauser Verlag, Basel, (1990),  pp. 209-215.

\bibitem{Dit_Ex_Seb90}
\emph{Dittrich J., Exner P., \v{S}eba P.} Dirac Hamiltonian with Coulomb potential and contact interaction on a sphere. Order, disorder and chaos in quantum systems (Dubna, 1989), 209-219, \emph{Oper. Theory Adv. Appl.}, \textbf{46}, BirkhЉuser, Basel, 1990.

\bibitem{Dit_Ex_Seb97}
\emph{Dittrich J., Exner P., \v{S}eba P.} Dirac Hamiltonian with Coulomb potential and spherically symmetric shell contact interaction, \emph{J. Math. Phys.}, \textbf{33}, 6, (1997),  2207-2214.

\bibitem{Exn_04}
\emph{Exner P.} Seize ans apr$\grave{\rm{e}}$s, \emph{Appendix K to} "Solvable Models in Quantum Mechanics" by Albeverio S., Gesztesy F., Hoegh-Krohn R., Holden H., Sec. Edition, AMS Chelsea Publ., 2004.

\bibitem{Ges_Hol_87}
\emph{Gesztesy F., Holden H.} A new class of solvable models in
quantum mechanics describing point interactions on the line,
\emph{J. Phys. A: Math. Gen.}, \textbf{20}, (1987), 5157-5177.


\bibitem{GS}
\emph{Gesztesy F., \v{S}eba P.} New analytically solvable models of relativistic point interactions, \emph{Lett. Math. Phys.}, \textbf{13}, (1987), 345-358.


\bibitem{Gor84}
\emph{Gorbachuk V.I., Gorbachuk M.L.} Boundary Value Problems for Operator Differential Equations, Mathematics and its Applications (Soviet Series), \textbf{48}, Kluwer Academic Publishers Group, Dordrecht, 1991.

\bibitem{Grosc}
\emph{Grosche C.} Boundary conditions in path integrals from point interactions for the path integral of the one-dimensional Dirac particle, \emph{J. Phys. A}, \textbf{32}, 9, (1999), 1675-1690.

\bibitem{Gro}
\emph{Grossmann A., Hoegh-Krohn R., Mebkhout M.} The one-particle theory of periodic point interactions, \emph{J. Math. Phys.}, \textbf{21}, (1980), 2376-2385.

\bibitem{Gru68}
\emph{Grubb G.} A characterization of the non-local boundary value problems associated with an elliptic operator, \emph{Ann. Scuola Norm. Sup. Pisa}, \textbf{22}, (1968), 425-513.

\bibitem{Gru08}
\emph{Grubb G.} Distributions and  Operators, Springer-Verlag, New York, 2009.

\bibitem{Gumbs}
\emph{Gumbs G.} Relativistic scattering states for a finite chain with $\delta$-function potentials of arbitrary position and strength, \emph{Phys. Rev. A}, \textbf{32}, (1985), 1208-1210.

\bibitem{Hou_01}
\emph{Hounkonnou M. N., Avossevou G. Y. H.} Spectral and resonance properties of $\delta$ and $\delta'$-type interactions in relativistic quantum mechanics, \emph{J. Math. Phys.}, \textbf{42}, 1, (2001), 30-51.


\bibitem{HryMyk12}
\emph{Hryniv R.O., Mykytyuk Ya. V.},
Self-adjointness of  Schr\"odinger operators with singular potentials,
\emph{Methods of Func. Anal. and Topology}, \textbf{18}, 2, (2012), 152-159.

\bibitem{Hugh}
\emph{Hughes R. J.} Relativistic Kronig-Penney-type Hamiltonians, \emph{Integ. Eq. Op. Th.}, \textbf{31}, 4, (1998), 436-448.

\bibitem{IsmKost10}
\emph{Ismagilov R.S.,  Kostjuchenko A.G.} On asymptotic of the spectrum of
Sturm-Liouville operator with point interaction, Funct. Analysis and
Appl., \textbf{44}, 4, (2010), 14-20.




\bibitem{Kato66}
\emph{Kato T.} Perturbation theory for linear operators,
Springer-Verlag, Berlin-Heidelberg, New York, 1966.

\bibitem{Koc_79}
\emph{Kochubei A.N.} Symmetric operators and nonclassical spectral problems, \emph{Math. Notes}, \textbf{25}, 3, (1979), 425-434.

\bibitem{Koc_89}
\emph{Kochubei A.N.} One-dimensional point interactions, \emph{Ukrain. Math. J.}, \textbf{41}, (1989), 1391-1395.

\bibitem{KM}
\emph{Kostenko A.S., Malamud M.M.} 1-D
Schr\"{o}dinger operators with local point interactions on a discrete set,
 \emph{J. Differential Equations}, \textbf{249}, (2010), 253-304.

 \bibitem{KosMal12}
\emph{Kostenko A., Malamud M.}  1--D  Schr\"odinger operators with
local point interactions: a review, in H. Holden et
al. (eds), Spectral Analysis, Differential Equations and Mathematical Physics: A Festschrift in Honor of Fritz Gesztesy's 60th Birthday, Proceedings of Symposia in Pure Mathematics, vol. 87, Amer. Math. Soc. 2013 (in press)
\bibitem{KosMir99}
\emph{Kostyuchenko A.G., Mirzoev K.A. } Generalized Jacobi matrices and deficiency numbers of ordinary differential operators with polynomial coefficients,
 \emph{Funct. Anal. Appl.}, \textbf{33}, (1999), 30-45.

\bibitem{KosMir01}
\emph{Kostyuchenko A.G., Mirzoev K.A. } Complete indefiniteness tests for Jacobi matrices with matrix entries,
 \emph{Funct. Anal. Appl.}, \textbf{35}, (2001), 265-269.


\bibitem{Kre47}
\emph{Krein M.\ G.} {Theory of selfadjoint extensions of
semibounded operators}, \emph{Mat. Sbornik} \textbf{20}, (1947), 431-495 (in Russian).
%

\bibitem{KL71}
\emph{Krein M.G., Langer H.} On defect subspaces and generalized resolvents of a Hermitian operator in a space $\Pi_\varkappa$, \emph{Funct. Anal. Appl.}, \textbf{5/6}, (1971/1972), 136-146, 217-228.

\bibitem{Kro_Pen}
\emph{Kronig R. de L., Penney W. G.} Quantum mechanics of electrons in crystal lattices, \emph{Proc. Roy. Soc. (London)}, \textbf{130A}, (1931), 499-513.

\bibitem{Lapidus}
\emph{Lapidus, I.R.} Relativistic one-dimensional hydrogen atom. \emph{Amer. J.Phys.}, {\bf 51}, (1983), 1036-1038.

\bibitem{LevSar88}
\emph{Levitan B.M. and Sargsyan I.S.} Operators of Sturm-Lioville and Dirac Operators, Kluver, 1990.

\bibitem{Lot11}
\emph{Lotoreichik V.} Singular continuous spectra of half--line Schr\"odinger operators with point interactions on a sparse set, \emph{Opus. Math.}, \textbf{31}, (2011), 615-628.

\bibitem{Mal10}
\emph{Malamud M. M.} Spectral theory  of elliptic operators in
exterior domains, \emph{Russ. J. Math. Phys.}, \textbf{17},
(2010), 97-126.

\bibitem{MM03}
\emph{Malamud M.M. and  Malamud S.M.}
 Spectral theory of operator measures in {H}ilbert space,
\emph{St. Petersburg Math. J.}, \textbf{15}, (2004),  323-373.

\bibitem{MalNei11}
\emph{Malamud M.M., Neidhardt H.} On the unitary equivalence of
absolutely continuous parts of self-adjoint extensions, \emph{J.
Funct. Anal.}, \textbf{260}, 3, (2011), 613-638 (arXiv:0907.0650v1 [math-ph]).

\bibitem{MN2012}
\emph{Malamud M.M., Neidhardt H.} Sturm-Liouville boundary value
problems with operator potentials and  unitary equivalence, \emph{J. Differential Equations},  \textbf{252}, (2012), 5875-5922 (arXiv:0907.0650v1 [math-ph]).

\bibitem{MalSch12}
\emph{Malamud M.M., Schm\"udgen K.} Spectral theory of Schr\"odinger operators with infinitely many point interactions and radial positive definite functions, \emph
{J. Functional Analysis}, \textbf {263}, 10, (2012), 3144-3194.


\bibitem{Mih_94a}
\emph{Mikhailets V.A.} One-dimensional Schr\"odinger operator with point interactions, \emph{Doklady Mathematics}, \textbf{335}, 4, (1994), 421-423.

\bibitem{Mih_96}
\emph{Mikhailets V.A.} A structure of the continuous spectrum of a one-dimensional Schr\"odinger operator with $\delta$-interactions, \emph{Funct. Anal. Appl.}, \textbf{30}, (1996), 144-146.

\bibitem{Min_86}
\emph{Minami N.} Schr\"odinger operator with potential which is the derivative of a temporally homogeneous Levy process, \emph{in} "Probability Theory and Mathematical Sciences", pp. 298-304, Proceedings, Kyoto, 1986, Lect. Notes in Math., \textbf{1299}, Springer, Berlin, 1988.

\bibitem{MirzSaf11}
\emph{Mirzoev K.A.,  Safonova T.A.} Singular Sturm-Liouville operators with
potential distribution in space of vector functions, Dokl Russian
Academy., \textbf{441}, 2,  (2011), 165-168.


\bibitem{Pos01}
\emph{Posilicano A.} A Krein-like Formula for Singular Perturbations of Self-Adjoint Operators and Applications, \emph{J. Funct. Anal.}, \textbf{183}, (2001), 109-147.

\bibitem{pos}
\emph{Posilicano A.} Boundary Triples and Weyl Functions for Singular Perturbations of Self-Adjoint Operators, \emph{Methods Funct. Anal. Topology}, \textbf{10}, (2004), 57-63.

\bibitem{RedSim78}
\emph{Reed, M., Simon, B.}  Methods of Modern Mathematical Physics, IV: Analysis of Operators, Academic Press, New York, 1978.

\bibitem{RedSim79}
\emph{Reed, M., Simon, B.}  Methods of Modern Mathematical Physics, III: Scattering Theory, Academic Press, New York, 1979.

\bibitem{Sch2012}
\emph{Schm\"udgen, K.} Unbounded Self-Adjoint Operators on Hilbert Space, Springer-Verlag, New York, 2012.


\bibitem{Chr_Sto_94}
\emph{Shubin C., Stolz G.} Spectral theory of
one-dimentional Schr\"{o}dinger operators with point interactions,
\emph{J. Math. Anal. Appl.}, \textbf{184}, (1994), 491-516.



\bibitem{Sha_02}
\emph{Shabani J., Vyabandi A.} Exactly solvable models of relativistic $\delta$-sphere interactions in quantum mechanics, \emph{J. Math. Phys.}, \textbf{43},  no. 12, (2002), 6064-6084.

\bibitem{SimSpe89}
\emph{Simon B.,  Spencer Th.}  Trace class perturbations and the absence of absolutely continuous spectra,  \emph{Commun. Math. Physics}, \textbf{125}, (1989), 113-125.


\bibitem{Tes_98}
\emph{Teschl G.}, Jacobi Operators and Completely Integrable Nonlinear Lattices, \emph{Math. Surveys Monographs}, \textbf{72}, AMS, 2000.

\bibitem{Tha92}
\emph{Thaller, B.}, The Dirac Equation, Texts and Monographs in Physics, Springer, 1992.

\bibitem{Yoshi}
\emph{Yoshitomi, K.}, Dirac operators with periodic $\delta$-interactions: spectral gaps and inhomogeneous Diophantine approximation, \emph{Michigan Math. J.}, \textbf{58}, (2009),
363-384.
\end{thebibliography}
\end{document}